\newif\ifsubmission

\ifsubmission
\documentclass{llncs}
\else
\documentclass[11pt]{article}

\usepackage[T1]{fontenc}
\usepackage{fullpage}
\fi

\usepackage{cmap} 
\usepackage[pdfstartview=FitH,pdfpagemode=None,colorlinks,linkcolor=blue,filecolor=blue,citecolor=purple,urlcolor=red]{hyperref}
\usepackage{multirow}
\usepackage{float}
\usepackage[section]{placeins}         
\usepackage{mathtools}  
\usepackage[dvipsnames]{xcolor}
\usepackage{amsmath,amssymb,algpseudocode,algorithm,cryptocode}\ifsubmission\else\usepackage{amsthm}\fi
\usepackage{mathrsfs} 
\usepackage[capitalize]{cleveref}
\usepackage{xspace}  
\usepackage{braket}
\usepackage{comment}
\usepackage{hhline}
\usepackage{tablefootnote}
\usepackage{tikz-cd}
\usepackage{url}
\usepackage{bbm}

\newif\ifnotes
\notestrue

\ifsubmission
\else
\newtheorem{theorem}{Theorem}[section]

\newtheorem{claim}[theorem]{Claim}
\newtheorem{lemma}[theorem]{Lemma}

\newtheorem{corollary}[theorem]{Corollary}

\newtheorem{definition}[theorem]{Definition}

\newtheorem{remark}[theorem]{Remark}

\theoremstyle{remark}
\fi

\newtheorem{importedtheorem}[theorem]{Imported Theorem}

\Crefname{importedtheorem}{Imported Theorem}{Imported Theorems}
\Crefname{theorem}{Theorem}{Theorems}
\Crefname{proposition}{Proposition}{Propositions}
\Crefname{claim}{Claim}{Claims}
\Crefname{lemma}{Lemma}{Lemmas}
\Crefname{conjecture}{Conjecture}{Conjectures}
\Crefname{corollary}{Corollary}{Corollaries}
\Crefname{construction}{Construction}{Constructions}
\Crefname{property}{Property}{Properties}

\Crefname{definition}{Definition}{Definitions}
\Crefname{assumption}{Assumption}{Assumptions}
\Crefname{notation}{Notation}{Notations}

\Crefname{question}{Question}{Questions}
\Crefname{remark}{Remark}{Remarks}
\Crefname{comment}{Comment}{Comments}
\Crefname{fact}{Fact}{Facts}

\newcommand{\secp}{\lambda}

\def\cA{{\cal A}}
\def\cB{{\cal B}}
\def\cC{{\cal C}}
\def\cD{{\cal D}}
\def\cE{{\cal E}}
\def\cF{{\cal F}}
\def\cG{{\cal G}}

\def\cL{{\cal L}}

\def\cQ{{\cal Q}}
\def\cR{{\cal R}}
\def\cS{{\cal S}}

\def\cV{{\cal V}}

\def\cZ{{\cal Z}}

\def\sA{{\mathsf A}}

\def\sC{{\mathsf C}}
\def\sD{{\mathsf D}}
\def\sE{{\mathsf E}}

\def\sM{{\mathsf M}}

\def\sO{{\mathsf O}}
\def\sP{{\mathsf P}}

\def\sR{{\mathsf R}}

\def\sT{{\mathsf T}}

\def\sV{{\mathsf V}}
\def\sW{{\mathsf W}}
\def\sX{{\mathsf X}}
\def\sY{{\mathsf Y}}
\def\sZ{{\mathsf Z}}

\def\sAux{{\mathsf{Aux}}}

\def\bbC{{\mathbb C}}

\def\bbI{{\mathbb I}}

\def\bbN{{\mathbb N}}

\def\bbR{{\mathbb R}}

\newcommand{\bq}{\mathbf{q}}

\def\negl{{\rm negl}}

\newcommand{\pk}{\mathsf{pk}}
\newcommand{\sk}{\mathsf{sk}}

\newcommand{\com}{\mathsf{com}}
\newcommand{\Com}{\mathsf{Com}}

\newcommand{\Enc}{\mathsf{Enc}}
\newcommand{\Dec}{\mathsf{Dec}}

\newcommand{\ct}{\mathsf{ct}}

\newcommand{\Eval}{\mathsf{Eval}}

\newcommand{\EXP}{\mathsf{EXP}}
\newcommand{\CEXP}{\mathsf{C}\text{-}\mathsf{EXP}}
\newcommand{\EVEXP}{\mathsf{EV}\text{-}\mathsf{EXP}}

\newcommand{\Gen}{\mathsf{Gen}}

\newcommand{\vk}{\mathsf{vk}}

\newcommand{\proref}[1]{Protocol~\protect\ref{#1}}

\newenvironment{boxfig}[2]{\begin{figure}[#1]\fbox{
    \begin{minipage}{\linewidth}
    \vspace{0.2em}\makebox[0.025\linewidth]{}    \begin{minipage}{0.95\linewidth}{{#2 }}
    \end{minipage}\vspace{0.2em}\end{minipage}}}{\end{figure}}

\newcommand{\pprotocol}[4]{
\begin{boxfig}{h!}{
\begin{center}
\textbf{#1}
\end{center}
    #4
\vspace{0.2em} } \caption{\label{#3} #2}
\end{boxfig}
}

\newcommand{\protocol}[4]{
\pprotocol{#1}{#2}{#3}{#4} }

\newcommand{\Real}{\mathsf{REAL}}
\newcommand{\Ideal}{\mathsf{IDEAL}}

\newcommand{\Ver}{\mathsf{Ver}}

\newcommand{\abort}{\mathsf{abort}}

\newcommand{\OT}{\mathsf{OT}}

\newcommand{\Del}{\mathsf{Del}}
\newcommand{\Rev}{\mathsf{Rev}}
\newcommand{\TD}{\mathsf{TD}}
\newcommand{\Fail}{\mathsf{FAIL}}

\newcommand{\TRE}{\mathsf{TRE}}

\newcommand{\PKE}{\mathsf{PKE}}
\newcommand{\Hyb}{\mathsf{Hyb}}
\newcommand{\Advt}{\mathsf{Advt}}
\newcommand{\Tr}{\mathsf{Tr}}
\newcommand{\cert}{\mathsf{cert}}
\newcommand{\counternew}{\ensuremath{\gamma}}
\newcommand{\SSS}{\mathsf{SS}}
\newcommand{\Share}{\mathsf{Share}}
\newcommand{\Rec}{\mathsf{Rec}}

\newcommand{\CD}{\mathsf{CD}}
\newcommand{\FHE}{\mathsf{FHE}}
\newcommand{\RTRE}{\mathsf{RTRE}}

\newcommand{\WE}{\mathsf{WE}}
\newcommand{\ABE}{\mathsf{ABE}}
\newcommand{\KG}{\mathsf{KeyGen}}
\newcommand{\msk}{\mathsf{msk}}

\newcommand{\socom}{\mathsf{so}\text{-}\mathsf{com}}

\newcommand{\DelReq}{\mathsf{DelReq}}
\newcommand{\DelRes}{\mathsf{DelRes}}

\begin{document}
\def\doi#1{\url{https://doi.org/#1}}

\title{Cryptography with Certified Deletion}
\ifsubmission
\else
\author{James Bartusek\thanks{UC Berkeley. Email: \texttt{bartusek.james@gmail.com}} \and Dakshita Khurana\thanks{UIUC. Email: \texttt{dakshita@illinois.edu}.}}
\fi
\date{}
\maketitle

\ifsubmission
\else
\pagenumbering{gobble}
\fi

\vspace{-10mm}
\begin{abstract}
    We propose a unifying framework that yields an array of cryptographic primitives with certified deletion. These primitives enable a party in possession of a quantum ciphertext to generate a classical certificate that the encrypted plaintext has been information-theoretically deleted, and cannot be recovered even given unbounded computational resources.
    
\begin{itemize}    

\item For $X \in \{\mathsf{public}\text{-}\mathsf{key},\mathsf{attribute\text{-}based},\mathsf{fully\text{-}homomorphic},\mathsf{witness},\mathsf{timed}\text{-}\mathsf{release}\}$, our compiler converts any (post-quantum) $X$ encryption to $X$ encryption with certified deletion.

In addition, we compile statistically-binding commitments to statistically-binding commitments with certified everlasting hiding. As a corollary, we also obtain statistically-sound zero-knowledge proofs for QMA with certified everlasting zero-knowledge assuming statistically-binding commitments.

\item We also obtain a strong form of everlasting security for two-party and multi-party computation in the dishonest majority setting. 
While simultaneously achieving everlasting security against \emph{all} parties in this setting is known to be impossible, we introduce {\em everlasting security transfer (EST)}. This enables \emph{any one} party (or a subset of parties)  to dynamically and certifiably information-theoretically delete other participants' data after protocol execution. 

We construct general-purpose secure computation with EST assuming statistically-binding commitments, which can be based on one-way functions or pseudorandom quantum states.
\end{itemize}

We obtain our results by developing a novel proof technique to argue that a bit $b$ has been {\em information-theoretically deleted} from an adversary's view once they output a valid deletion certificate, despite having been previously {\em information-theoretically determined} by the ciphertext they held in their view. This technique may be of independent interest.

\end{abstract}
\newpage

\ifsubmission
\else
{
  \hypersetup{linkcolor=Violet}
  \setcounter{tocdepth}{2}
  \tableofcontents
}
\newpage
\fi

\ifsubmission
\else
\pagenumbering{arabic}
\fi
\section{Introduction}
\paragraph{Deletion in a classical world.} On classical devices, data is stored and exchanged as a string of bits. There is nothing that can prevent an untrusted device with access to such a string from making arbitrarily many copies of it. Thus, it seems hopeless to try to {\em force} an untrusted device to delete classical data. Even if the string is merely a ciphertext encoding an underlying plaintext, there is no way to prevent a server from keeping that ciphertext around in memory forever. If at some point in the future, the security of the underlying encryption scheme is broken either via brute-force or major scientific advances, or if the key is compromised and makes its way to the server, the server will be able to recover the underlying plaintext. This may be unacceptable in situations where extremely sensitive data is being transmitted or computed upon. 

In fact, there has recently been widespread interest in holding data collectors accountable in responding to ``data deletion requests'' from their clients, as evidenced by data deletion clauses in legal regulations adopted by the European Union \cite{european_commission_regulation_2016} and California \cite{Cal}. Unfortunately, the above discussion shows that these laws cannot be cryptographically enforced against malicious data collectors, though there has been recent work on cryptographically \emph{formalizing} what it means for \emph{honest} data collectors to follow such guidelines \cite{EC:GarGolVas20}. 


\vspace{-2mm}
\paragraph{Deletion in a quantum world.} The \emph{uncertainty principle} \cite{1927ZPhy...43..172H}, which lies at the foundation of quantum mechanics, completely disrupts the above classical intuition. It asserts the existence of pairs of measurable quantities such that precisely determining one quantity (e.g. the position of an electron) implies the \emph{inability} to determine the other (e.g. the momentum of the electron). While such effects only become noticeable at an extreme microscopic scale, the pioneering work of Wiesner \cite{Wiesner1983ConjugateC} suggested that the peculiar implications of the uncertainty principle could be leveraged to perform seemingly impossible ``human-scale'' information processing tasks.

Given the inherent ``destructive'' properties of information guaranteed by the uncertainty principle, provable data deletion appears to be a natural information processing task that, while impossible classically, may become viable quantumly. Surprisingly, the explicit study of data deletion in a quantum world has only begun recently. However, over the last few years, this question has been explored in many different contexts. Initial work studied deletion in the context of non-local games \cite{PhysRevA.97.032324} and information-theoretic proofs of deletion with partial security~\cite{CRW20}, while the related notion of \emph{revocation} was introduced in~\cite{EC:Unruh14}.

The work of~\cite{10.1007/978-3-030-64381-2_4} first considered certified deletion in the context of encryption schemes, leveraging the uncertainty principle to obtain one-time pad encryption with certified deletion. This caused a great deal of excitement, leading to many recent followup works on deletion in a cryptographic context: device-independent security of one-time pad encryption with certified deletion \cite{https://doi.org/10.48550/arxiv.2011.12704}, public-key and attribute-based encryption with certified deletion~\cite{10.1007/978-3-030-92062-3_21}, commitments and zero-knowledge with certified everlasting hiding~\cite{cryptoeprint:2021:1315}, and most recently fully-homomorphic encryption with certified deletion~\cite{cryptoeprint:2022:295}.


\vspace{-2mm}
\paragraph{This work.}
Our work makes new definitional, conceptual and technical contributions.
Our key contribution is a new proof technique to show that many natural encryption schemes satisfy security with certified deletion.
This improves prior work in many ways, as we summarize below.
\begin{enumerate}
\item {\bf A unified framework.} We present a simple compiler that relies on conjugate coding/BB84 states \cite{Wiesner1983ConjugateC,BenBra84} to bootstrap semantically-secure cryptosystems to semantically-secure cryptosystems with certified deletion.
For any $X \in$ \{public-key encryption, attribute-based encryption, witness encryption, timed-release encryption, statistically-binding commitment\}, 
we immediately obtain
``$X$ with certified deletion'' by plugging $X$ into our compiler.
This compiler builds on~\cite{10.1007/978-3-030-64381-2_4}, who used BB84 states in the context of certified deletion for one-time pad encryption.

%
%

\item {\bf Stronger definitions. }
We consider a strong definition of security with certified deletion for public-key primitives, which stipulates that if an adversary in possession of a quantum ciphertext encrypting bit $b$ issues a certificate of deletion which passes verification, then the bit $b$ must now be {\em information-theoretically} hidden from the adversary. 

Previous definitions of public-key and fully-homomorphic encryption with certified deletion \cite{10.1007/978-3-030-92062-3_21,cryptoeprint:2022:295} considered a weaker experiment, inspired by~\cite{10.1007/978-3-030-64381-2_4}, where after deletion, the adversary is explicitly given the secret key, but is still required to be computationally bounded. For the public-key setting, we consider this prior definition to capture a (strong) \emph{security against key leakage} property, as opposed to a \emph{certified deletion} property%
%
\footnote{In contrast, in the one-time pad encryption setting as considered by \cite{10.1007/978-3-030-64381-2_4}, the original encrypted message is already information-theoretically hidden from the adversary, so to obtain any interesting notion of certified deletion, one must explicitly consider leaking the secret key.}.
We show that the everlasting flavor of our definition implies prior definitions (\cref{sec:implication}). Intuitively, this is because for public-key schemes, an adversary can sample a secret key on its own given sufficient computational resources.
Moreover, in the case of fully-homomorphic encryption (FHE), prior work \cite{cryptoeprint:2022:295}  considered definitions (significantly) weaker than semantic security.\footnote{Subsequent to the original posting of our paper on arXiv, an update to \cite{cryptoeprint:2022:295} was posted with somewhat different results. We provide a comparison between our work and the updated version of \cite{cryptoeprint:2022:295} in \cref{subsec:concurrent}.} We obtain the first semantically-secure FHE with certified deletion from standard LWE.

\item {\bf Simpler constructions and weaker assumptions.} Our compiler removes the need to rely on complex cryptographic primitives such as non-committing encryption and indistinguishability obfuscation as in~\cite{10.1007/978-3-030-92062-3_21}, or idealized models such as random oracles as in~\cite{EC:Unruh14,cryptoeprint:2021:1315}, or complex quantum states (such as Gaussian coset states) as in~\cite{cryptoeprint:2022:295}, instead yielding simple schemes satisfying certified deletion for a range of primitives from BB84 states and minimal assumptions.

In fact, reliance on non-committing encryption was a key reason that prior techniques did not yield homomorphic encryption schemes with certified deletion, since compact homomorphic encryption schemes cannot simultaneously be non-committing~\cite{PKC:KatThiZho13}. Our work builds simple homomorphic encryption schemes that support certified deletion by eliminating the need to rely on non-committing properties, and instead only relying on semantic security of an underlying encryption scheme.
\item {\bf Overcoming barriers to provable security.} How can one prove that a bit $b$ has been {\em information-theoretically deleted} from an adversary's view once they produce a valid deletion certificate, while it was previously information-theoretically {\em determined} by the ciphertext they hold in their view?

Prior work~\cite{EC:Unruh14,10.1007/978-3-030-92062-3_21,cryptoeprint:2021:1315,cryptoeprint:2022:295} resorted to either idealized models or weaker definitions, and constructions with layers of indirection, in order to get around this barrier.
We develop a novel proof technique that resolves this issue by (1) carefully deferring the dependence of the experiment on the plaintext bit, and (2) identifying an efficiently checkable predicate on the adversary's state after producing a valid deletion certificate. We rely on semantic security of encryption to show that this predicate must hold, and we argue that if the predicate holds, the adversary's left-over state is statistically independent of the plaintext bit.
This allows us to prove certified deletion security for simple and natural schemes.

\item {\bf New implications to secure computation: Everlasting Security Transfer (EST).}
We introduce the concept of {\em everlasting security transfer}. Everlasting security guarantees (malicious) security against a participant in a secure two-(or multi-)party computation protocol even if the participant becomes computationally unbounded after protocol execution. We introduce and build secure computation protocols where participants are able to {\em transfer} everlasting security properties from one party to another, even after the protocol ends.  
\end{enumerate}
%
We elaborate on our results in more detail below, then we provide an overview of our techniques.

\subsection{Our results}
\paragraph{Warmup: secret sharing with certified deletion.}
We begin by considering certified deletion in the context of one of the simplest cryptographic primitives: information-theoretic, two-out-of-two secret sharing.
Here, a dealer Alice would like to share a classical secret bit $b$ between two parties Bob and Charlie, such that
\begin{enumerate}
\item {\bf (Secret sharing.)} The individual views of Bob and Charlie perfectly hide $b$, while the joint view of Bob and Charlie can be used to reconstruct $b$, and 
\item {\bf (Certified deletion.)} Bob may generate a deletion certificate for Alice, guaranteeing that $b$ has been {\em information theoretically removed} from the \emph{joint} view of Bob and Charlie.
\end{enumerate}
That is, as long as Bob and Charlie do not collude at the time of generating the certificate of deletion, their joint view upon successful verification of this certificate is guaranteed to become independent of $b$. As long as the certificate verifies, $b$ will be perfectly hidden from Bob and Charlie {\em even if} they decide to later collude.



To build such a secret sharing scheme, we start by revisiting the usage of conjugate coding/BB84 states to obtain encryption with certified deletion, which was first explored in~\cite{10.1007/978-3-030-64381-2_4}.
While the construction in~\cite{10.1007/978-3-030-64381-2_4} relies on a seeded randomness extractor in combination with BB84 states, we suggest a simpler alternative that replaces the seeded extractor with the XOR function. Looking ahead, this simplification combined with other proof techniques will help generically lift our secret sharing scheme to obtain several encryption schemes with certified deletion.

Consider a random string $x \gets \{0,1\}^\secp$, and a random set of bases $\theta \gets \{0,1\}^\secp$ (where 0 corresponds to the standard basis and 1 corresponds to the Hadamard basis).
To obtain a scheme with certifiable deletion, we will build on the intuition that it is impossible to recover $x$ given only BB84 states $\ket{x}_\theta$ without knowledge of the basis $\theta$. Furthermore, measuring $\ket{x}_\theta$ in an incorrect basis $\theta'$ will destroy (partial) information about $x$.

Thus to secret-share a bit $b$ in a way that supports deletion, 
the dealer will sample $x \gets \{0,1\}^\secp$ and bases $\theta \gets \{0,1\}^\secp$.
Bob's share is then
$$\ket{x}_\theta$$ 
and Charlie's share is
$$\theta, b' = b \oplus \bigoplus_{i:\theta_i = 0}x_i$$
That is, in Charlie's share, $b$  is masked by the bits of $x$ that are encoded in the standard basis.

We note that Bob's share contains only BB84 states while Charlie's share is entirely classical. 
Bob can now produce a certificate of deletion by returning the results of measuring all his BB84 states in the Hadamard basis,  
and Alice will accept as a valid certificate any string $x'$ such that $x_i = x'_i$ for all $i$ where $\theta_i = 1$.
We show that this scheme is indeed a two-out-of-two secret sharing scheme that satisfies certified deletion as defined above.

\paragraph{A conceptually simple and generic compiler.}
As our key technical contribution, we upgrade the secret sharing with certified deletion scheme to the public-key setting by encrypting Charlie's share.
In more detail, to encrypt a bit $b$ with respect to any encryption scheme, we first produce two secret shares of $b$ as described above, and then release a ciphertext that contains (1) Bob's share in the clear and (2) an encryption of Charlie's share.
To certifiably delete a ciphertext, one needs to simply measure the quantum part of the ciphertext (i.e., Bob's share) in the Hadamard basis.
Intuitively, since information about the bases (Charlie's share) is hidden at the time of producing the certificate of deletion, generating a certificate that verifies must mean information theoretically losing the description of computational basis states. 

This method of converting a two-party primitive (i.e. secret sharing with certified deletion) into one-party primitives (i.e. encryption schemes with certified deletion) is reminiscent of other similar compilers in the literature, for instance those converting probabilistically checkable proofs to succinct arguments~\cite{BMW,C:KalRaz09}. In our case, just like those settings, while the intuition is relatively simple, the proof turns out to be fairly non-trivial.

\paragraph{Our main theorem.} In (almost) full generality, our main theorem says the following.\footnote{In order to fully capture all of our applications, we actually allow $\cZ_\secp$ to operate on all inputs, including the BB84 states. See \cref{sec:main} for the precise details.} Consider an arbitrary family of distributions $\{\cZ_\secp(\theta)\}_{\secp \in \bbN, \theta \in \{0,1\}^\secp}$ and an arbitrary class $\mathscr{A}$ of computationally bounded adversaries $\cA = \{\cA_\secp\}_{\secp \in \bbN}$, such that $\cZ_\secp(\theta)$ semantically hides $\theta$ against $\cA_\secp$. Then, consider the following distribution $\widetilde{\cZ}_\secp^{\cA_\secp}(b)$ over quantum states, parameterized by a bit $b \in \{0,1\}$. 

\begin{itemize}
    \item Sample $x,\theta \gets \{0,1\}^\secp$ and initialize $\cA_\secp$ with \[\cZ_\secp(\theta),b\oplus \bigoplus_{i: \theta_i = 0}x_i,\ket{x}_\theta.\]
    \item $\cA_\secp$'s output is parsed as a bitstring $x' \in \{0,1\}^\secp$ and a residual state on register $\sA'$.
    \item If $x_i = x_i'$ for all $i$ such that $\theta_i = 1$ then output $\sA'$, and otherwise output a special symbol $\bot$.
\end{itemize}

Then,

\begin{theorem}\label{thm:main-intro}
For every $\cA \in \mathscr{A}$, 
the trace distance between  $\widetilde{Z}^{\cA_\secp}_\secp(0)$ and $\widetilde{Z}^{\cA_\secp}_\secp(1)$ is $\negl(\secp)$.
\end{theorem}

Intuitively, this means that as long as the adversary $\cA_\secp$ is computationally bounded {\em at the time of producing any deletion certificate} $x'$ that properly verifies (meaning that $x'_i$ is the correct bit encoded at index $i$ for any indices encoded in the Hadamard basis), their left-over state {\em statistically} contains only negligible information about the original encrypted bit $b$. That is, once the certificate verifies, information about $b$ cannot be recovered  information-theoretically even given unbounded time from the adversary's residual state.

This theorem is both quite simple and extremely general. The quantum part that enables certified deletion only involves simple BB84 states, and we require no additional properties of the underlying distribution $\cZ_\secp$ except for the fact that $\cZ_\secp(\theta)$ and $\cZ_\secp(0^\secp)$ are indistinguishable to some class of adversaries.
\footnote{It may seem counter-intuitive that the certified deletion guarantees provided by our theorem hold even when
instantiating $\cZ_\secp$ with general semantically secure schemes, such as a fully-homomorphic encryption scheme. In
particular, what if an adversary evaluated the FHE to recover a classical encryption of $b$, and then reversed their
computation and finally produced a valid deletion certificate? This may seem to contradict everlasting security,
since a classical ciphertext could be used to recover $b$ given unbounded time. However, this attack is actually not
feasible. After performing FHE evaluation coherently, the adversary would obtain a register holding a superposition
over classical ciphertexts encrypting $b$, but with different random coins. Measuring this superposition to obtain a
single classical ciphertext would collapse the state, and prevent the adversary from reversing their computation to
eventually produce a valid deletion certificate. Indeed, our Theorem rules out this (and all other) efficient attacks.} We now discuss our (immediate) applications in more detail. 


\paragraph{Public-key, attribute-based and witness encryption.} Instantiating the distribution $\cZ_\secp$ with the encryption procedure for any public-key encryption scheme, 
we obtain a public-key encryption scheme with certified deletion. 

We also observe that we can instantiate the distribution $\cZ_\secp$ with the encryption procedure for any \emph{attribute-based} encryption scheme, and immediately obtain an attribute-based encryption scheme with certified deletion. Previously, this notion was only known under the assumption of indistinguishability obfuscation, and also only satisfied the weaker key leakage style definition discussed above \cite{10.1007/978-3-030-92062-3_21}. Finally, instantiating $\cZ_\secp$ with any \emph{witness encryption} scheme implies a witness encryption scheme with certified deletion.

\paragraph{Fully-homomorphic encryption.} Next, we consider the question of computing on encrypted data. We observe that, if $\cZ_\secp$ is instantiated with the encryption procedure $\Enc$ for a \emph{fully-homomorphic} encryption scheme \cite{10.1145/1536414.1536440,6108154,C:GenSahWat13}, then given $\ket{x}_\theta, \Enc(\theta, b \oplus \bigoplus_{i: \theta_i = 0} x_i)$, one could run a homomorphic evaluation procedure in superposition to recover (a superposition over) $\Enc(b)$. Additionally, given multiple ciphertexts, one can even compute arbitrary functionalities over the encrypted plaintexts. Moreover, if such evaluation is done \emph{coherently} (without performing measurements), then it can be reversed and the deletion procedure can subsequently be run on the original ciphertexts. 

This immediately implies what we call a ``blind delegation with certified deletion'' protocol, which allows a computationally weak client to utilize the resources of a computationally powerful server, while (i) keeping its data hidden from the server during the protocol, and (ii) ensuring that its data is \emph{information-theoretically} deleted from the server afterwards, by requesting a certificate of deletion. We show that, as long as the server behaves honestly
during the ``function evaluation'' phase of the protocol, 
then even if it is arbitrarily malicious after the function evaluation phase, it cannot both pass deletion verification and maintain any information about the client's original plaintexts.

Recently, Poremba \cite{cryptoeprint:2022:295} also constructed a fully-homomorphic encryption scheme satisfying a weaker notion of certified deletion.\footnote{We discuss comparisons with a recently updated version of \cite{cryptoeprint:2022:295} in \cref{subsec:concurrent}.} 
In particular, the guarantee in~\cite{cryptoeprint:2022:295} is that from the perspective of any server that passes deletion with \emph{sufficiently high probability}, there is significant entropy in the client's original \emph{ciphertext}. This does not necessarily imply anything about the underlying plaintext, since a ciphertext encrypting a fixed bit $b$ may be (and usually will be) highly entropic. Moreover, their construction makes use of relatively complicated and highly entangled \emph{Gaussian coset states} in order to obtain these deletion properties. In summary, our framework simultaneously strengthens the security (to standard semantic security of the plaintext) and simplifies the construction of fully-homomorphic encryption with certified deletion. We also remark that neither our work nor \cite{cryptoeprint:2022:295} considers security against servers that may be malicious during the function evaluation phase of the blind delegation with certified deletion protocol.
We leave obtaining security against fully malicious servers as an interesting direction for future research.

\paragraph{Commitments and zero-knowledge.} Next, we consider \emph{commitment schemes}. A fundamental result in quantum cryptography states that one cannot use quantum communication to build a commitment that is simultaneously statistically hiding and statistically binding \cite{Mayers97,LoChau97}. Intriguingly, \cite{cryptoeprint:2021:1315} demonstrated the feasibility of statistically-binding commitments with a \emph{certified everlasting hiding} property, where hiding is computational during the protocol, but becomes information-theoretic after the receiver issues a valid deletion certificate. However, their construction relies on the idealized quantum random oracle model. 
Using our framework, we show that \emph{any} (post-quantum) statistically-binding computationally-hiding commitment implies a statistically-binding commitment with certified everlasting hiding. Thus, we obtain statistically-binding commitments with certified everlasting hiding in the plain model from post-quantum one-way functions, and even from plausibly weaker assumptions like \emph{pseudorandom quantum states} \cite{cryptoeprint:2021:1663,cryptoeprint:2021:1691}. 

Following implications in \cite{cryptoeprint:2021:1315} from commitments with certified deletion to zero-knowledge, we also obtain interactive proofs for NP (and more generally, QMA) with \emph{certified everlasting zero-knowledge}.
These are proofs that are statistically sound, and additionally the verifier may issue a classical certificate {\em after the protocol ends} showing that the verifier has information-theoretically deleted all secrets about the statement being proved. Once a computationally bounded verifier issues a valid certificate, the proof becomes {\em statistically} zero-knowledge (ZK).
Similarly to the case of commitments, while proofs for QMA or NP are unlikely to simultaneously satisfy {\em statistical soundness} and {\em statistical ZK},~\cite{cryptoeprint:2021:1315} previously introduced and built
statistically sound, certified everlasting ZK proofs in the random oracle model.
On the other hand, we obtain a construction in the plain model from any statistically-binding commitment.

\paragraph{Timed-release encryption.} As another immediate application, we consider the notion of \emph{revocable} timed-release encryption. Timed-release encryption schemes (also known as time-lock puzzles) have the property that, while ciphertexts can eventually be decrypted in some polynomial time, it takes \emph{at least} some (parallel) $T(\secp)$ time to do so. \cite{EC:Unruh14} considered adding a \emph{revocable} property to such schemes, meaning that the recipient of a ciphertext can either eventually decrypt the ciphertext in $\geq T(\secp)$ time, or issue a certificate of deletion proving that they will \emph{never} be able to obtain the plaintext. \cite{EC:Unruh14} constructs semantically-secure revocable timed-release encryption assuming post-quantum timed-release encryption, but with the following drawbacks: the certificate of deletion is a \emph{quantum state}, and the underlying scheme must either be \emph{exponentially} hard or security must be proven in the idealized quantum random oracle model. 

We can plug any post-quantum timed-release encryption scheme into our framework, and obtain revocable timed-released encryption from (polynomially-hard) post-quantum timed-released encryption, with a classical deletion certificate. Note that, when applying our main theorem, we simply instantiate the class of adversaries to be those that are $T(\secp)$-parallel time bounded.

\paragraph{Secure computation with Everlasting Security Transfer (EST).} 
Secure computation allows mutually distrusting participants to compute on joint private inputs while revealing no information beyond the output of the computation.
The first templates for secure computation that make use of quantum information were proposed in a combination of works by Crépeau and Kilian~\cite{FOCS:CreKil88}, and Kilian~\cite{STOC:Kilian88}. For a while~\cite{MayersSalvail94,STOC:Yao95} it was believed that {\em unconditionally secure computation} could be realized based on a specific cryptographic building block: an {\em unconditionally secure quantum bit commitment}. Unfortunately, beliefs that unconditionally secure quantum bit commitments exist~\cite{FOCS:BCJL93} were subsequently proven false~\cite{Mayers97,LoChau97}, and the possibility of unconditional secure computation was also ruled out~\cite{Lo}.

As such, secure computation protocols must either assume an honest majority or necessarily rely on computational hardness to achieve security against adversaries that are computationally bounded. But this may be troublesome when participants wish to compute on extremely sensitive data, such as medical or government records.
In particular, consider a server that computes on highly sensitive data and keeps information from the computation around in memory forever. Such a server may be able to eventually recover data if the underlying hardness assumption breaks down in the future.
In this setting, it is natural to ask: 
Can we use computational assumptions to design ``everlasting'' secure protocols against an adversary that is computationally bounded during protocol execution but becomes \emph{computationally unbounded} after protocol execution?

Unfortunately, everlasting secure computation against {\em every participant in a protocol} is {\em also} impossible~\cite{C:Unruh13} for most natural two-party functionalities (or multi-party functionalities against dishonest majority corruptions).
For the specific case of two parties, this means that it is impossible to achieve everlasting security against {\em both} players, without relying on special tools like trusted/ideal hardware. 
Nevertheless, 
it is still possible to obtain 
everlasting (or even the stronger notion of statistical) security against one unbounded participant (see eg.,~\cite{TCC:KhuMug20} and references therein).
But in {\em all existing protocols}, which party may be unbounded and which one must be assumed to be computationally bounded must necessarily be fixed {\em before protocol execution}.
We ask if this is necessary. That is,

\begin{center}
    {\em Can participants transfer everlasting security from one party\\ to another even after a protocol has already been executed?}
\end{center}

We show that the answer is yes, under the weak cryptographic assumption that (post-quantum) statistically-binding computationally-hiding bit commitments exist. These commitments can in turn be based on one-way functions~\cite{C:Naor89} or even pseudo-random quantum states~\cite{cryptoeprint:2021:1691,cryptoeprint:2021:1663}.

We illustrate our novel security property by considering it in the context of Yao's classic millionaire problem~\cite{FOCS:Yao82b}. Stated simply, this toy problem requires two millionaires to securely compute who is richer without revealing to each other or anyone else information about their wealth.
That is, the goal is to only reveal the bit indicating whether $x_1 > x_2$ where
$x_1$ is Alice’s private input and $x_2$ is Bob’s private input.
In our extension, the millionaires would also like (certified) everlasting security against the wealthier party, while maintaining standard simulation-based security against the other party. Namely, if $x_1 > x_2$ then the protocol should satisfy certified everlasting security against Alice and standard simulation-based security against computationally bounded Bob; and if it turns out that $x_2 \geq x_1$, then the protocol should satisfy certified everlasting security against Bob and simulation-based security against bounded Alice.

More generally, our goal is to enable any one party (or a subset of parties) to dynamically and certifiably information-theoretically delete other participants' inputs, during or even after a secure computation protocol completes. At the same time, the process of deletion should not destroy standard simulation-based security. 

We build a two-party protocol that is (a) designed to be secure against computationally {\em unbounded Alice} and computationally {\em bounded Bob}. In addition, even after the protocol ends, (b) Bob 
has the capability to generate a proof whose validity 
certifies that the protocol has now become secure against {\em unbounded Bob} while remaining secure against {\em bounded Alice}. 
In other words, verification of the proof implies that everlasting security roles have switched: this is why we call this property {\em everlasting security transfer}.
This 
implies zero-knowledge proofs for NP/QMA with certified everlasting ZK as a special case.
We also extend this result to obtain \emph{multi-party computation} where even after completion of the protocol, any arbitrary subset of parties can 
certifiably, information-theoretically remove information about the other party inputs from their view. 

At a high level, we build these protocols by carefully combining \cref{thm:main-intro} with additional techniques to ensure that having one party generate a certificate of deletion does not ruin standard (simulation-based, computational) security against the other party.

In what follows, we provide a detailed overview of our techniques.

\subsection{Techniques}\label{subsec:tech-overview}

We first provide an overview of our proof of \cref{thm:main-intro}.

Our construction and analysis include a couple of crucial differences from previous work on certified deletion. First, our analysis diverges from recent work~\cite{10.1007/978-3-030-64381-2_4,cryptoeprint:2022:295} that relies on ``generalized uncertainty relations'' which provide lower bounds on the sum of entropies resulting from two incompatible measurements, and instead builds on the simple but powerful ``quantum cut-and-choose'' formalism of Bouman and Fehr \cite{C:BouFeh10}. Next, we make crucial use of an \emph{unseeded} randomness extractor (the XOR function), as opposed to a seeded extractor, as used by \cite{10.1007/978-3-030-64381-2_4}. 

\paragraph{Delaying the dependence on $b$.} A key tension that must be resolved when proving a claim like \cref{thm:main-intro} is the following: how to \emph{information-theoretically} remove the bit $b$ from the adversary's view, when it is initially information-theoretically \emph{determined} by the adversary's input. Our first step towards a proof is a simple change in perspective. We will instead imagine sampling the distribution by \emph{guessing} a uniformly random $b' \gets \{0,1\}$, and initializing the adversary with $\ket{x}_\theta, b',\cZ_\secp(\theta)$. Then, we abort the experiment (output $\bot$) if it happens that $b' \neq b \oplus \bigoplus_{i: \theta_i = 0} x_i$. Since $b'$ was a uniformly random guess, we always abort with probability exactly 1/2, and thus the trace distance between the $b=0$ and $b=1$ outputs of this experiment is at least half the trace distance between the outputs of the original experiment.\footnote{One might be concerned that extending this argument to multi-bit messages may eventually reduce the advantage by too much, since the entire message must be guessed. However, it actually suffices to prove \cref{thm:main-intro} for single bit messages and then use a bit-by-bit hybrid argument to obtain security for any polynomial-length message.}

Now, the bit $b$ is only used by the experiment to determine whether or not to output $\bot$. This is not immediately helpful, since the result of this ``abort decision'' is of course included in the output of the experiment. However, we can make progress by delaying this abort decision (and thus, the dependence on $b$) until \emph{after} the adversary outputs $x'$ and their residual state on register $\sA'$. To do so, we will make use of a common strategy in quantum cryptographic proofs: replace the BB84 states $\ket{x}_\theta$ with halves of EPR pairs $\frac{1}{\sqrt{2}}(\ket{00} + \ket{11})$. Let $\sC$ be the register holding the ``challenger's'' halves of EPR pairs, and $\sA$ be the register holding the other halves, which is part of the adversary's input. This switch is perfectly indistinguishable from the adversary's perspective, and it allows us to \emph{delay} the measurement of $\sC$ in the $\theta$-basis (and thus, delay the determination of the string $x$ and subsequent abort decision), until after the adversary outputs $(x',\sA')$.

We still have not shown that when the deletion certificate is accepted, information about $b$ doesn't exist in the output of the experiment. However, note that at this point it suffices to argue that $\bigoplus_{i: \theta_i = 0}x_i$ is distributed like a uniformly random bit, even conditioned on the adversary's ``side information'' on register $\sA'$ (which may be entangled with $\sC$). This is because, if $\bigoplus_{i: \theta_i = 0}x_i$ is uniformly random, then the outcome of the abort decision, whether $b'$ = $b \oplus \bigoplus_{i: \theta_i = 0} x_i$, is also a uniformly random bit, regardless of $b$.

\paragraph{Identifying an efficiently-checkable predicate.} 
To prove
that $\bigoplus_{i: \theta_i = 0}x_i$ is uniformly random, we will need to establish that the measured bits $\{x_i\}_{i: \theta_i = 0}$ contain sufficient entropy. To do this, we will need to make some claim about the structure of the state on registers $\sC_{i: \theta_i = 0}$. These registers are measured in the computational basis to produce $\{x_i\}_{i: \theta_i = 0}$, so if we could claim that these registers are in a Hadamard basis state, we would be done. We won't quite be able to claim something this strong, but we don't need to. Instead, we will rely on the following claim:
consider any (potentially entangled) state on systems $\sX$ and $\sY$, such that the part of the state on system $\sY$ is in a superposition of Hadamard basis states $\ket{u}_1$ where each $u$ is a vector of somewhat low Hamming weight.\footnote{It suffices to require that the relative Hamming weight of each $u$ is $< 1/2$.} Then, measuring $\sY$ in the computational basis and computing the XOR of the resulting bits produces a bit that is \emph{uniformly random and independent} of system $\sX$.\footnote{This proof strategy is inspired by the techniques of \cite{C:BouFeh10}, who show a similar claim using a \emph{seeded} extractor.} This claim can be viewed as saying that XOR is a good (seedless) randomness extractor for the quantum source of entropy that results from measuring certain structured states in the conjugate basis. Indeed, such a claim was developed to remove the need for \emph{seeded} randomness extraction in applications like quantum oblivious transfer \cite{ABKK}, and it serves a similar purpose here.\footnote{If we had tried to rely on generic properties of a seeded randomness extractor, as done in \cite{10.1007/978-3-030-64381-2_4}, 
we would still have had to deal with the fact the adversary's view includes an encryption of the seed, which is required to be \emph{uniform and independent} of the source of entropy. Even if the challenger's state can be shown to produce a sufficient amount of min-entropy when measured in the standard basis, we cannot immediately claim that this source of entropy is perfectly independent of the seed of the extractor. Similar issues with using seeded randomness extraction in a related context are discussed by \cite{EC:Unruh14} in their work on revocable timed-release encryption.}

Thus, it suffices to show that the state on registers $\sC_{i: \theta_i = 0}$ is only supported on low Hamming weight vectors in the Hadamard basis.
A priori, it is not clear why this would even be true, since $\sC,\sA$ are initialized with EPR pairs, and the adversary, who has access to $\sA$, can simply measure its halves of these EPR pairs in the computational basis. However, recall that the experiment we are interested in only outputs the adversary's final state when its certificate of deletion is valid, and moreover, a valid deletion certificate is a string $x'$ that matches $x$ in all the \emph{Hadamard} basis positions. Moreover, which positions will be checked is semantically hidden from the adversary. Thus, in order to be sure that it passes the verification, an adversary should intuitively be measuring most of its registers $\sA$ in the Hadamard basis.

\paragraph{Reducing to semantic security.}
One remaining difficulty in formalizing this intuition is that if the adversary knew $\theta$, it could decide which positions to measure in the Hadamard basis to pass the verification check, and then measure $\sA_{i:\theta_i = 0}$ in the computational basis in order to thwart the above argument from going through. And in fact, the adversary \emph{does} have information about $\theta$, encoded in the distribution $\cZ_\secp(\theta)$. 

This is where the assumption that $\cA_\secp$ cannot distinguish between $\cZ_\secp(\theta)$ and $\cZ_\secp(0^\secp)$ comes into play. We interpret the condition that registers $\sC_{i:\theta_i = 0}$ must be in a superposition of low Hamming weight vectors in the Hadamard basis (or verification doesn't pass) as an efficient predicate (technically a binary projective measurement) that can be checked by a reduction to the indistinguishability of distributions $\cZ_\secp(\theta)$ and $\cZ_\secp(0^\secp)$. Thus, this predicate must have roughly the same probability of being true when the adversary receives $\cZ_\secp(0^\secp)$. But now, since $\theta$ is independent of the adversary's view, we can show \emph{information-theoretically} that this predicate must be true with overwhelming probability. 


We note that the broad strategy of identifying an efficiently-checkable predicate which implies the \emph{uncheckable property that some information is random and independent of the adversary's view} has been used in similar (quantum cryptographic) contexts by Gottesman \cite{Gottesman2003UncloneableE} in their work on the related concept of \emph{uncloneable} (or perhaps more accurately, \emph{tamper-detectable}) encryption\footnote{In this notion, the adversary is an eavesdropper who sits between a ciphertext generator Alice and a ciphertext receiver Bob (using a symmetric-key encryption scheme), who attempts to learn some information about the ciphertext. The guarantee is that, \emph{either} the eavesdropper gains information-theoretically no information about the underlying plaintext, \emph{or} Bob can detect that the ciphertext was tampered with. While this is peripherally related to our setting, \cite{Gottesman2003UncloneableE} does not consider public-key encryption, and moreover Bob's detection procedure is quantum.} and by Unruh \cite{EC:Unruh14} in their work on revocable timed-release encryption.

\paragraph{Application: A variety of encryption schemes with certified deletion.} 
For any $X \in$ \{public-key encryption, attribute-based encryption, witness encryption, statistically-binding commitment, timed-release encryption\}, 
we immediately obtain ``$X$ with certified deletion'' by instantiating the distribution $\cZ_\lambda$ with the encryption/encoding procedure for $X$, and additionally encrypting/encoding the bit $b\oplus\bigoplus_{i:\theta_i = 0}x_i$ to ensure that semantic security holds regardless of whether the adversary deletes the ciphertext or not. 

Similarly, if $\cZ_\secp$ is instantiated with the encryption procedure for a \emph{fully-homomorphic} encryption scheme \cite{10.1145/1536414.1536440,6108154,C:GenSahWat13}, then the scheme also allows for arbitrary homomorphic operations over the ciphertext. We also note that such a scheme can be used for blind delegation with certified deletion, allowing a weak client to outsource computations to a powerful server and subsequently verify deletion of the plaintext.
In particular, a server may perform homomorphic evaluation coherently (i.e. by not performing any measurements), and return the register containing the output to the client. 
The client can coherently decrypt this register to obtain a classical outcome, then reverse the decryption operation and return the output register to the server. Finally, the server can use this register to reverse the evaluation operation and recover the original ciphertext. Then, the server can prove deletion of the original plaintext as above, i.e. measure the quantum state associated with this ciphertext in the Hadamard basis, and report the outcomes as their certificate.

\paragraph{Application: Secure computation with Everlasting Security Transfer (EST).}
Recall that in building two-party computation with EST, the goal is to build protocols
(a) secure against {\em unbounded Alice} and computationally {\em bounded Bob} such that, during or even after the protocol ends, (b) Bob 
can generate a proof whose validity 
certifies that the protocol has now become secure against {\em unbounded Bob} while remaining secure against {\em bounded Alice}.

Our goal is to realize two-party secure computation with EST from minimal cryptographic assumptions. 
We closely inspect a class of protocols for secure computation that do not a-priori have any EST guarantees, and develop techniques to equip them with EST. 

In particular, we observe that 
a key primitive called quantum oblivious transfer (QOT) is known to unconditionally imply secure computation of {\em all classical (and quantum) circuits}~\cite{STOC:Kilian88,C:CreVanTap95,EC:DGJMS20}. 
Namely, given OT with information-theoretic security, it is possible to build secure computation with everlasting (and even unconditional) security against unbounded participants. We recall that information-theoretically secure OT cannot exist in the plain model, even given quantum resources~\cite{Lo}.
However, for the case of EST, we 
establish a general sequential composition theorem (\cref{thm:compose-EST}) 
which shows that oblivious transfer with EST can be plugged into the above unconditional protocols to yield secure computation protocols with EST. 

Furthermore, a recent line of work~\cite{FOCS:CreKil88,C:DFLSS09,C:BouFeh10,BCKM2021,EC:GLSV21} establishes {\em ideal} commitments\footnote{The term ``ideal committment'' can sometimes refer to the commitment \emph{ideal funtionality}, but in this work we use the term ideal commitment to refer to a \emph{real-world protocol} that can be shown to securely implement the commitment ideal functionality.} as the basis for QOT. Intuitively, these are commitments that satisfy the (standard) notion of simulation-based security against computationally bounded quantum committers and receivers. Namely, for every adversarial committer (resp., receiver) that interacts with an honest receiver (resp., committer) in the real protocol, there is a simulator that interacts with the ideal commitment functionality and generates a simulated state that is indistinguishable from the committer's (resp., receiver's) state in the real protocol.
%
Our composition theorem (\cref{thm:compose-EST}) combined with~\cite{BCKM2021} also immediately shows that ideal commitments {\em with EST} imply QOT with EST. Thus, the problem reduces to building ideal commitments with EST.

%

\paragraph{Constructing Ideal Commitments with EST.}
An ideal commitment with EST 
satisfies statistical simulation-based security against unbounded committers, and computational simulation-based security against bounded receivers. Furthermore, after an optional delete/transfer phase succeeds, everlasting security is {\em transfered}: that is, then the commitment satisfies statistical (simulation-based) security against unbounded receivers, and remains computationally (simulation-based) secure against bounded committers.

To build ideal commitments with EST, we start with any commitment that satisfies standard computational hiding, and a strong form of binding: namely, simulation-based security against an unbounded malicious committer. At a high level, this means that there is an efficient extractor that can extract the input committed by an unbounded committer, thereby statistically simulating the view of the adversarial committer in its interaction with the ideal commitment functionality.
We call this a \emph{computationally-hiding statistically-efficiently-extractable} (CHSEE) commitment, and 
observe that prior work (\cite{BCKM2021}) builds such commitments from black-box use of any statistically-binding, computationally-hiding commitment. 
Our construction of ideal commitments with EST starts with CHSEE commitments, and proceeds in two steps, where the first involves new technical insights and the second follows from ideas in prior work~\cite{BCKM2021}.\\

\noindent{\bf Step 1: One-Sided Ideal Commitments with EST.} While CHSEE commitments satisfy simulation-based security against a malicious committer, they do not admit security transfer. Therefore, our first step is to add the EST property to CHSEE commitments, which informally additionally allows receivers to certifiably, information-theoretically, delete the committed input. We call the resulting primitive {\em one-sided ideal commitments with EST}. The word ``one-sided'' denotes that these commitments satisfy simulation-based security against any malicious committer, but are not necessarily simulation-secure against malicious receivers. Instead, these commitments semantically hide the committed bit from a malicious receiver and furthermore, support certified everlasting hiding against malicious receivers.
%
    
We observe that invoking \cref{thm:main-intro} while instantiating $\cZ_\secp$ with a CHSEE commitments already helps us add the certified everlasting hiding property to any CHSEE commitment. While this ensures the desired certified everlasting security against malicious receivers, the scheme appears to become insecure against malicious committers after certified deletion! 

To see why, recall that the resulting commitment is now
$\ket{x}_\theta, \Com\left(\theta, b'\right)$, where $\Com$ is a CHSEE commitment and $b' = b\oplus \bigoplus_{i: \theta_i = 0}x_i$.
In particular, to simulate (i.e., to extract the bit committed by) a malicious committer $\cC^*$, a simulator must extract the bases $\theta$ and masked bit $b'$ from the CHSEE commitment, measure the accompanying state $\ket{\psi}$ in basis $\theta$ to recover $x$, and then XOR the parity $\bigoplus_{i:\theta_i = 0} x_i$ with $b'$ to obtain the committed bit $b$. Thus, the 
simulator will have to first measure qubits of $\ket{\psi}$ that correspond to $\theta_i = 0$ in the computational basis to recover $x_i$ values at these positions. If the committer makes a delete request after this point, the simulator must measure {\em all positions} in the Hadamard basis to generate the certificate of deletion. But consider a cheating committer that (maliciously) generates the qubit at a certain position (say $i = 1$) as a half of an EPR pair, keeping the other half to itself. Next, this committer commits to $\theta_i = 0$ (i.e., computational basis) corresponding to the index $i = 1$. The simulation strategy outlined above will first measure the first qubit of $\ket{\psi}$ in the computational basis, and then later in the Hadamard basis to generate a deletion certificate. On the other hand, an honest receiver will only ever measure this qubit in the Hadamard basis to generate a deletion certificate. This makes it easy for such a committer to distinguish simulation from an honest receiver strategy, simply by measuring its half of the EPR pair in the Hadamard basis, thereby breaking simulation security post-deletion.
    
To prevent this attack, we modify the scheme so that the committer $\cC^*$ {\em only ever obtains} the receiver's outcomes of Hadamard basis measurements on indices where the committed $\theta_i = 1$. In particular, we make the delete phase interactive: the receiver will first commit to all measurement outcomes in Hadamard bases, $\cC^*$ will then decommit to $\theta$, and then finally the receiver will {\em only} open the committed measurement outcomes on indices $i$ where $\theta_i = 1$.
Against malicious receivers, we prove that this scheme is computationally hiding before deletion, and is certified everlasting hiding after deletion. 
Against a malicious committer, we prove statistical simulation-based security before deletion, and show that computational simulation-based security holds \emph{even after deletion}.\\
    
    \noindent{\bf Step 2: Ideal Commitments with EST.} Next, we upgrade the one-sided ideal commitments with EST obtained above to build (full-fledged) ideal commitments with EST. Recall that the one-sided ideal commitments with EST do not satisfy simulation-based security against malicious receivers. Intuitively,  simulation-based security against malicious receivers requires the existence of a simulator that interacts with a malicious receiver to produce a state in the commit phase, that can later be opened (or {\em equivocated}) to a bit that is only revealed to the simulator at the end of the commit phase. We show that this property can be generically obtained (with EST) by relying on a previous compiler, namely an {\em equivocality compiler} from~\cite{BCKM2021}. 
    We defer additional details of this step to~\cref{subsec:ideal} since this essentially follows from ideas in prior work~\cite{BCKM2021}.
    This also completes an overview of our techniques.

\ifsubmission
\else
\paragraph{Roadmap.} We refer the reader to \cref{sec:main} for the proof of our main theorem, \cref{sec:app} for a variety of encryption and commitment schemes with everlasting security, and \cref{sec:2pc-est} for details on building secure computation with everlasting security transfer.
\fi

\subsection{Concurrent and independent work}\label{subsec:concurrent}

Subsequent to the original posting of our paper on arXiv, an updated version of \cite{cryptoeprint:2022:295} was posted with some independent new results on fully-homomorphic encryption with certified deletion. The updated FHE scheme with certified deletion is shown to satisfy standard semantic security, but under a newly introduced conjecture that a particular hash function is ``strong Gaussian-collapsing''. Proving this conjecture based on a standard assumption such as LWE is left as an open problem in \cite{cryptoeprint:2022:295}. Thus, the FHE scheme presented in our paper is the first to satisfy certified deletion based on a standard assumption (and in addition satisfies \emph{everlasting hiding}). On the other hand, the updated scheme of \cite{cryptoeprint:2022:295} also satisfies the property of publicly-verifiable deletion, which we do not consider in this work.

Also, a concurrent and independent work of Hiroka et al. \cite{cryptoeprint:2022/969} was posted shortly after the original posting of our paper. In \cite{cryptoeprint:2022/969}, the authors construct public-key encryption schemes satisfying the definition of security that we use in this paper: certified everlasting security. However, their constructions are either in the \emph{quantum random oracle model}, or require a \emph{quantum} certificate of deletion. Thus, our construction of PKE with certified everlasting security, which is simple, in the plain model, and has a classical certificate of deletion, subsumes these results. On the other hand,~\cite{cryptoeprint:2022/969} introduce and construct the primitive of (bounded-collusion) \emph{functional encryption} with certified deletion, which we do not consider in this work.

\section{Preliminaries}

Let $\secp$ denote the security parameter. We write $\negl(\cdot)$ to denote any \emph{negligible} function, which is a function $f$ such that for every constant $c \in \mathbb{N}$ there exists $N \in \mathbb{N}$ such that for all $n > N$, $f(n) < n^{-c}$.

Given an alphabet $A$ and string $x \in A^n$, let $h(x)$ denote the Hamming weight (number of non-zero indices) of $x$, and $\omega(x) \coloneqq h(x)/n$ denote the \emph{relative} Hamming weight of $x$. Given two strings $x,y \in \{0,1\}^n$, let $\Delta(x,y) \coloneqq \omega(x \oplus y)$ denote the \emph{relative Hamming distance} between $x$ and $y$.

\subsection{Quantum preliminaries}

A register $\sX$ is a named Hilbert space $\bbC^{2^n}$. A pure quantum state on register $\sX$ is a unit vector $\ket{\psi}^{\sX} \in \bbC^{2^n}$, and we say that $\ket{\psi}^{\sX}$ consists of $n$ qubits. A mixed state on register $\sX$ is described by a density matrix $\rho^{\sX} \in \bbC^{2^n \times 2^n}$, which is a positive semi-definite Hermitian operator with trace 1. 

A \emph{quantum operation} $F$ is a completely-positive trace-preserving (CPTP) map from a register $\sX$ to a register $\sY$, which in general may have different dimensions. That is, on input a density matrix $\rho^{\sX}$, the operation $F$ produces $F(\rho^{\sX}) = \tau^{\sY}$ a mixed state on register $\sY$. We will sometimes write a quantum operation $F$ applied to a state on register $\sX$ and resulting in a state on register $\sY$ as $\sY \gets F(\sX)$. Note that we have left the actual mixed states on these registers implicit in this notation, and just work with the names of the registers themselves.

A \emph{unitary} $U: \sX \to \sX$ is a special case of a quantum operation that satisfies $U^\dagger U = U U^\dagger = \bbI^{\sX}$, where $\bbI^{\sX}$ is the identity matrix on register $\sX$. A \emph{projector} $\Pi$ is a Hermitian operator such that $\Pi^2 = \Pi$, and a \emph{projective measurement} is a collection of projectors $\{\Pi_i\}_i$ such that $\sum_i \Pi_i = \bbI$.

Let $\Tr$ denote the trace operator. For registers $\sX,\sY$, the \emph{partial trace} $\Tr^{\sY}$ is the unique operation from $\sX,\sY$ to $\sX$ such that for all $(\rho,\tau)^{\sX,\sY}$, $\Tr^{\sY}(\rho,\tau) = \Tr(\tau)\rho$. The \emph{trace distance} between states $\rho,\tau$, denoted $\TD(\rho,\tau)$ is defined as \[\TD(\rho,\tau) \coloneqq \frac{1}{2}\|\rho-\tau\|_1 \coloneqq \frac{1}{2}\Tr\left(\sqrt{(\rho-\tau)^\dagger(\rho-\tau)}\right).\] We will often use the fact that the trace distance between two states $\rho$ and $\tau$ is an upper bound on the probability that any (unbounded) algorithm can distinguish $\rho$ and $\tau$. When clear from context, we will write $\TD(\sX,\sY)$ to refer to the trace distance between a state on register $\sX$ and a state on register $\sY$.

\begin{lemma}[Gentle measurement \cite{DBLP:journals/tit/Winter99}]\label{lemma:gentle-measurement}
Let $\rho^{\sX}$ be a quantum state and let $(\Pi,\bbI-\Pi)$ be a projective measurement on $\sX$ such that $\Tr(\Pi\rho) \geq 1-\delta$. Let \[\rho' = \frac{\Pi\rho\Pi}{\Tr(\Pi\rho)}\] be the state after applying $(\Pi,\bbI-\Pi)$ to $\rho$ and post-selecting on obtaining the first outcome. Then, $\TD(\rho,\rho') \leq 2\sqrt{\delta}$.
\end{lemma}

We will make use of the convention that $0$ denotes the computational basis $\{\ket{0},\ket{1}\}$ and $1$ denotes the Hadamard basis $\left\{\frac{\ket{0} + \ket{1}}{\sqrt{2}}, \frac{\ket{0} - \ket{1}}{\sqrt{2}}\right\}$. For a bit $r \in \{0,1\}$, we write $\ket{r}_0$ to denote $r$ encoded in the computational basis, and $\ket{r}_1$ to denote $r$ encoded in the Hadamard basis. For strings $x,\theta \in \{0,1\}^\secp$, we write $\ket{x}_\theta$ to mean $\ket{x_1}_{\theta_1},\dots,\ket{x_\secp}_{\theta_\secp}$.

A non-uniform quantum polynomial-time (QPT) machine $\{\cA_\secp,\ket{\psi}_\secp\}_{\secp \in \bbN}$ is a family of polynomial-size quantum machines $\cA_\secp$, where each is initialized with a polynomial-size advice state $\ket{\psi_\secp}$. Each $\cA_\secp$ is in general described by a CPTP map. Similar to above, when we write $\sY \gets \cA(\sX)$, we mean that the machine $\cA$ takes as input a state on register $\sX$ and produces as output a state on register $\sY$, and we leave the actual descripions of these states implicit. Finally, a quantum \emph{interactive} machine is simply a sequence of quantum operations, with designated input, output, and work registers.

\subsection{The XOR extractor}

We make use of a result from \cite{ABKK} which shows that the XOR function is a good randomness extractor from certain \emph{quantum} sources of entropy, even given quantum side information. We include a proof here for completeness.

\begin{importedtheorem}[\cite{ABKK}]\label{thm:XOR-extractor}
Let $\sX$ be an $n$-qubit register, and consider any quantum state $\ket{\gamma}^{\sA,\sX}$ that can be written as \[\ket{\gamma}^{\sA,\sX} = \sum_{u: h(u) < n/2} \ket{\psi_u}^{\sA} \otimes \ket{u}^{\sX},\] where $h(\cdot)$ denotes the Hamming weight. Let $\rho^{\sA,\sP}$ be the mixed state that results from measuring $\sX$ in the Hadamard basis to produce a string $x \in \{0,1\}^n$, and writing $\bigoplus_{i \in [n]}x_i$ into a single qubit register $\sP$. Then it holds that \[\rho^{\sA,\sP} = \Tr^{\sX}(\ket{\gamma}\bra{\gamma}) \otimes \left(\frac{1}{2}\ket{0}\bra{0} + \frac{1}{2}\ket{1}\bra{1}\right).\] 
\end{importedtheorem}

\begin{proof}
First, write the state on registers $\sA,\sX,\sP$ that results from applying Hadamard to $\sX$ and writing the parity, denoted by $p(x) \coloneqq \bigoplus_{i \in [n]}x_i$, to $\sP$:
\[\frac{1}{2^{n/2}}\sum_{x \in \{0,1\}^n}\left(\sum_{u:h(u) < n/2}(-1)^{u \cdot x}\ket{\psi_u}^{\sA}\right)\ket{x}^{\sX}\ket{p(x)}^{\sP} \coloneqq \frac{1}{2^{n/2}}\sum_{x \in \{0,1\}^n}\ket{\phi_x}^{\sA}\ket{x}^{\sX}\ket{p(x)}^{\sP}.\]
Then, tracing out the register $\sX$, we have that
\begin{align*}
    \rho^{\sA,\sP} &= \frac{1}{2^n}\sum_{x \in \{0,1\}^n}\ket{\phi_x}\ket{p(x)}\bra{p(x)}\bra{\phi_x} \\
    &= \frac{1}{2^n}\sum_{x: p(x) = 0}\ket{\phi_x}\bra{\phi_x} \otimes \ket{0}\bra{0} + \frac{1}{2^n}\sum_{x: p(x) = 1}\ket{\phi_x}\bra{\phi_x} \otimes \ket{1}\bra{1}\\
    &= \frac{1}{2^n}\sum_{x: p(x) = 0} \left(\sum_{u_1,u_2: h(u_1),h(u_2) < n/2}(-1)^{(u_1 \oplus u_2)\cdot x}\ket{\psi_{u_1}}\bra{\psi_{u_2}}\right) \otimes \ket{0}\bra{0}\\ &\ \ \ + \frac{1}{2^n}\sum_{x: p(x) = 1} \left(\sum_{u_1,u_2: h(u_1),h(u_2) < n/2}(-1)^{(u_1 \oplus u_2)\cdot x}\ket{\psi_{u_1}}\bra{\psi_{u_2}}\right) \otimes \ket{1}\bra{1} \\
    &= \sum_{u_1,u_2: h(u_1),h(u_2) < n/2}\ket{\psi_{u_1}}\bra{\psi_{u_2}} \otimes \left(\frac{1}{2^n}\sum_{x:p(x)=0}(-1)^{(u_1 \oplus u_2)\cdot x}\ket{0}\bra{0} + \frac{1}{2^n}\sum_{x:p(x)=1}(-1)^{(u_1 \oplus u_2)\cdot x}\ket{1}\bra{1}\right) \\ &= \sum_{u: h(u) < n/2} \ket{\psi_{u}}\bra{\psi_{u}} \otimes \left(\frac{1}{2}\ket{0}\bra{0} + \frac{1}{2}\ket{1}\bra{1}\right) \\ &= \Tr^{\sX}(\ket{\gamma}\bra{\gamma}) \otimes \left(\frac{1}{2}\ket{0}\bra{0} + \frac{1}{2}\ket{1}\bra{1}\right),
\end{align*}
where the 5th equality is due to the following claim, plus the observation that $u_1 \oplus u_2 \neq 1^n$ for any $u_1,u_2$ such that $h(u_1) < n/2$ and $h(u_2) < n/2$.

\begin{claim}
For any $u \in \{0,1\}^n$ such that $u \notin \{0^n,1^n\}$, it holds that \[\sum_{x:p(x)=0}(-1)^{u\cdot x} = \sum_{x:p(x)=1}(-1)^{u\cdot x} = 0.\]
\end{claim}

\begin{proof}
For any such $u \notin \{0^n,1^n\}$, define $S_0 = \{i : u_i = 0\}$ and $S_1 = \{i : u_i = 1\}$. Then, for any $y_0 \in \{0,1\}^{|S_0|}$ and $y_1 \in \{0,1\}^{|S_1|}$, define $x_{y_0,y_1} \in \{0,1\}^n$ to be the $n$-bit string that is equal to $y_0$ when restricted to indices in $S_0$ and equal to $y_1$ when restricted to indices in $S_1$. Then,

\begin{align*}
    &\sum_{x:p(x)=0}(-1)^{u \cdot x} = \sum_{y_1 \in \{0,1\}^{|S_1|}}\sum_{y_0 \in \{0,1\}^{|S_0|}: p(x_{y_0,y_1})=0} (-1)^{u \cdot x_{y_0,y_1}}\\ &= \sum_{y_1 \in \{0,1\}^{|S_1|}}2^{|S_0|-1}(-1)^{1^{|S_1|} \cdot y_1} = 2^{|S_0|-1}\sum_{y_1 \in \{0,1\}^{|S_1|}}(-1)^{p(y_1)} = 0,
\end{align*}
where the second equality can be seen to hold by noting that for any fixed $y_1 \in \{0,1\}^{|S_1|}$, there are exactly  $2^{|S_0|-1}$ strings $y_0 \in \{0,1\}^{|S_0|}$ such that the parity of $x_{y_0,y_1}$ is 0. Finally, the same sequence of equalities can be seen to hold for $x : p(x) = 1$.
\end{proof}
\end{proof}

\subsection{Quantum rewinding}
We will make use of the following lemma from~\cite{STOC:Watrous06}.

\begin{lemma}\label{lem:qrl} Let $\cQ$ be a quantum circuit that takes $n$ qubits as input and outputs a classical bit $b$ and $m$ qubits. For an $n$-qubit state $\ket{\psi}$, let $p(\ket{\psi})$ denote the probability that $b = 0$ when executing $\cQ$ on input $\ket{\psi}$. Let $p_0, q \in (0,1)$ and $\epsilon \in (0,1/2)$ be such that:

\begin{itemize}
    \item For every $n$-qubit state $\ket{\psi},p_0 \leq p(\ket{\psi})$,
    \item For every $n$-qubit state $\ket{\psi}$, $|p(\ket{\psi})-q| < \epsilon$,
    \item $p_0(1-p_0) \leq q(1-q)$,
\end{itemize}

Then, there is a quantum circuit $\widehat{\cQ}$ of size $O\left(\frac{\log(1/\epsilon)}{4 \cdot p_0(1-p_0)} |\cQ|\right)$, taking as input $n$ qubits, and returning as output $m$ qubits, with the following guarantee. For an $n$ qubit state $\ket{\psi}$, let $\cQ_0(\ket{\psi})$ denote the output of $\cQ$ on input $\ket{\psi}$ conditioned on $b=0$, and let $\widehat{\cQ}(\ket{\psi})$ denote the output of $\widehat{\cQ}$ on input $\ket{\psi}$. Then, for any $n$-qubit state $\ket{\psi}$,

$$\mathsf{TD}\left(\cQ_0(\ket{\psi}),\widehat{\cQ}(\ket{\psi})\right) \leq 4\sqrt{\epsilon}\frac{\log(1/\epsilon)}{p_0(1-p_0)}.$$

\end{lemma}
\section{Main theorem}
\label{sec:main}

\begin{theorem}\label{thm:main}
Let $\{\cZ_\secp(\cdot,\cdot,\cdot)\}_{\secp \in \bbN}$ be a quantum operation with three arguments: a $\secp$-bit string $\theta$, a bit $b'$, and a $\secp$-bit quantum register $\sA$. Let $\mathscr{A}$ be a class of adversaries\footnote{Technically, we require that for any  $\{\cA_\secp\}_{\secp \in \bbN} \in \mathscr{A}$, every adversary $\cB$ with time and space complexity that is linear in $\secp$ more than that of $\cA_\secp$, is also in $\mathscr{A}$.} such that for all $\{\cA_\secp\}_{\secp \in \bbN} \in \mathscr{A}$, and for any string $\theta \in \{0,1\}^\secp$, bit $b' \in \{0,1\}$, and state $\ket{\psi}^{\sA,\sC}$ on $\secp$-bit register $\sA$ and arbitrary size register $\sC$,

\[\bigg|\Pr[\cA_{\secp}(\cZ_\secp(\theta,b',\sA),\sC) = 1] - \Pr[\cA_{\secp}(\cZ_\secp(0^\secp,b',\sA),\sC)=1] \bigg| = \negl(\secp).\] That is, $\cZ_\secp$ is semantically-secure against $\cA_\secp$ with respect to its first input. For any $\{\cA_\secp\}_{\secp \in \bbN} \in \mathscr{A}$, consider the following distribution $\left\{\widetilde{\cZ}_\secp^{\cA_\secp}(b)\right\}_{\secp \in \bbN, b \in \{0,1\}}$ over quantum states, obtained by running $\mathcal{A}_\lambda$ as follows.


\begin{itemize}
    \item Sample $x,\theta \gets \{0,1\}^\secp$ and initialize $\cA_\secp$ with \[\cZ_\secp\left(\theta, b \oplus \bigoplus_{i:\theta_i = 0}x_i,\ket{x}_\theta\right).\]
    \item $\cA_\secp$'s output is parsed as a string $x' \in \{0,1\}^\secp$ and a residual state on register $\sA'$.
    \item If $x_i = x_i'$ for all $i$ such that $\theta_i = 1$ then output $\sA'$, and otherwise output a special symbol $\bot$.
\end{itemize}

Then, 

\[\TD\left(\widetilde{\cZ}_\secp^{\cA_\secp}(0),\widetilde{\cZ}_\secp^{\cA_\secp}(1)\right) = \negl(\secp).\]
\end{theorem}

\begin{remark}
We note that, in fact, the above theorem is true as long as $x,\theta$ are $\omega(\mathsf{log} \secp)$ bits long.
\end{remark}






\begin{proof} We define a sequence of hybrid distributions.

\begin{itemize}
    \item $\Hyb_0(b):$ This is the distribution $\left\{\widetilde{\cZ}_\secp^{\cA_\secp}(b)\right\}_{\secp \in \bbN}$ described above.
    
    \item $\Hyb_1(b):$ This distribution is sampled as follows.
    \begin{itemize}
        \item Prepare $\secp$ EPR pairs $\frac{1}{\sqrt{2}}(\ket{00} + \ket{11})$ on registers $(\sC_1,\sA_1),\dots,(\sC_\secp,\sA_\secp)$. Define $\sC \coloneqq \sC_1,\dots,\sC_\secp$ and $\sA \coloneqq \sA_1,\dots,\sA_\secp$.
        \item Sample $\theta \gets \{0,1\}^\secp, b' \gets \{0,1\}$, measure register $\sC$ in basis $\theta$ to obtain $x \in \{0,1\}^\secp$, and initialize $\cA_\secp$ with $\cZ_\secp(\theta,b',\sA)$.
        \item If $b' = b \oplus \bigoplus_{i: \theta_i = 0} x_i$ then proceed as in $\Hyb_0$ and otherwise output $\bot$.
    \end{itemize}
    \item $\Hyb_2(b):$ This is the same as $\Hyb_1(b)$ except that measurement of register $\sC$ to obtain $x$ is performed after $\cA_\secp$ outputs $x'$ and $\rho$.
\end{itemize}

We define $\Advt(\Hyb_i) \coloneqq \TD\left(\Hyb_i(0),\Hyb_i(1)\right).$ Then, we have that \[\Advt(\Hyb_1) \geq \Advt(\Hyb_0)/2,\] which follows because $\Hyb_1(b)$ is identically distributed to the distribution that outputs $\bot$ with probability 1/2 and otherwise outputs $\Hyb_0(b)$. Next, we have that \[\Advt(\Hyb_2) = \Advt(\Hyb_1),\] which follows because the register $\sC$ is disjoint from the registers that $\cA_\secp$ operates on. Thus, it remains to show that \[\Advt(\Hyb_2) = \negl(\secp).\]

To show this, we first define the following hybrid.

\begin{itemize}
    \item $\Hyb_2'(b):$ This is the same as $\Hyb_2$ except that $\cA_\secp$ is initialized with $\cZ_\secp(0^\secp,b',\sA)$.
\end{itemize}

Now, for any $b \in \{0,1\}$, consider the state on register $\sC$ immediately after $\cA_\secp$ outputs $(x',\sA')$ in $\Hyb'_2(b)$. For any $\theta \in \{0,1\}^\secp$, define sets $\theta_0 \coloneqq \{i : \theta_i =0\}$ and $\theta_1 \coloneqq \{i : \theta_i =1\}$, and define the projector \[\Pi_{x',\theta} \coloneqq \left(H^{\otimes |\theta_1|}\ket{x'_{\theta_1}}\bra{x'_{\theta_1}}H^{\otimes |\theta_1|}\right)^{\sC_{\theta_1}} \otimes \sum_{\substack{y \in \{0,1\}^{|\theta_0|} \text{ s.t.}\\ \Delta\left(y,x'_{\theta_0}\right) \geq 1/2}}\left(H^{\otimes |\theta_0|}\ket{y}\bra{y}H^{\otimes |\theta_0|}\right)^{\sC_{\theta_0}},\] where $\Delta(\cdot,\cdot)$ denotes relative Hamming distance. Then, let $\Pr[\Pi_{x',\theta},\Hyb_2'(b)]$ be the probability that a measurement of $\left\{\Pi_{x',\theta}, \bbI - \Pi_{x',\theta}\right\}$ accepts (returns the outcome associated with $\Pi_{x',\theta}$) in $\Hyb_2'(b)$.

\begin{claim}\label{claim:BF10}
For any $b \in \{0,1\}$, $\Pr[\Pi_{x',\theta},\Hyb_2'(b)] = \negl(\secp)$.
\end{claim}

\begin{proof}
Consider running $\Hyb_2'(b)$ until $\cA_\secp$ outputs $x'$ and a state on register $\sA'$ that may be entangled with the challenger's state on register $\sC$. Note that we can sample $\theta \gets \{0,1\}^\secp$ independently since it is no longer in $\cA_\secp$'s view. Then since $\Pi_{x',\theta}$ is diagonal in the Hadamard basis for any $(x',\theta)$, we have that

\[\Pr[\Pi_{x',\theta},\Hyb_2'(b)] = \Pr_{x',\theta,y}\left[y_{\theta_1} = x'_{\theta_1} \wedge \Delta\left(y_{\theta_0},x'_{\theta_0}\right) \geq 1/2\right],\]

where the second probability is over $\cA_\secp$ outputting $x'$, the challenger sampling $\theta \gets \{0,1\}^\secp$, and the challenger measuring register $\sC$ in the Hadamard basis to obtain $y$. For any fixed string $x'$, this probability can be bound by standard Hoeffding inequalities. For example, in \cite[Appendix B.3]{C:BouFeh10}, it is shown to be bounded by $4e^{-\secp(1/2)^2/32} = \negl(\secp)$, which completes the proof.

\end{proof}

Now we consider the corresponding event in $\Hyb_2(b)$, denoted $\Pr[\Pi_{x',\theta},\Hyb_2(b)]$.

\begin{claim}\label{claim:predicate}
For any $b \in \{0,1\}$, $\Pr[\Pi_{x',\theta},\Hyb_2(b)] = \negl(\secp)$.
\end{claim}

\begin{proof}
This follows by a direct reduction to semantic security of $\{\cZ_\secp(\cdot,\cdot,\cdot)\}_{\secp \in \bbN}$ with respect to its first input. The reduction samples $\theta \gets \{0,1\}^\secp$, $b' \gets \{0,1\}$, prepares $\secp$ EPR pairs on registers $(\sA,\sC)$, and sends $(\theta,b',\sA)$ to its challenger. It receives either $\cZ_\secp(\theta,b',\sA)$ or $\cZ_\secp(0^\secp,b',\sA)$, which its sends to $\cA_{\secp}$. After $\cA_{\secp}$ outputs $(x',\sA')$, the reduction measures $\left\{\Pi_{x',\theta},\bbI - \Pi_{x',\theta}\right\}$ on register $\sC$. Note that the complexity of this reduction is equal to the complexity of $\cA_{\secp}$ plus an extra $\secp$ bits of space and an extra linear time operation, so it is still in $\mathscr{A}$. If $\Pr[\Pi_{x',\theta},\Hyb_2(b)]$ is non-negligible this can be used to distinguish $\cZ_\secp(\theta,b',\sA)$ from $\cZ_\secp(0^\secp,b',\sA)$, due to \cref{claim:BF10}.
\end{proof}

Finally, we can show the following claim, which completes the proof.

\begin{claim}
$\Advt(\Hyb_2) = \negl(\secp).$
\end{claim}

\begin{proof}
First, we note that for any $b \in \{0,1\}$, the global state of $\Hyb_2(b)$ immediately after $\cA_\secp$ outputs $x'$ is within negligible trace distance of a state $\tau^{\sC,\sA'}_{\mathsf{Ideal}}$ in the image of $\bbI - \Pi_{x',\theta}$. This follows immediately from \cref{claim:predicate} and Gentle Measurement (\cref{lemma:gentle-measurement}). Now, consider measuring registers $\sC_{\theta_1}$ of $\tau_{\mathsf{Ideal}}^{\sC,\sA'}$ to determine whether the experiment outputs $\bot$. That is, the procedure measures $\sC_{\theta_1}$ in the Hadamard basis and checks if the resulting string is equal to $x'_{\theta_1}$. There are two options.

\begin{itemize}
    \item If the measurement fails, then the experiment outputs $\bot$, independent of whether $b=0$ or $b=1$, so there is 0 advantage in this case.
    \item If the measurement succeeds, then we know that the state on register $\sC_{\theta_0}$ is only supported on vectors $H^{\otimes |\theta_0|}\ket{y}$ such that $\Delta(y,x'_{\theta_0}) < 1/2$, since $\tau_{\mathsf{Ideal}}^{\sC,\sA'}$ was in the image of $\bbI - \Pi_{x',\theta}$. These registers are then measured in the computational basis to produce bits $\{x_i\}_{i: \theta_i = 0}$, and the experiment outputs $\bot$ if $\bigoplus_{i:\theta_i = 0}x_i \neq b' \oplus b$ and otherwise outputs the state on register $\sA'$. Note that (i) this decision is the \emph{only} part of the experiment that depends on $b$, and (ii) it follows from \cref{thm:XOR-extractor} that the bit  $\bigoplus_{i:\theta_i = 0}x_i$ is \emph{uniformly random and independent} of the register $\sA'$, which is disjoint (but possibly entangled with) $\sC$. Thus, there is also 0 advantage in this case. 
    
    Indeed, \cref{thm:XOR-extractor} says that making a Hadamard basis measurement of a register that is in a superposition of computational basis vectors with relative Hamming weight $< 1/2$ will produce a set of bits $\{x_i\}_{i: \theta_i = 0}$ such that $\bigoplus_{i: \theta_i = 0}x_i$ is a uniformly random bit, even given potentially entangled quantum side information. We can apply this lemma to our system on $\sC_{\theta_0},\sA'$ by considering a change of basis that maps $H^{\otimes |\theta_0|}\ket{x'_{\theta_0}} \to \ket{0^{|\theta_0|}}$. That is, the change of basis first applies Hadamard gates, and then an XOR with the fixed string $x'_{\theta_0}$. Applying such a change of basis maps $\sC_{\theta_0}$ to a state that is supported on vectors $\ket{y}$ such that $\omega(y) < 1/2$, and we want to claim that a \emph{Hadamard} basis measurement of the resulting state produces $\{x_i\}_{i: \theta_i = 0}$ such that $\bigoplus_{i: \theta_i = 0}x_i$ is uniformly random and independent of $\sA'$. This is exactly the statement of \cref{thm:XOR-extractor}.
\end{itemize}

This completes the proof, since we have shown that there exists a single distribution, defined by $\tau^{\sC,\sA'}_{\mathsf{Ideal}}$, that is negligibly close to both $\Hyb_2(0)$ and $\Hyb_2(1)$.
\end{proof}

\end{proof}

\section{Cryptography with Certified Everlasting Security}
\label{sec:app}

\subsection{Secret sharing}\label{subsec:ss}

We give a simple construction of a 2-out-of-2 secret sharing scheme where there exists a designated party that the dealer can ask to produce a certificate of deletion of their share. If this certificate verifies, then the underlying plaintext is information theoretically deleted, even given the other share.

\paragraph{Definition.}
First, we augment the standard syntax of secret sharing to include a deletion algorithm $\Del$ and a verification algorithm $\Ver$. Formally, consider a secret sharing scheme $\CD$-$\SSS = (\Share,\Rec,\Del,\Ver)$ with the following syntax.

\begin{itemize}
    \item $\Share(m) \to (s_1, s_2, \vk)$ is a quantum algorithm that takes as input a classical message $m$, and outputs a quantum share $s_1$, a classical share $s_2$ and a (potentially quantum) verification key $\vk$.
    \item $\Rec(s_1,s_2) \to \{m,\bot\}$ is a quantum algorithm that takes as input two shares and outputs either a message $m$ or a $\bot$ symbol.
    \item $\Del(s_1) \to \cert$ is a quantum algorithm that takes as input a quantum share $s_1$ and outputs a (potentially quantum) deletion certificate $\cert$.
    \item $\Ver(\vk,\cert) \to \{\top,\bot\}$ is a (potentially quantum) algorithm that takes as input a (potentially quantum) verification key $\vk$ and a (potentially quantum) deletion certificate $\cert$ and outputs either $\top$ or $\bot$.
\end{itemize}

We say that $\CD$-$\SSS$ satisfies correctness of deletion if the following holds.

\begin{definition}[Correctness of deletion]\label{def:CD-correctness-SS}
$\CD$-$\SSS = (\Share,\Rec,\Del,\Ver)$ satisfies \emph{correctness of deletion} if for any $m$, it holds with $1-\negl(\secp)$ probability over $(s_1,s_2,\vk) \gets \Share(m),\cert \gets \Del(s_1), \mu \gets \Ver(\vk,\cert)$ that $\mu = \top$.
\end{definition}

Next, we define certified deletion security for a secret sharing scheme.

\begin{definition}[Certified deletion security]\label{def:CD-security}

Let $\cA = \{\cA_\secp\}_{\secp \in \bbN}$ denote an unbounded adversary and $b$ denote a classical bit. Consider experiment $\EVEXP_\secp^{\cA}(b)$ which describes everlasting security given a deletion certificate, and is defined as follows.
\begin{itemize}
    \item Sample $(s_1,s_2,\vk) \gets \Share(b)$.
    \item Initialize $\cA_{\secp}$ with $s_1$.
    \item Parse $\cA_{\secp}$'s output as a deletion certificate $\cert$ and a residual state on register $\sA'$.
    \item If $\Ver(\vk,\cert) = \top$ then output $(\sA', s_2)$, and otherwise output $\bot$.
\end{itemize}
Then $\CD$-$\SSS = (\Share,\Rec,\Del,\Ver)$ satisfies \emph{certified deletion security} if for any unbounded adversary $\cA$, 
it holds that \[\TD\left(\EVEXP_\secp^{\cA}(0),\EVEXP_\secp^{\cA}(1)\right) = \negl(\secp),\]
%
\end{definition}

\begin{corollary}
The scheme $\CD$-$\SSS = (\Share,\Rec,\Del,\Ver)$ defined as follows is a secret sharing scheme with certified deletion.
\begin{itemize}
    \item $\Share(m):$ sample $x,\theta \gets \{0,1\}^\secp$ and output \[s_1 \coloneqq \ket{x}_\theta,
    s_2 \coloneqq  \left(\theta,b \oplus \bigoplus_{i: \theta_i = 0} x_i\right), \ \ \ \vk \coloneqq (x,\theta).\]
    \item $\Rec(s_1, s_2):$ parse $s_1 \coloneqq \ket{x}_\theta, s_2 \coloneqq \left(\theta,b'\right)$, measure $\ket{x}_\theta$ in the $\theta$-basis to obtain $x$, and output $b = b' \oplus \bigoplus_{i:\theta_i = 0}x_i$.
    \item $\Del(s_1):$ parse $s_1 \coloneqq \ket{x}_\theta$ and measure $\ket{x}_\theta$ in the Hadamard basis to obtain a string $x'$, and output $\cert \coloneqq x'$.
    \item $\Ver(\vk,\cert):$ parse $\vk$ as $(x,\theta)$ and $\cert$ as $x'$ and output $\top$ if and only if $x_i = x_i'$ for all $i$ such that $\theta_i = 1$.
\end{itemize}
\end{corollary}
\begin{proof}
    Correctness of deletion follows immediately from the description of the scheme. Certified deletion security, i.e. 
    \[\TD\left(\EVEXP_\secp^{\cA}(0),\EVEXP_\secp^{\cA}(1)\right) = \negl(\secp)\]
    follows by following the proof strategy of \cref{thm:main}. This setting is slightly different than the setting considered in the proof of \cref{thm:main} since here we consider unbounded $\cA_\secp$ that are not given access to $\theta$ while \cref{thm:main} considers bounded $\cA_\secp$ that are given access to an encryption of $\theta$. However, the proof is almost identical, defining hybrids as follows.
    \item $\Hyb_0(b):$ This is the distribution $\left\{\EVEXP_\secp^{\cA_\secp}(b)\right\}_{\secp \in \bbN}$ described above.
    
    \item $\Hyb_1(b):$ This distribution is sampled as follows.
    \begin{itemize}
        \item Prepare $\secp$ EPR pairs $\frac{1}{\sqrt{2}}(\ket{00} + \ket{11})$ on registers $(\sC_1,\sA_1),\dots,(\sC_\secp,\sA_\secp)$. Define $\sC \coloneqq \sC_1,\dots,\sC_\secp$ and $\sA \coloneqq \sA_1,\dots,\sA_\secp$.
        \item Sample $\theta \gets \{0,1\}, b' \gets \{0,1\}$, measure register $\sC$ in basis $\theta$ to obtain $x \in \{0,1\}^\secp$, and initialize $\cA_{\secp}$ with register $\sA$.
        \item If $b' = b \oplus \bigoplus_{i: \theta_i = 0} x_i$ then proceed as in $\Hyb_0$ and otherwise output $\bot$.
    \end{itemize}
    \item $\Hyb_2(b):$ This is the same as $\Hyb_1(b)$ except that measurement of register $\sC$ to obtain $x$ is performed after $\cA_{\secp}$ outputs $x'$ and $\sA'$.
    
    Indistinguishability between these hybrids closely follows the proof of \cref{thm:main}. The key difference is that  $\Hyb_2'(b)$ is identical to $\Hyb_2(b)$ except that $s_2$ is set to $(b',0^\secp)$. Then, $\Pr[\Pi_{x',\theta},\Hyb_2'(b)] = \negl(\secp)$ follows identically to the proof in \cref{thm:main}, whereas 
    $\Pr[\Pi_{x',\theta},\Hyb_2(b)] = \negl(\secp)$ follows because the view of $\cA_{\secp}$ is identical in both hybrids.
    The final claim, that $\Advt(\Hyb_2) = \negl(\secp)$ follows identically to the proof in \cref{thm:main}.
\end{proof}

\begin{remark}[One-time pad encryption]
We observe that the above proof, which considers unbounded $\cA_\secp$ who don't have access to $\theta$ until after they produce a valid deletion certificate, can also be used to establish the security of a simple one-time pad encryption scheme with certified deletion. The encryption of a bit $b$ would be the state $\ket{x}_\theta$ together with a one-time pad encryption $k \oplus b \oplus \bigoplus_{i:\theta_i = 0}x_i$ with key $k \gets \{0,1\}$. The secret key would be $(k,\theta)$. Semantic security follows from the one-time pad, while certified deletion security follows from the above secret-sharing proof. This somewhat simplifies the construction of one-time pad encryption with certified deletion of \cite{10.1007/978-3-030-64381-2_4}, who required a seeded extractor.\end{remark}

\subsection{Public-key encryption}\label{subsec:PKE}

In this section, we define and construct post-quantum public-key encryption with certified deletion for classical messages, assuming the existence of post-quantum public-key encryption for classical messages.

\paragraph{Public-Key encryption with certified deletion.}
First, we augment the standard syntax to include a deletion algorithm $\Del$ and a verification algorithm $\Ver$. Formally, consider a public-key encryption scheme $\CD$-$\PKE = (\Gen,\Enc,\Dec,\Del,\Ver)$ with syntax

\begin{itemize}
    \item $\Gen(1^\secp) \to (\pk,\sk)$ is a classical algorithm that takes as input the security parameter and outputs a public key $\pk$ and secret key $\sk$.
    \item $\Enc(\pk,m) \to (\ct,\vk)$ is a quantum algorithm that takes as input the public key $\pk$ and a message $m$, and outputs a (potentially quantum) verification key $\vk$ and a quantum ciphertext $\ct$.
    \item $\Dec(\sk,\ct) \to \{m,\bot\}$ is a quantum algorithm that takes as input the secret key $\sk$ and a quantum ciphertext $\ct$ and outputs either a message $m$ or a $\bot$ symbol.
    \item $\Del(\ct) \to \cert$ is a quantum algorithm that takes as input a quantum ciphertext $\ct$ and outputs a (potentially quantum) deletion certificate $\cert$.
    \item $\Ver(\vk,\cert) \to \{\top,\bot\}$ is a (potentially quantum) algorithm that takes as input a (potentially quantum) verification key $\vk$ and a (potentially quantum) deletion certificate $\cert$ and outputs either $\top$ or $\bot$.
\end{itemize}

We say that $\CD$-$\PKE$ satisfies correctness of deletion if the following holds.

\begin{definition}[Correctness of deletion]\label{def:CD-correctness}
$\CD$-$\PKE = (\Gen,\Enc,\Dec,\Del,\Ver)$ satisfies \emph{correctness of deletion} if for any $m$, it holds with $1-\negl(\secp)$ probability over $(\pk,\sk) \gets \Gen(1^\secp), (\ct,\vk) \gets \Enc(\pk,m),\cert \gets \Del(\ct), \mu \gets \Ver(\vk,\cert)$ that $\mu = \top$.
\end{definition}

Next, we define certified deletion security. Our definition has multiple parts, which we motivate as follows. The first experiment is the everlasting security experiment, which requires that conditioned on the (computationally bounded) adversary producing a valid deletion certificate, their left-over state is information-theoretically independent of $b$. However, we still want to obtain meaningful guarantees against adversaries that do not produce a valid deletion certificate. That is, we hope for standard semantic security against arbitrarily malicious but computationally bounded adversaries. Since such an adversary can query the ciphertext generator with an arbitrarily computed deletion certificate, we should include this potential interaction in the definition, and require that the response from the ciphertext generator still does not leak any information about $b$.\footnote{One might expect that the everlasting security definition described above already captures this property, since whether the certificate accepts or rejects is included in the output of the experiment. However, this experiment does not include the output of the adversary in the case that the certificate is rejected. So we still need to capture the fact that the \emph{joint} distribution of the final adversarial state and the bit indicating whether the verification passes semantically hides $b$.} Note that, in our constructions, the verification key $\vk$ is actually completely independent of the plaintext $b$, and thus for our schemes this property follows automatically from semantic security.

\begin{definition}[Certified deletion security]\label{def:CD-security}
$\CD$-$\PKE = (\Gen,\Enc,\Dec,\Del,\Ver)$ satisfies \emph{certified deletion security} if for any non-uniform QPT adversary $\cA = \{\cA_\secp,\ket{\psi}_\secp\}_{\secp \in \bbN}$, it holds that \[\TD\left(\EVEXP_\secp^{\cA}(0),\EVEXP_\secp^{\cA}(1)\right) = \negl(\secp),\]
and
\[\bigg|\Pr\left[\CEXP_\secp^{\cA}(0) = 1\right] - \Pr\left[\CEXP_\secp^{\cA}(1) = 1\right]\bigg| = \negl(\secp),\]
where the experiment $\EVEXP_\secp^{\cA}(b)$ considers everlasting security given a deletion certificate, and is defined as follows.
\begin{itemize}
    \item Sample $(\pk,\sk) \gets \Gen(1^\secp)$ and $(\ct,\vk) \gets \Enc(\pk,b)$.
    \item Initialize $\cA_\secp(\ket{\psi_\secp})$ with $\pk$ and $\ct$.
    \item Parse $\cA_\secp$'s output as a deletion certificate $\cert$ and a residual state on register $\sA'$.
    \item If $\Ver(\vk,\cert) = \top$ then output $\sA'$, and otherwise output $\bot$.
\end{itemize}
and the experiment $\CEXP_\secp^{\cA}(b)$ is a strengthening of semantic security, defined as follows.
\begin{itemize}
    \item Sample $(\pk,\sk) \gets \Gen(1^\secp)$ and $(\ct,\vk) \gets \Enc(\pk,b)$.
    \item Initialize $\cA_\secp(\ket{\psi_\secp})$ with $\pk$ and $\ct$.
    \item Parse $\cA_\secp$'s output as a deletion certificate $\cert$ and a residual state on register $\sA'$.
    \item Output $\cA_\secp \left(\sA', \Ver(\vk,\cert) \right)$.
\end{itemize}
\end{definition}

Now we can formally define the notion of public-key encryption with certified deletion.

\begin{definition}[Public-key encryption with certified deletion]
$\CD$-$\PKE = (\Gen,\Enc,\Dec,\Del,\Ver)$ is a secure \emph{public-key encryption scheme with certified deletion} if it satisfies (i) correctness of deletion (\cref{def:CD-correctness}), and (ii) certified deletion security (\cref{def:CD-security}).
\end{definition}

Then, we have the following corollary of \cref{thm:main}. 

\begin{corollary}
Given any post-quantum semantically-secure public-key encryption scheme $\PKE = (\Gen,\Enc,\Dec)$, the scheme $\CD$-$\PKE = (\Gen,\Enc',\Dec',\Del,\Ver)$ defined as follows is a public-key encryption scheme with certified deletion.
\begin{itemize}
    \item $\Enc'(\pk,m):$ sample $x,\theta \gets \{0,1\}^\secp$ and output \[\ct \coloneqq \left(\ket{x}_\theta,\Enc\left(\pk, \left(\theta,b \oplus \bigoplus_{i: \theta_i = 0} x_i\right)\right)\right), \ \ \ \vk \coloneqq (x,\theta).\]
    \item $\Dec'(\sk,\ct):$ parse $\ct \coloneqq \left(\ket{x}_\theta,\ct'\right)$, compute $(\theta,b') \gets \Dec(\sk,\ct')$, measure $\ket{x}_\theta$ in the $\theta$-basis to obtain $x$, and output $b = b' \oplus \bigoplus_{i:\theta_i = 0}x_i$.
    \item $\Del(\ct):$ parse $\ct \coloneqq \left(\ket{x}_\theta,\ct'\right)$ and measure $\ket{x}_\theta$ in the Hadamard basis to obtain a string $x'$, and output $\cert \coloneqq x'$.
    \item $\Ver(\vk,\cert):$ parse $\vk$ as $(x,\theta)$ and $\cert$ as $x'$ and output $\top$ if and only if $x_i = x_i'$ for all $i$ such that $\theta_i = 1$.
\end{itemize}
\end{corollary}

\begin{proof}
Correctness of deletion follows immediately from the description of the scheme. For certified deletion security, we consider the following:
\begin{itemize}
\item First, we observe that 
\[\TD\left(\EVEXP_\secp^{\cA}(0),\EVEXP_\secp^{\cA}(1)\right) = \negl(\secp)\]
follows from \cref{thm:main} and the semantic security of $\PKE$ by setting the distribution $\cZ_\secp(\theta,b',\sA)$ to sample $(\pk,\sk) \gets \Gen(1^\secp)$, and output $(\sA,\Enc(\pk,(\theta,b')))$, and setting the class of adversaries $\mathscr{A}$ to be all non-uniform families of QPT adversaries $\{\cA_\secp,\ket{\psi_\secp}\}_{\secp \in \bbN}$.
\item Next, we observe that \[\bigg|\Pr\left[\CEXP_\secp^{\cA}(0) = 1\right] - \Pr\left[\CEXP_\secp^{\cA}(1) = 1\right]\bigg| = \negl(\secp)\]
follows from the fact that the encryption scheme remains (computationally) semantically secure even when the adversary is given the verification key $x$ corresponding to the challenge ciphertext, since the bit $b$ remains encrypted with $\Enc$.
\end{itemize}
This completes our proof.
\end{proof}
The notion of certified deletion security can be naturally generalized to consider multi-bit messages, as follows.

\begin{definition}[Certified deletion security for multi-bit messages]\label{def:CD-security-multi}
$\CD$-$\PKE = (\Gen,\Enc,\Dec,\Del,\Ver)$ satisfies \emph{certified deletion security} if for any non-uniform QPT adversary $\cA = \{\cA_\secp,\ket{\psi}_\secp\}_{\secp \in \bbN}$, it holds that \[\TD\left(\EVEXP_\secp^{\cA}(0),\EVEXP_\secp^{\cA}(1)\right) = \negl(\secp),\]
and
\[\bigg|\Pr\left[\CEXP_\secp^{\cA}(0) = 1\right] - \Pr\left[\CEXP_\secp^{\cA}(1) = 1\right]\bigg| = \negl(\secp),\]
where the experiment $\EVEXP_\secp^{\cA}(b)$ considers everlasting security given a deletion certificate, and is defined as follows.
\begin{itemize}
    \item Sample $(\pk,\sk) \gets \Gen(1^\secp)$.
    Initialize $\cA_\secp(\ket{\psi_\secp})$ with $\pk$ and parse its output as $(m_0, m_1)$.
    \item Sample $(\ct,\vk) \gets \Enc(\pk,m_b)$.
    \item Run $\cA_\secp$ on input $\ct$ and parse $\cA_\secp$'s output as a deletion certificate $\cert$, and a residual state on register $\sA'$.
    \item If $\Ver(\vk,\cert) = \top$ then output $\sA'$, and otherwise output $\bot$.
\end{itemize}
and the experiment $\CEXP_\secp^{\cA}(b)$ is a strengthening of semantic security, defined as follows.
\begin{itemize}
    \item Sample $(\pk,\sk) \gets \Gen(1^\secp)$.
    Initialize $\cA_\secp(\ket{\psi_\secp})$ with $\pk$ and parse its output as $(m_0, m_1)$.
    \item Sample $(\ct,\vk) \gets \Enc(\pk,m_b)$.
    \item Run $\cA_\secp$ on input $\ct$ and parse $\cA_\secp$'s output as a deletion certificate $\cert$, and a residual state on register $\sA'$.
    \item Output $\cA_\secp \left(\sA', \Ver(\vk,\cert) \right)$.
\end{itemize}
\end{definition}

A folklore method converts any public-key bit encryption scheme to a public-key string encryption scheme, by separately encrypting each bit in the underlying string one-by-one and appending all resulting ciphertexts. Semantic security of the resulting public-key encryption scheme follows by a hybrid argument, where one considers intermediate hybrid experiments that only modify one bit of the underlying plaintext at a time.
We observe that the same transformation from bit encryption to string encryption also preserves certified deletion security, and this follows by a similar hybrid argument. That is, as long as the encryption scheme for bits satisfies certified deletion security for single-bit messages per Definition~\ref{def:CD-security}, 
the resulting scheme for multi-bit messages 
satisfies certified deletion security according to Definition~\ref{def:CD-security-multi}.

\ifsubmission\else
\ifsubmission \section{Additional Primitives with Certified Everlasting Security}\label{sec:additional}\else\fi

\ifsubmission 

\subsection{Attribute-based encryption}
We observe that if the $\PKE$ scheme from \cref{subsec:PKE} is an \emph{attribute-based encryption} scheme, then the scheme with certified deletion that results from the compiler inherits these properties. Thus, we obtain an attribute-based encryption scheme with certified deletion, assuming any standard (post-quantum) attribute-based encryption. The previous work of \cite{10.1007/978-3-030-92062-3_21} also constructs an attribute-based encryption scheme with certified deletion, but under the assumption of (post-quantum) indistinguishability obfuscation. We formalize our construction below.

\else

\paragraph{Attribute-based encryption with certified deletion.}
We observe that if the underlying scheme $\PKE$ is an \emph{attribute-based encryption} scheme, then the scheme with certified deletion that results from the above compiler inherits these properties. Thus, we obtain an attribute-based encryption scheme with certified deletion, assuming any standard (post-quantum) attribute-based encryption. The previous work of \cite{10.1007/978-3-030-92062-3_21} also constructs an attribute-based encryption scheme with certified deletion, but under the assumption of (post-quantum) indistinguishability obfuscation. We formalize our construction below.
\fi

We first describe the syntax of an attribute-based encryption scheme with certified deletion $\CD\text{-}\ABE = (\Gen,\KG,\Enc,\Dec,\Del,\Ver)$. This augments the syntax of an ABE scheme $\ABE = (\Gen,\KG,\Enc,\Dec)$ by adding the $\Del$ and $\Ver$ algorithms. 
Let $p = p(\secp)$ denote a polynomial.
\begin{itemize}
    \item $\Gen(1^\secp) \to (\pk,\msk)$ is a classical algorithm that takes as input the security parameter and outputs a public key $\pk$ and master secret key $\msk$.
    \item $\KG(\msk,P) \to \sk_P$ is a classical key generation algorithm that on input the master secret key and a predicate $P:\{0,1\}^{p(\secp)} \rightarrow \{0,1\}$, outputs a secret key $\sk_P$. 
    \item $\Enc(\pk, X, m) \rightarrow (\ct_X,\vk)$ is a quantum algorithm that on input a message $m$ and an attribute $X$ outputs a (potentially quantum) verification key $\vk$ and quantum ciphertext $\ct_X$.
    \item $\Dec(\sk_P, \ct_X) \to \{m',\bot\}$ on input a secret key $\sk_P$ and a quantum ciphertext $\ct_X$ outputs either a message $m'$ or a $\bot$ symbol.
    \item $\Del(\ct) \to \cert$ is a quantum algorithm that takes as input a quantum ciphertext $\ct$ and outputs a (potentially quantum) deletion certificate $\cert$.
    \item $\Ver(\vk,\cert) \to \{\top,\bot\}$ is a (potentially quantum) algorithm that takes as input a (potentially quantum) verification key $\vk$ and a (potentially quantum) deletion certificate $\cert$ and outputs either $\top$ or $\bot$.
\end{itemize}
Correctness of decryption for $\CD\text{-}\ABE$ is the same as that for $\ABE$. We define correctness of deletion, and certified deletion security for $\CD\text{-}\ABE$ below.

\begin{definition}[Correctness of deletion]
$\CD$-$\ABE = (\Gen,\KG,\Enc,\Dec,\Del,\Ver)$ satisfies \emph{correctness of deletion} if for any $m,X$, it holds with $1-\negl(\secp)$ probability over $(\pk,\msk) \gets \Gen(1^\secp), (\ct,\vk) \gets \Enc(\pk,X,m),\cert \gets \Del(\ct), \mu \gets \Ver(\vk,\cert)$ that $\mu = \top$.
\end{definition}

\begin{definition}[Certified deletion security]
\label{def:CD-sec-ABE}
$\CD$-$\ABE = (\Gen,\KG,\Enc,\Dec,\Del,\Ver)$ satisfies \emph{certified deletion security} if for any non-uniform QPT adversary $\cA = \{\cA_\secp,\ket{\psi}_\secp\}_{\secp \in \bbN}$, it holds that \[\TD\left(\EVEXP_\secp^{\cA}(0),\EVEXP_\secp^{\cA}(1)\right) = \negl(\secp),\] 
and
\[\bigg|\Pr\left[\CEXP_\secp^{\cA}(0) = 1\right] - \Pr\left[\CEXP_\secp^{\cA}(1) = 1\right]\bigg| = \negl(\secp),\] 
where the experiments $\EVEXP_\secp^{\cA}(b)$ and $\CEXP_\secp^{\cA}(b)$ are defined as follows.
\begin{itemize}
    \item Sample $(\pk,\msk) \gets \Gen(1^\secp)$
    and 
    initialize
    $\cA_\secp(\ket{\psi_\secp})$ with $\pk$.
    \item Set $i = 1$.
    \item If $\cA_\secp$ outputs a key query $P_i$, return $\sk_{P_i} \leftarrow \KG(\msk, P_i)$ to $\cA_\secp$ and set $i = i + 1$. This process can be repeated polynomially many times.
    \item If $\cA_\secp$ outputs an attribute $X^*$ where $P_i(X^*) = 0$ for all predicates $P_i$ queried so far, then compute $(\vk,\ct) \gets \Enc(\pk, X^*, b)$, and return $\ct$ to $\cA_\secp$. Else exit and output $\bot$.
    \item If $\cA_\secp$ outputs a key query $P_i$ such that $P_i(X^*) = 0$, return $\sk_{P_i} \leftarrow \KG(\msk, P_i)$ to $\cA_\secp$ (otherwise return $\bot$) and set $i = i + 1$. This process can be repeated polynomially many times.
    \item Parse $\cA_\secp$'s output as a deletion certificate $\cert$ and a residual state on register $\sA'$.
    \item If $\Ver(\vk,\cert) = \top$ then $\EVEXP_\secp^{\cA}(b)$ outputs $\sA'$, and otherwise $\EVEXP_\secp^{\cA}(b)$ outputs $\bot$, and ends.
    \item $\CEXP_\secp^{\cA}(b)$ sends the output $\Ver(\vk,\cert)$ to $\cA_\secp$. Again, upto polynomially many times, $\cA_\secp$ sends key queries $P_i$. For each $i$, if $P_i(X^*) = 0$, return $\sk_{P_i} \leftarrow \KG(\msk, P_i)$ to $\cA_\secp$ (otherwise return $\bot$) and set $i = i + 1$.
    Finally, $\cA_\secp$ generates an output bit, which is set to be the output of $\CEXP_\secp^{\cA}(b)$.
\end{itemize}
\end{definition}

\begin{corollary}
Given any post-quantum attribute-based encryption scheme $\ABE = (\Gen,\KG,\Enc,\Dec)$, the scheme $\CD$-$\ABE = (\Gen,\KG,\Enc',\Dec',\Del,\Ver)$ defined as follows is an attribute-based encryption scheme with certified deletion.
\begin{itemize}
    \item $\Enc'(\pk,X,b):$ sample $x,\theta \gets \{0,1\}^\secp$ and output \[\ct \coloneqq \left(\ket{x}_\theta,\Enc\left(\pk, X, \left(\theta,b \oplus \bigoplus_{i: \theta_i = 0} x_i\right)\right)\right), \ \ \ \vk \coloneqq (x,\theta).\]
    \item $\Dec'(\sk_P,\ct):$ parse $\ct \coloneqq \left(\ket{x}_\theta,\ct'\right)$, compute $(\theta,b') \gets \Dec(\sk_P,\ct')$, measure $\ket{x}_\theta$ in the $\theta$-basis to obtain $x$, and output $b = b' \oplus \bigoplus_{i:\theta_i = 0}x_i$.
    \item $\Del(\ct):$ parse $\ct \coloneqq \left(\ket{x}_\theta,\ct'\right)$ and measure $\ket{x}_\theta$ in the Hadamard basis to obtain a string $x'$, and output $\cert \coloneqq x'$.
    \item $\Ver(\vk,\cert):$ parse $\vk$ as $(x,\theta)$ and $\cert$ as $x'$ and output $\top$ if and only if $x_i = x_i'$ for all $i$ such that $\theta_i = 1$.
\end{itemize}
\end{corollary}
\begin{proof}
Correctness of decryption and deletion follow from the description of the scheme. For certified deletion security, we consider the following:
\begin{itemize}
\item First, we observe that 
\[\TD\left(\EVEXP_\secp^{\cA}(0),\EVEXP_\secp^{\cA}(1)\right) = \negl(\secp)\]
follows from \cref{thm:main} and the semantic security of $\ABE$. To see this, we imagine splitting $\cA_\secp$ into two parts: $\cA_{\secp,0}$ which interacts in the $\ABE$ security game until it obtains its challenge ciphertext and all the keys $\sk_{P_i}$ that it wants, and $\cA_{\secp,1}$ which takes the final state of $\cA_{\secp,0}$ and produces a deletion certificate $\cert$ and final state on register $\sA'$. We set the distribution $\cZ_\secp(\theta,b',\sA)$ to run the ABE security game with $\cA_{\secp,0}(\ket{\psi_\secp})$ where the challenge ciphertext is an encryption of $m = (\theta,b')$,\footnote{Although we have only defined the ABE security game above for challengers that encrypt a single bit $b$, we can also consider the challenger encrypting an arbitrary bit string $m$.} and output $\cA_{\secp,0}$'s final state. Then, we set the class of adversaries $\mathscr{A}$ to include all $\cA_{\secp,1}$, which is the class of all  uniform families of QPT adversaries. By the semantic security of $\ABE$, $\cA_{\secp,1}$ cannot distinguish between $\cZ_\secp(\theta,b',\sA)$ and $\cZ_{\secp}(0^\secp,b',\sA)$, and thus the guarantees of \cref{thm:main} apply.
\item Next, we observe that \[\bigg|\Pr\left[\CEXP_\secp^{\cA}(0) = 1\right] - \Pr\left[\CEXP_\secp^{\cA}(1) = 1\right]\bigg| = \negl(\secp)\]
follows from the fact that the encryption scheme remains semantically secure even when the adversary is given the verification key corresponding to the challenge ciphertext.
\end{itemize}
This completes our proof.
\end{proof}

\begin{remark}
Similarly to the setting of public-key encryption, single-bit certified deletion security for ABE implies multi-bit certified deletion security.
\end{remark}

\paragraph{Relation with  \cite{10.1007/978-3-030-92062-3_21}'s definitions.} 
The definition of certified deletion security for public-key (resp., attribute-based) encryption in \cite{10.1007/978-3-030-92062-3_21} is
different than our definition in two primary respects: (1) it only considers \emph{computationally-bounded} adversaries even after the deletion certificate is computed, and (2) explicitly gives the adversary the secret key $\sk$ after the deletion certificate is computed.

Our definition allows the adversary to be unbounded after deletion, which gives a strong \emph{everlasting security} property. 
Indeed, in Appendix \ref{sec:implication}, we show that our definition implies \cite{10.1007/978-3-030-92062-3_21}'s definition for public-key (attribute-based) encryption schemes.
To see this, we consider any adversary $\cA$ that contradicts \cite{10.1007/978-3-030-92062-3_21}'s notion of security and construct a reduction $\cR$ that contradicts our notion of security. $\cR$ will run $\cA$ on the challenge that it receives from its challenger, and forward the deletion certificate $\cert$ received from $\cA$. $\cR$ will then, in unbounded time, reverse sample a $\sk$ such that $(\pk,\sk)$ is identically distributed to the output of the honest $\Gen$ algorithm. Finally, $\cR$ runs $\cA$ on $\sk$ to obtain $\cA$'s guess for $b$. We can show that the view of $\cA$ produced by such an $\cR$ matches its view in the \cite{10.1007/978-3-030-92062-3_21} challenge, thus the
advantage of $\cR$ in contradicting our definition will match that of $\cA$ in contradicting~\cite{10.1007/978-3-030-92062-3_21}'s definition.

\ifsubmission 
\subsection{Witness encryption}

Next, we observe that if the scheme $\PKE$ is a \emph{witness encryption} scheme, then the scheme with certified deletion that results from the compiler becomes a witness encryption scheme with certified deletion. That is, we compile any (post-quantum) witness encryption into a witness encryption scheme with certified deletion.
Similar to the case of PKE and ABE, we can augment the syntax of any witness encryption scheme to include a deletion algorithm $\mathsf{Del}$ and a verification algorithm $\mathsf{Ver}$.
That is, the scheme consists of algorithms $(\mathsf{Enc}, \mathsf{Dec})$ with syntax and properties identical to standard witness encryption schemes~\cite{STOC:GGSW13}, except where the ciphertexts are potentially quantum, and where the encryption algorithm outputs a (potentially quantum) verification key $\mathsf{vk}$ along with a ciphertext.
$\Ver(\vk,\cert) \to \{\top,\bot\}$ is a (potentially quantum) algorithm that takes as input a (potentially quantum) verification key $\vk$ and a (potentially quantum) deletion certificate $\cert$ and outputs either $\top$ or $\bot$.
$\mathsf{Del}(\mathsf{ct}) \rightarrow \mathsf{cert}$ is a quantum algorithm that on input a quantum ciphertext $\mathsf{ct}$ outputs a (potentially quantum) deletion certificate $\mathsf{cert}$.

\else

\paragraph{Witness encryption for NP with certified deletion.}
Finally, we observe that if the underlying scheme $\PKE$ is a \emph{witness encryption} scheme, then the scheme with certified deletion that results from the above compiler becomes a witness encryption scheme with certified deletion. That is, we compile any (post-quantum) witness encryption into a witness encryption scheme with certified deletion.
Similar to the case of PKE and ABE, we can augment the syntax of any witness encryption scheme to include a deletion algorithm $\mathsf{Del}$ and a verification algorithm $\mathsf{Ver}$.
That is, the scheme consists of algorithms $(\mathsf{Enc}, \mathsf{Dec})$ with syntax and properties identical to standard witness encryption schemes~\cite{STOC:GGSW13}, except where the ciphertexts are potentially quantum, and where the encryption algorithm outputs a (potentially quantum) verification key $\mathsf{vk}$ along with a ciphertext.
$\Ver(\vk,\cert) \to \{\top,\bot\}$ is a (potentially quantum) algorithm that takes as input a (potentially quantum) verification key $\vk$ and a (potentially quantum) deletion certificate $\cert$ and outputs either $\top$ or $\bot$.
$\mathsf{Del}(\mathsf{ct}) \rightarrow \mathsf{cert}$ is a quantum algorithm that on input a quantum ciphertext $\mathsf{ct}$ outputs a (potentially quantum) deletion certificate $\mathsf{cert}$.

\fi

Correctness of decryption is the same as that for (regular) witness encryption. We define correctness of deletion, and certified deletion security for $\CD\text{-}\WE$ below.

\begin{definition}[Correctness of deletion]
$\CD$-$\WE = (\Enc,\Dec,\Del,\Ver)$ satisfies \emph{correctness of deletion} if for every statement $X$ and message $m$, it holds with $1-\negl(\secp)$ probability over $(\ct,\vk) \gets \Enc(X,m),\cert \gets \Del(\ct), \mu \gets \Ver(\vk,\cert)$ that $\mu = \top$.
\end{definition}

\begin{definition}[Certified deletion security]
\label{def:CD-sec-ABE}
$\CD$-$\WE = (\Enc,\Dec,\Del,\Ver)$ satisfies \emph{certified deletion security} if for any non-uniform QPT adversary $\cA = \{\cA_\secp,\ket{\psi}_\secp\}_{\secp \in \bbN}$, there is a negligible function $\negl(\cdot)$ for which it holds that \[\TD\left(\EVEXP_\secp^{\cA}(0),\EVEXP_\secp^{\cA}(1)\right) = \negl(\secp),\] 
and
\[\bigg|\Pr\left[\CEXP_\secp^{\cA}(0) = 1\right] - \Pr\left[\CEXP_\secp^{\cA}(1) = 1\right]\bigg| = \negl(\secp),\] 
where the experiments $\EVEXP_\secp^{\cA}(b)$ and $\CEXP_\secp^{\cA}(b)$ are defined as follows.
Both experiments take an input $b$, and interact with $\cA$ as follows.
\begin{itemize}
    \item Obtain statement $X$, language $\cL$ and messages $(m_0, m_1)$ from $\cA_\secp(\ket{\psi}_\secp)$. If $X \in L$, abort, otherwise continue.
    \item Set $(\ct,\vk) \gets \Enc(X,m_b)$.
    \item Run $\cA_\secp$ on input $\ct$ and parse $\cA_\secp$'s output as a deletion certificate $\cert$ and a residual state on register $\sA'$.
    \item If $\Ver(\vk,\cert) = \top$ then output $\sA'$, and otherwise output $\bot$.
\end{itemize}
and the experiment $\CEXP_\secp^{\cA}(b)$ is a strengthening of semantic security, defined as follows.
\begin{itemize}
    \item Obtain statement $X$, language $\cL$ and messages $(m_0, m_1)$ from $\cA_\secp(\ket{\psi}_\secp)$. If $X \in L$, abort, otherwise continue.
    \item Set $(\ct,\vk) \gets \Enc(X,m_b)$.
    \item Run $\cA_\secp$ on input $\ct$ and parse $\cA_\secp$'s output as a deletion certificate $\cert$ and a residual state on register $\sA'$.
    \item Output $\cA_\secp \left(\sA', \Ver(\vk,\cert) \right)$.
\end{itemize}
\end{definition}

\begin{corollary}
Given any post-quantum semantically-secure witness encryption scheme $\WE = (\Enc,\Dec)$, the scheme $\CD$-$\WE = (\Enc',\Dec',\Del,\Ver)$ defined as follows is a witness encryption scheme with certified deletion.
\begin{itemize}
    \item $\Enc'(X,m):$ sample $x,\theta \gets \{0,1\}^\secp$ and output \[\ct \coloneqq \left(\ket{x}_\theta,\Enc\left(X, \left(\theta,b \oplus \bigoplus_{i: \theta_i = 0} x_i\right)\right)\right), \ \ \ \vk \coloneqq (x,\theta).\]
    \item $\Dec'(W,\ct):$ parse $\ct \coloneqq \left(\ket{x}_\theta,\ct'\right)$, compute $(\theta,b') \gets \Dec(W,\ct')$, measure $\ket{x}_\theta$ in the $\theta$-basis to obtain $x$, and output $b = b' \oplus \bigoplus_{i:\theta_i = 0}x_i$.
    \item $\Del(\ct):$ parse $\ct \coloneqq \left(\ket{x}_\theta,\ct'\right)$ and measure $\ket{x}_\theta$ in the Hadamard basis to obtain a string $x'$, and output $\cert \coloneqq x'$.
    \item $\Ver(\vk,\cert):$ parse $\vk$ as $(x,\theta)$ and $\cert$ as $x'$ and output $\top$ if and only if $x_i = x_i'$ for all $i$ such that $\theta_i = 1$.
\end{itemize}
\end{corollary}

\begin{proof}
Correctness of deletion follows immediately from the description of the scheme. For certified deletion security, we consider the following:
\begin{itemize}
\item First, we observe that 
\[\TD\left(\EVEXP_\secp^{\cA}(0),\EVEXP_\secp^{\cA}(1)\right) = \negl(\secp)\]
follows from \cref{thm:main} and the semantic security of $\WE$ by setting the distribution $\cZ_\secp(\theta,b',\sA)$ to sample $\ct \gets \Enc(X,b)$ and output $(\sA,\Enc(X,(\theta,b')))$, and setting the class of adversaries $\mathscr{A}$ to be all non-uniform families of QPT adversaries $\{\cA_\secp,\ket{\psi_\secp}\}_{\secp \in \bbN}$.
\item Next, we observe that \[\bigg|\Pr\left[\CEXP_\secp^{\cA}(0) = 1\right] - \Pr\left[\CEXP_\secp^{\cA}(1) = 1\right]\bigg| = \negl(\secp)\]
follows from the fact that the witness encryption scheme remains (computationally) semantically secure even when the adversary is given the verification key corresponding to the challenge ciphertext.
\end{itemize}
This completes our proof.
\end{proof}

\subsection{Fully-homomorphic encryption}

Next, we consider the syntax of a \emph{fully-homomorphic encryption} scheme (for classical circuits) with certified deletion. Such a scheme consists of algorithms $\CD$-$\FHE$ = ($\CD$-$\FHE.\Gen$, $\CD$-$\FHE.\Enc$, $\CD$-$\FHE.\Eval$, $\CD$-$\FHE.\Dec$, $\CD$-$\FHE.\Del$, $\CD$-$\FHE.\Ver$) with the same syntax as $\CD$-$\PKE$ (\cref{subsec:PKE}), but including the additional algorithm $\CD$-$\FHE.\Eval$.
\begin{itemize}
    \item $\CD$-$\FHE.\Eval(\pk,C,\ct) \to \widetilde{\ct}:$ On input the public key $\pk$, a classical circuit $C$, and a quantum ciphertext $\ct$, the evaluation algorithm returns a (potentially quantum) evaluated ciphertext $\widetilde{\ct}$.
\end{itemize}

We say that $\CD$-$\FHE$ satisfies evaluation correctness if the following holds.

\begin{definition}[$\CD$-$\FHE$ evaluation correctness]\label{def:eval-correctness}
A $\CD$-$\FHE$ scheme satisfies evaluation correctness if for any message $x$, and all polynomial-size circuits $C$, it holds with $1-\negl(\secp)$ probability over $(\pk,\sk) \gets \CD$-$\FHE.\Gen(1^\secp), \ct \gets \CD$-$\FHE.\Enc(\pk,x)$, $\widetilde{\ct} \gets \CD$-$\FHE.\Eval(\pk,C,\ct), y \gets \CD$-$\FHE.\Dec(\sk,\widetilde{\ct})$ that $y = C(x)$.
\end{definition}

Now we can formally define the notion of fully-homomorphic encryption with certified deletion.

\begin{definition}[Fully-homomorphic encryption with certified deltion]
$\CD$-$\FHE$ = ($\CD$-$\FHE.\Gen$, $\CD$-$\FHE.\Enc$, $\CD$-$\FHE.\Eval$, $\CD$-$\FHE.\Dec$, $\CD$-$\FHE.\Del$, $\CD$-$\FHE.\Ver$) is a secure \emph{fully-homomorphic encryption scheme with certified deletion} if it satisfies 
(i) correctness of deletion (\cref{def:CD-correctness}), (ii) certified deletion security (\cref{def:CD-security}), and (iii) evaluation correctness (\cref{def:eval-correctness}).
\end{definition}

\subsubsection{Blind delegation with certified deletion} 

So far, we have described a $\PKE$ scheme with certified deletion augmented with a procedure that allows for homomorphic evaluation over ciphertexts. A fascinating application for such a scheme, as discussed by \cite{10.1007/978-3-030-64381-2_4,cryptoeprint:2022:295}, is the following. A computationally weak client wishes to use the resources of a powerful server to peform some intensive computation $C$ on their input data $x$. However, they would like to keep $x$ private from the server, and, moreover, they would like to be certain that their data is \emph{deleted} by the server after the computation takes place. Here, by deleted, we mean that the original input $x$ becomes \emph{information-theoretically} hidden from the server after the computation has taken place.

While it is not necessarily clear from the syntax described so far that the server can both compute on \emph{and later} delete the client's input data, we demonstrate, via an interaction pattern described by \cite{cryptoeprint:2022:295}, a protocol that achieves this functionality. We refer to such a protocol as a ``blind delegation with certified deletion'' protocol, and describe it in \proref{fig:private-delegation}.

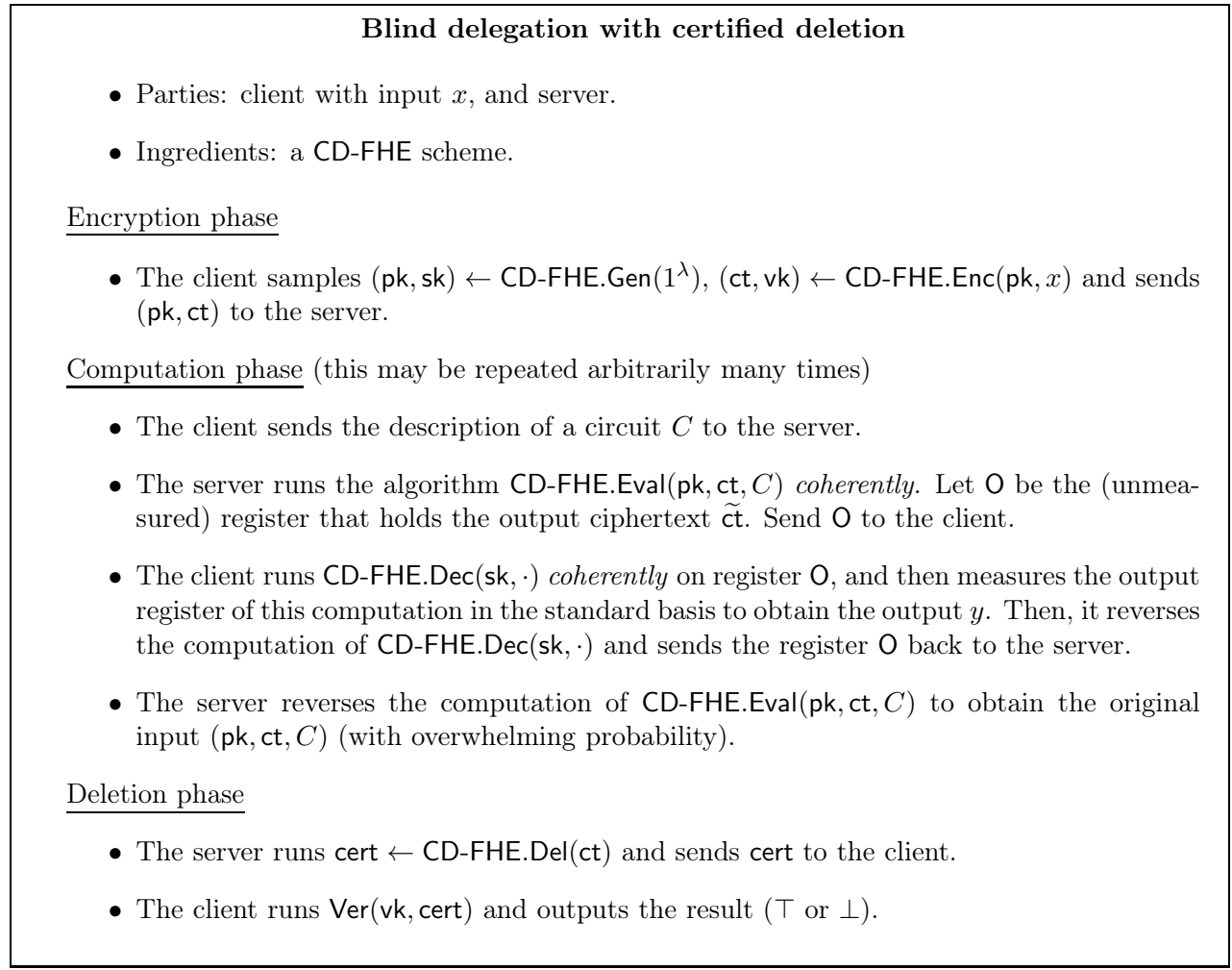
\protocol{Blind delegation with certified deletion}{A generic construction of blind delegation with certified deletion, from any $\CD$-$\FHE$ scheme.}{fig:private-delegation}{
\begin{itemize}
    \item Parties: client with input $x$, and server.
    \item Ingredients: a $\CD$-$\FHE$ scheme.
\end{itemize}

\underline{Encryption phase}
\begin{itemize}
    \item The client samples $(\pk,\sk) \gets \CD$-$\FHE.\Gen(1^\secp)$, $(\ct,\vk) \gets \CD$-$\FHE.\Enc(\pk,x)$ and sends $(\pk,\ct)$ to the server.
\end{itemize}

\underline{Computation phase} (this may be repeated arbitrarily many times)
\begin{itemize}
    \item The client sends the description of a circuit $C$ to the server.
    \item The server runs the algorithm $\CD$-$\FHE.\Eval(\pk,\ct,C)$ \emph{coherently}. Let $\sO$ be the (unmeasured) register that holds the output ciphertext $\widetilde{\ct}$. Send $\sO$ to the client.
    \item The client runs $\CD$-$\FHE.\Dec(\sk,\cdot)$ \emph{coherently} on register $\sO$, and then measures the output register of this computation in the standard basis to obtain the output $y$. Then, it reverses the computation of $\CD$-$\FHE.\Dec(\sk,\cdot)$ and sends the register $\sO$ back to the server.
    \item The server reverses the computation of $\CD$-$\FHE.\Eval(\pk,\ct,C)$ to obtain the original input $(\pk,\ct,C)$ (with overwhelming probability).
\end{itemize}

\underline{Deletion phase}
\begin{itemize}
    \item The server runs $\cert \gets \CD$-$\FHE.\Del(\ct)$ and sends $\cert$ to the client.
    \item The client runs $\Ver(\vk,\cert)$ and outputs the result ($\top$ or $\bot$).
\end{itemize}
}

A blind delegation with certified deletion protocol should satisfy the following notions of correctness and security. We present each definition for the case of a single circuit $C$ queried by the client (one repetition of the computation phase), but they easily extend to considering multiple repetitions of the computation phase.

\begin{definition}[Correctness for blind delegation with certified deletion]\label{def:delegation-correctness}
A blind delegation with certified deletion protocol is \emph{correct} if the honest client and server algorithms satisfy the following properties. First, for any $x,C$, the client obtains $y = C(x)$ after the computation with probability $1-\negl(\secp)$. Second, for any $x,C$, the client outputs $\top$ after the deletion phase with probability $1-\negl(\secp)$.
\end{definition}

\begin{definition}[Security for blind delegation with certified deletion]\label{def:delegation-security}
A blind delegation with certified deletion protocol is \emph{secure} against a class of adversarial servers $\mathscr{S}$ if for any $x_0,x_1$, circuit $C$, and $\{\cS_{\secp}\}_{\secp \in \bbN} \in \mathscr{S}$, the following two properties hold.
\begin{itemize}
    \item \textbf{Privacy:} For any QPT distinguisher $\{\cD_\secp\}_{\secp \in \bbN}$, it holds that \[\bigg|\Pr\left[\cD_\secp( \langle \cC(x_0,C),\cS_\secp \rangle) = 1\right] - \Pr\left[\cD_\secp( \langle \cC(x_1,C),\cS_\secp \rangle) = 1\right]\bigg| = \negl(\secp),\] where $\langle \cC(x,C),\cS_\secp \rangle$ denotes the output state of adversary $\cS_\secp$ after interacting (in the Encryption, Computation, and Delete phases) with an honest client $\cC$ with input $x$ and circuit $C$.

    \item \textbf{Certified deletion:} It holds that \[\TD\left(\EVEXP_\secp^{\cS}(x_0,C),\EVEXP_\secp^{\cS}(x_1,C)\right) = \negl(\secp),\] where the experiment $\EVEXP_\secp^{\cS}(x,C)$ is defined as follows.
    \begin{itemize}
        \item Run the Encryption, Computation, and Deletion phases between client $\cC(x,C)$ and server $\cS_\secp$, obtaining the server's final state on register $\sA'$ and the client's decision $\top$ or $\bot$. If $\top$ output $\sA'$, and otherwise output $\bot$.
    \end{itemize}
\end{itemize}
\end{definition}

Next, we define a class of adversaries $\mathscr{S}$ that we call \emph{evaluation-honest}. The defining feature of an evaluation-honest adversary $\{\cS_\secp\}_{\secp \in \bbN}$ is that, for any client input $x$ and circuit $C$, the state on register $\sO$ returned by $\cS_\secp$ during the Computation phase is within negligible trace distance of the state on register $\sO$ returned by the \emph{honest} server. Otherwise, $\cS_\secp$ may be arbitrarily malicious, including during the Deletion phase and after. Morally, evaluation-honest adversaries are those that are \emph{specious} \cite{C:DupNieSal10} (which is a quantum analogue of semi-honest) during the Computation phase, and \emph{malicious} afterwards, though we do not give a formal definition of specious here.\footnote{Roughly, a specious adversary is one who may, at any step of the computation, apply an operation to their private state such that the joint state of the resulting system is negligibly close to the joint state of an honest interaction. Note that for any specious adversary, the registers they send to the honest party during any round must be negligibly close to the register sent by the honest party during this round, since after transmission, this register is no longer part of their private state.  }

Then, we have the following theorem.

\begin{theorem}
When instantiated with any $\CD$-$\FHE$ scheme, \proref{fig:private-delegation} is a blind delegation with certified deletion scheme, secure against any \emph{evaluation-honest} adversarial server $\{\cS^*_{\secp}\}_{\secp \in \bbN}$. 
\end{theorem}

\begin{proof}
First, we argue that correctness (\cref{def:delegation-correctness}) holds. The evaluation correctness of the underlying $\CD$-$\FHE$ scheme (\cref{def:eval-correctness}) implies that the register measured by the client during the computation phase is within negligible trace distance of $\ket{y}$ for $y = C(x)$. This implies the first property of correctness for blind delegation with certified deletion, and it also implies that the state on register $\sO$ returned to the server is negligibly close to the original state on register $\sO$. So, after reversing the computation, the server obtains a $\ct'$ that is negligibly close to the original $\ct$ received from the client (due to the Gentle Measurement Lemma). Thus, the second property of correctness for blind delegation with certified deletion follows from the correctness of deletion property of $\CD$-$\FHE$ (\cref{def:CD-correctness}).

Next, we argue that security (\cref{def:delegation-security}) holds against any evaluation-honest server $\{\cS^*_\secp\}_{\secp \in \bbN}$. Consider a hybrid experiment in which there is no interaction between client and server during the Computation phase, that is, the register $\sO$ is not touched by the client and is immediately return the the server. By the evaluation-honesty of the server, and the above arguments, it follows that the server's view of this hybrid experiment is negligibly close to its view of the real interaction. But now observe that this hybrid experiment is equivalent to the experiment described in \cref{def:CD-security-multi} for defining certified deletion security for standard public-key encryption for multi-bit messages. Thus, the Privacy and Certified Deletion properties of blind delegation follow directly from the certified deletion properties (\cref{def:CD-security-multi}) of the underlying $\CD$-$\FHE$ scheme.
\end{proof}

\subsubsection{Construction of $\CD$-$\FHE$} 

Now, to obtain a blind delegation with certified deletion protocol, it suffices to construct a $\CD$-$\FHE$ scheme. In this section, we show that such a scheme follows from a standard fully-homorphic encryption scheme, and our main theorem.

\begin{corollary}
Given any classical fully-homomorphic encryption scheme $\FHE = (\FHE.\Gen,\allowbreak\FHE.\Enc,\allowbreak \FHE.\Eval,\allowbreak \FHE.\Dec)$, the scheme defined below $\CD$-$\FHE$ = $(\FHE.\Gen,\Enc,\Eval,\Dec,\allowbreak\Del,\Ver)$ is an $\FHE$ scheme with certified deletion. We define encryption, decryption, deletion, and verification for one-bit plaintexts, and evaluation over ciphertexts encrypting $n$ bits (which is simply a concatenation of $n$ ciphertexts each encrypting one bit).
\begin{itemize}
    \item $\Enc(\pk,b):$ Sample $x,\theta \gets \{0,1\}^\secp$ and output \[\ct \coloneqq \left(\ket{x}_{\theta},\FHE.\Enc\left(\theta,b \oplus \bigoplus_{i: \theta_{i} = 0} x_{i}\right)\right), \ \ \ \vk \coloneqq (x,\theta).\] 
    \item $\Eval(\pk,C,\ct):$ Parse $\ct \coloneqq (\ket{x_1}_{\theta_1},\ct'_1),\dots,(\ket{x_n}_{\theta_n},\ct'_n)$. Consider the circuit $\widetilde{C}$ that takes $(x'_1,\theta_1,b_1'),\dots,(x'_n,\theta_n,b_n')$ as input, for each $i \in [n]$ computes $b_i = b_i' \oplus \bigoplus_{j:\theta_{i,j} = 0}x'_j$, and then computes and outputs $C(b_1,\dots,b_n)$. Then, apply $\widetilde{C}$ homomorphically in superposition to $\ct$ to obtain $\widetilde{\ct}$. Optionally, measure $\widetilde{\ct}$ in the standard basis to obtain a classical output ciphertext.

    \item $\Dec(\sk,\ct):$ parse $\ct \coloneqq \left(\ket{x}_\theta,\ct'\right)$, compute $(\theta,b') \gets \FHE.\Dec(\sk,\ct')$, measure $\ket{x}_\theta$ in the $\theta$-basis to obtain $x$, and output $b = b' \oplus \bigoplus_{i:\theta_i = 0}x_i$.
    \item $\Del(\ct):$ parse $\ct \coloneqq \left(\ket{x}_\theta,\ct'\right)$ and measure $\ket{x}_\theta$ in the Hadamard basis to obtain a string $x'$, and output $\cert \coloneqq x'$.
    \item $\Ver(\vk,\cert):$ parse $\vk$ as $(x,\theta)$ and $\cert$ as $x'$ and output $\top$ if and only if $x_i = x_i'$ for all $i$ such that $\theta_i = 1$.
\end{itemize}
\end{corollary}

\begin{proof}
Semantic security follows immediately from the semantic security of $\FHE$. Correctness of deletion follows immediately by definition the scheme. Certified deletion security follows from \cref{thm:main} and the semantic security of $\FHE$ by setting the distribution $\cZ_\secp(\theta,b',\sA)$ to sample $(\pk,\sk) \gets \FHE.\Gen(1^\secp)$, and output $(\sA,\FHE.\Enc(\pk,(\theta,b')))$, and setting the class of adversaries $\mathscr{A}$ to be all non-uniform families of QPT adversaries $\{\cA_\secp,\ket{\psi_\secp}\}_{\secp \in \bbN}$.
\end{proof}

\subsection{Commitments and zero-knowledge}\label{subsec:com-ZK}

A bit commitment scheme is an interactive protocol between two (potentially quantum) interactive machines, a committer $\cC = \{\cC_{\Com,\secp},\cC_{\Rev,\secp}\}_{\secp \in \bbN}$ and a receiver $\cR = \{\cR_{\Com,\secp},\cR_{\Rev,\secp}\}_{\secp \in \bbN}$. It operates in two stages.

\begin{itemize}
    \item In the Commit phase, the committer $\cC_{\Com,\secp}(b)$ with input bit $b$ interacts with the receiver $\cR_{\Com,\secp}$. This interaction results in a joint state on a committer and receiver register, which we denote $(\sC,\sR) \gets \Com\langle \cC_{\Com,\secp}(b),\allowbreak\cR_{\Com,\secp} \rangle$.
    \item In the Reveal phase, the parties continue to interact, and the receiver outputs a trit $\mu \in \{0,1,\bot\}$, which we denote by $\mu \gets \Rev\langle \cC_{\Rev,\secp}(\sC),\cR_{\Rev,\secp}(\sR)\rangle$.
    
\end{itemize}

A commitment scheme that is \emph{statistically binding and computationally hiding} is one that satisfies the following three properties.

\begin{definition}[Correctness of decommitment]\label{def:correctness-decommitment}
A commitment scheme satisfies \emph{correctness of decommitment} if for any $b \in \{0,1\}$, it holds with overwhelming probability over $(\sC,\sR) \gets \Com\langle\cC_{\Com,\secp}(b),\allowbreak\cR_{\Com,\secp}\rangle, \mu \gets \Rev\langle\cC_{\Rev,\secp}(\sC),\cR_{\Rev,\secp}(\sR)\rangle$ that $\mu = b$.
\end{definition}

\begin{definition}[Computational hiding]\label{def:comp-hiding}
A commitment scheme satisfies \emph{computational hiding} if for any non-uniform QPT adversary and distinguisher $\cR^* = \{\cR^*_{\Com,\secp},\cD^*_{\secp},\ket{\psi_\secp}\}_{\secp \in \bbN}$, where $\ket{\psi_\secp}$ is a state on two registers $(\sR^*,\sD^*)$, it holds that 
\begin{align*}
    &\bigg|\Pr\left[\cD^*_\secp(\sR^*,\sD^*) = 1 : (\sC,\sR^*) \gets \Com\langle\cC_{\Com,\secp}(0),\cR^*_{\Com,\secp}(\sR^*)\rangle \right] \\ &- \Pr\left[\cD^*_\secp(\sR^*,\sD^*) = 1 : (\sC,\sR^*) \gets \Com\langle\cC_{\Com,\secp}(1),\cR^*_{\Com,\secp}(\sR^*)\rangle \right]\bigg| = \negl(\secp).
\end{align*}
\end{definition}

We follow \cite{cryptoeprint:2021:1663}'s notion of statistical binding, which asks for an unbounded extractor that obtains the committer's bit during the Commit phase.

\begin{definition}[Statistical binding]\label{def:SB}
A commitment scheme satisfies \emph{statistical binding} if for any unbounded adversary $\{\cC^*_{\Com,\secp}\}_{\secp \in \bbN}$ in the Commit phase, there exists an unbounded extractor $\cE = \{\cE_\secp\}_{\secp \in \bbN}$ such that for every unbounded adversary $\{\cC^*_{\Rev,\secp}\}_{\secp \in \bbN}$ in the Reveal phase, \[\TD\left(\Real_\secp^{\cC^*},\Ideal_\secp^{\cC^*,\cE}\right) = \negl(\secp),\] where $\Real_\secp^{\cC^*}$ and $\Ideal_\secp^{\cC^*,\cE}$ are defined as follows.
\begin{itemize}
    \item $\Real_\secp^{\cC^*}$: Execute the Commit phase $(\sC^*,\sR) \gets \Com\langle \cC^*_{\Com,\secp},\cR_{\Com,\secp}\rangle$. Execute the Reveal phase to obtain a trit $\mu \gets \Rev\langle \cC^*_{\Rev,\secp}(\sC^*),\cR_{\Rev,\secp}(\sR)\rangle$ along with the updated committer's state on register $\sC^*$. Output $(\mu,\sC^*)$.
    \item $\Ideal_\secp^{\cC^*,\cE}$: Run the extractor $(\sC^*,\sR,b^*) \gets \cE_\secp$, which outputs a joint state on registers $\sC^*,\sR$ along with a bit $b^*$. Next, execute the Reveal phase to obtain a trit $\mu \gets \Rev\langle \cC^*_{\Rev,\secp}(\sC^*),\cR_{\Rev,\secp}(\sR)\rangle$ along with the updated committer's state on register $\sC^*$. If $\mu \in \{\bot,b^*\}$ output $(\mu,\sC^*)$, and otherwise output a special symbol $\Fail$.
\end{itemize}
\end{definition}

We will also consider commitment schemes with an additional (optional) \emph{Delete} phase. That is, the committer and receiver will be written as three components: $\cC_\secp = \{\cC_{\Com,\secp},\cC_{\Del,\secp},\cC_{\Rev,\secp}\}$, and $\cR_\secp = \{\cR_{\Com,\secp},\cR_{\Del,\secp},\cR_{\Rev,\secp}\}$, and the protocol proceeds as follows.

\begin{itemize}
    \item In the Commit phase, the committer $\cC_{\Com,\secp}(b)$ with input bit $b$ interacts with the receiver $\cR_{\Com,\secp}$. This interaction results in a joint state on a committer and receiver register, which we denote $(\sC,\sR) \gets \Com\langle \cC_{\Com,\secp}(b),\allowbreak\cR_{\Com,\secp} \rangle$.
    \item In the Delete phase, the parties continue to interact. The committer outputs a bit $d_\cC \in \{\top,\bot\}$ indicating whether they accept or reject. We denote the resulting output and joint state of the committer and receiver by $(d_\cC,\sC,\sR) \gets \Del\langle \cC_{\Del,\secp}(\sC),\cR_{\Del,\secp}(\sR)\rangle$.
    \item The Reveal phase is only executed if the Delete phase has not been executed, and the receiver outputs a trit $\mu \in \{0,1,\bot\}$, which we denote by $\mu \gets \Rev\langle \cC_{\Rev,\secp}(\sC),\cR_{\Rev,\secp}(\sR)\rangle$.

\end{itemize}

For such commitments, we ask for an additional correctness property, and a stronger hiding property.

\begin{definition}[Correctness of deletion, \cite{cryptoeprint:2021:1315}]\label{def:correctness-deletion}
A bit commitment scheme satisfies \emph{correctness of deletion} if for any $b \in \{0,1\}$, it holds with overwhelming probability over $(\sC,\sR) \gets \Com\langle\cC_{\Com,\secp}(b),\cR_{\Com,\secp}\rangle, \allowbreak(d_\cC,\sC,\sR) \gets \Del\langle\cC_{\Del,\secp}(\sC),\cR_{\Del,\secp}(\sR)\rangle$ that $d_\cC = \top$.
\end{definition}



\begin{definition}[Certified everlasting hiding, \cite{cryptoeprint:2021:1315}]\label{def:CEH}
A commitment scheme satisfies \emph{certified everlasting hiding} if it satisfies the following two properties. First, for any non-uniform QPT adversary and distinguisher $\cR^* = \{\cR^*_{\Com,\secp},\cR^*_{\Del,\secp},\cD^*_\secp,\ket{\psi_\secp}\}_{\secp \in \bbN}$, where $\ket{\psi_\secp}$ is a state on two registers $(\sR^*,\sD^*)$, it holds that 
\begin{align*}
    &\bigg|\Pr\left[\cD^*_\secp(d_\cC,\sR^*,\sD^*) = 1 : \begin{array}{r}(\sC,\sR^*) \gets \Com\langle\cC_{\Com,\secp}(0),\cR^*_{\Com,\secp}(\sR^*)\rangle \\ (d_\cC,\sC,\sR^*) \gets \Del\langle\cC_{\Del,\secp}(\sC),\cR^*_{\Del,\secp}(\sR^*)\rangle\end{array}\right] \\ &- \Pr\left[\cD^*_\secp(d_\cC,\sR^*,\sD^*) = 1 : \begin{array}{r}(\sC,\sR^*) \gets \Com\langle\cC_{\Com,\secp}(1),\cR^*_{\Com,\secp}(\sR^*)\rangle \\ (d_\cC,\sC,\sR^*) \gets \Del\langle\cC_{\Del,\secp}(\sC),\cR^*_{\Del,\secp}(\sR^*)\rangle\end{array}\right]\bigg| = \negl(\secp).
\end{align*}

Second, for any non-uniform QPT adversary $\cR^* = \{\cR^*_{\Com,\secp},\cR^*_{\Del,\secp},\ket{\psi}\}_{\secp \in \bbN}$, where $\ket{\psi_\secp}$ is a state on two registers $(\sR^*,\sD^*)$, it holds that \[\TD\left(\EVEXP_\secp^{\cR^*}(0),\EVEXP_\secp^{\cR^*}(1)\right) = \negl(\secp),\] where the experiment $\EVEXP_\secp^{\cR^*}(b)$ is defined as follows.
\begin{itemize}
    \item Execute the Commit phase $(\sC,\sR^*) \gets \Com\langle\cC_{\Com,\secp}(b),\cR^*_{\Com,\secp}(\sR^*)\rangle$.
    \item Execute the Delete phase $(d_\cC,\sC,\sR^*) \gets \Del\langle\cC_{\Del,\secp}(\sC),\cR^*_{\Del,\secp}(\sR^*))\rangle$.
    \item If $d_\cC = \top$ then output $(\sR^*,\sD^*)$, and otherwise output $\bot$.
\end{itemize}
\end{definition}

Then, we have the following corollary of \cref{thm:main}.

\begin{corollary}\label{cor:cecom}
Given any statistically binding computationally hiding commitment scheme $\Com$, the commitment defined as follows is a statistically binding commitment scheme with certified everlasting hiding.
\begin{itemize}
    \item The committer, on input $b \in \{0,1\}$, samples $x,\theta \gets \{0,1\}^\secp$. Then, the committer and receiver engage in the Commit phase of $\Com$, where the committer has input $(\theta,b \oplus \bigoplus_{i: \theta_i = 0}x_i)$. Finally, the committer sends $\ket{x}_\secp$ to the receiver.
    \item For the Delete phase, the receiver measures the state $\ket{x}_\theta$ in the Hadamard basis to obtain a string $x'$, and sends $x'$ to the committer. The committer outputs $\top$ if and only if $x_i = x'_i$ for all $i$ such that $\theta_i = 1$.
    \item For the Reveal phase, the committer and receiver engage in the Reveal phase of $\Com$, where the committer reveals the committed input $(\theta,b')$. If this passes, the receiver measures $\ket{x}_\theta$ in the $\theta$ basis to obtain $x$ and outputs $b = b' \oplus \bigoplus_{i: \theta_i = 0}x_i$.
\end{itemize}
\end{corollary}

\begin{proof}
First, we show that statistical binding is preserved. Given a malicious committer $\{\cC^*_{\Com,\secp}\}_{\secp \in \bbN}$, consider the experiment $\Ideal_\secp^{\cC^*,\cE}$ specified by an extractor $\cE$ defined as follows. 

\begin{itemize}
    \item Invoke the extractor for the underlying commitment scheme on the first part of $\cC^*_{\Com,\secp}$, which produces a joint state on $(\sC^*,\sR)$ and extracted values $(\theta^*,b^*)$.
    \item Continue running $\cC^*_{\Com,\secp}$ until it outputs a $\secp$-qubit state on register $\sX$, which in the honest case will hold a state of the form $\ket{x}_\theta$.
    \item Measure the register $\sX$ in the $\theta^*$ basis to produce $x^*$, and set the extracted bit $\widehat{b}^* \coloneqq b^* \oplus \bigoplus_{i: \theta_i^* = 0}x^*_i$.
\end{itemize}

Then, the Reveal phase of the underlying commitment scheme is run to produce a final committer's state on register $\sC^*$ and receiver's output, which is either $\bot$ or some $(\theta',b')$. Finally, the receiver either outputs $\bot$ or completes the Reveal phase by measuring register $\sX$ in the $\theta$ basis to obtain $x$, and outputting $b' \oplus \bigoplus_{i: \theta_i = 0}x_i$.

Note that, by the statistical binding property of the underlying commitment, the final state on $\sC^*$ produced by $\cC^*_{\Rev,\secp}$ in this experiment will be within negligible trace distance of the state on $\sC^*$ output in the $\Real_\secp^{\cC^*}$ experiment, and moreover the probability that $\cR_{\Rev,\secp}$ accepts opened values $(\theta',b')$ that are not equal to the previously extracted values $(\theta^*,b^*)$ is negligible. Thus, conditioned on opening accepting, with all but negligible probability the extractor's and receiver's measurement of $\sX$ will be identical. Thus, the extracted bit and receiver's output will be the same, and the outcome $\mathsf{FAIL}$ will only occur with negligible probability.

Next, we show certified everlasting hiding. The first property follows immediately from the hiding of the underlying commitment scheme, since there are no messages from $\cC$ in the delete phase, and the bit $d_\cC$ is computed independently of $b$. The second property follows from hiding of the underlying commitment scheme and \cref{thm:main} by setting  $\cZ_\secp(\theta,b',\sA)$ and $\cA_\secp$ as follows.

\begin{itemize}
    \item $\cZ_\secp(\theta,b',\sA)$ initializes registers $(\sR^*,\sD^*)$ with $\ket{\psi_\secp}$, runs $(\sC_\theta,\sR^*) \gets \Com\langle\cC_{\secp}(\theta,b'),\cR^*_{\Com,\secp}(\sR^*)\rangle$ with the first part of $\cR^*_{\Com,\secp}$ (where $\cC_\secp$ is the commit algorithm of the underlying commitment scheme $\Com$), and outputs the resulting state on registers $\sR^*,\sA$ (recall that $\sA$ holds the BB84 states $\ket{x}_\theta$).
    \item $\cA_\secp$ receives $(\sR^*,\sA)$ and runs $\cR^*_{\Del,\secp}(\sR^*,\sA)$, which outputs a classical certificate and a left-over quantum state. 
    
\end{itemize}
\end{proof}

\begin{remark}
We note that the above corollary explicitly considers underlying statistically binding commitment schemes that may include quantum communication, and thus one implication is that statistically binding commitments with certified everlasting hiding can be built just from the assumption that pseudo-random quantum states exist \cite{cryptoeprint:2021:1691,cryptoeprint:2021:1663}.
\end{remark}

\begin{remark}
Similarly to the setting of public-key encryption, single-bit certified everlasting hiding for statistically binding commitments implies multi-bit certified everlasting hiding.
\end{remark}

\subsubsection{Certified everlasting zero-knowledge proofs for QMA}
We begin by defining proofs for QMA with certified everlasting zero-knowledge, introduced in~\cite{cryptoeprint:2021:1315}. Our definition is identical to theirs, except that we also guarantee computational zero-knowledge in the case that the verifier outputs invalid deletion certificates. In what follows, we will assume familiarity with the notion of a (statistically sound) proof for a QMA promise problem. 

\begin{definition}\label{def:cehzk}
A certified everlasting zero-knowledge proof for a QMA promise problem $A = (A_{\mathsf{yes}}, A_{\mathsf{no}})$ is a proof for $A$ that
additionally satisfies the following properties.
\begin{itemize}
\item{ \bf (Perfect) Correctness of certified deletion.}
For every instance $x \in A_{\mathsf{yes}}$ and every state $\ket{\psi} \in R_A(x)$, the prover  outputs $\top$ as its output in the interaction $\langle \mathcal{P}(x,\ket{\psi}^{\otimes {k(|x|)}}),\cV(x) \rangle$.

\item{\bf Certified everlasting zero-knowledge.}
Let $\mathsf{REAL}_\secp{\langle \mathcal{P}(x,\ket{\psi}^{\otimes {k(|x|)}}),\cV^*(x)\rangle}$ denote the joint distribution of the output of an honest prover and the state of an arbitrary QPT verifier $\cV^*$ after they execute the proof on instance $x \in A_{\mathsf{yes}}$, where the prover has as quantum input a polynomial number $k(|x|)$ copies of a state $\ket{\psi} \in R_A(x)$.
Then there exists a QPT algorithm $\mathsf{Sim}$ that on input any $x \in A_{\mathsf{yes}}$ and with oracle access to any non-uniform QPT $\cV^* = \{\cV^*_\secp\}_{\secp \in \bbN}$, outputs distribution $\mathsf{Sim}_\secp^{\cV^*}(x)$ such that:
\begin{itemize}
    \item First, we have everlasting zero-knowledge against adversaries that produce a valid deletion certificate, i.e., \[\TD\left(\mathsf{EV}\left(\mathsf{REAL}_\secp{\langle \mathcal{P}\left(x,\ket{\psi}^{\otimes {k(|x|)}}\right),\cV^*(x)\rangle}\right),\mathsf{EV}\left(\mathsf{Sim}_\secp^{\cV^*}(x)\right)\right) = \negl(\secp),\] where 
    $\mathsf{EV}(\cdot)$ is a quantum circuit that on input a classical string $o \in \{\top, \bot\}$ and a quantum state $\rho$ outputs $(\top, \rho)$ when $o = \top$, and otherwise outputs $(\bot, \bot)$. 
    \item Second, we have computational zero-knowledge against all adversaries, even when they do not necessarily output valid deletion certificates, i.e.,
    for every QPT distinguisher $\cD^* = \{\cD^*_\secp\}_{\secp \in \bbN}$, 
    \[\bigg|\Pr\left[\cD^*_\secp\left(\mathsf{REAL}_\secp{\langle \mathcal{P}\left(x,\ket{\psi}^{\otimes {k(|x|)}}\right), \cV^*(x) \rangle}\right) = 1\right] - \Pr\left[\cD^*_\secp\left(\mathsf{Sim}_\secp^{\cV^*}(x)\right) = 1\right]\bigg| = \negl(\secp)\]
\end{itemize}
\end{itemize}
\end{definition}

Next, we define a notion of classical extractor-based binding for commitments. This definition was introduced in~\cite{cryptoeprint:2021:1315}, and while their definition requires perfect extraction, we observe that their theorem holds even if the underlying commitment satisfies only statistical extraction. As such we allow for $\negl(\secpar)$ statistical error in our definition.

\begin{definition}[Classical extractor-based binding~\cite{cryptoeprint:2021:1315}]
A quantum commitment with classical non-interactive decommitment satisfies classical extractor-based binding if there exists an unbounded-time deterministic algorithm $\mathsf{Ext}$ that on input the classical transcript of a (possibly quantum) commitment $\mathsf{com}$, outputs the only unique classical decommitment string $d$ that will cause the verifier to accept the reveal phase, except with negligible probability.
\end{definition}

Finally, we will rely on the following theorem from~\cite{cryptoeprint:2021:1315}, which we describe below, paraphrased according to our definitions.

\begin{theorem}[\cite{cryptoeprint:2021:1315}] \label{thm:imp}
Assuming the existence of commitments satisfying statistical {\em classical extractor-based} binding and certified everlasting hiding (according to Definition \ref{def:CEH}), there exists a zero-knowledge proof for QMA satisfying certified everlasting zero-knowledge (according to Definition \ref{def:cehzk}).
\end{theorem}

We obtain the following corollary of Theorem \ref{thm:imp} and our Corollary \ref{cor:cecom}.
\begin{corollary}
Assuming the existence of post-quantum one-way functions, there exists a zero-knowledge proof for QMA satisfying certified everlasting zero-knowledge (according to Definition \ref{def:cehzk}).
\end{corollary}
This corollary follows from the observation that our construction of commitments with certified everlasting hiding, when instantiated with any classical statistically binding commitment (and in particular Naor's commitment from one-way functions) satisfies classical extractor-based binding. The extractor simply outputs the decommitment of the classical part of our commitment. The resulting commitment with certified everlasting hiding and classical extractor-based binding can be plugged into Theorem \ref{thm:imp} to obtain the corollary above.

\subsection{Timed-release encryption}

A timed-release encryption scheme \cite{Rivest1996TimelockPA,EC:Unruh14} $\TRE = (\TRE.\Enc,\TRE.\Dec)$  has the following syntax.

\begin{itemize}
    \item $\TRE.\Enc(1^\secp,m) \to \ct$ is a polynomial-time algorithm that takes as input the security parameter $1^\secp$ and a message $m$ and outputs a ciphertext $\ct$.
    \item $\TRE.\Dec(\ct) \to m$ is a polynomial-time algorithm that takes as input a ciphertext $\ct$ and outputs a message $m$.
\end{itemize}

A timed-released encryption scheme is (post-quantum) $T(\secp)$-hiding if the following holds.

\begin{definition}[Hiding time-released encryption]\label{def:TRE-hiding}
A timed-released encryption scheme $\TRE = (\TRE.\Enc,\allowbreak\TRE.\Dec)$ is $T(\secp)$-hiding if for any non-uniform quantum polynomial-time\footnote{As discussed in \cite{EC:Unruh14}, it is important to have a polynomial-time bound on the overall complexity of the adversary, in addition to the $T(\secp)$ parallel time bound.} adversary $\cA = \{\cA_\secp,\ket{\psi_\secp}\}_{\secp \in \bbN}$ with at most $T(\secp)$ parallel time, \[\bigg|\Pr\left[\cA_\secp(\TRE.\Enc(1^\secp,0)) = 1\right] - \Pr\left[\cA_\secp(\TRE.\Enc(1^\secp,1)) = 1\right]\bigg| = \negl(\secp).\]
\end{definition}

Now, we augment the syntax of a $\TRE$ scheme with algorithms $\RTRE.\Del, \RTRE.\Ver$ to arrive at the notion of a \emph{revocable} timed-release encryption scheme $\RTRE$.

\begin{itemize}
    \item $\RTRE.\Enc(1^\secp,m) \to (\ct,\vk)$ is a polynomial-time algorithm that takes as input the security parameter $1^\secp$ and a message $m$ and outputs a quantum ciphertext $\ct$ and a (potentially quantum) verification key $\vk$.
    \item $\RTRE.\Dec(\ct) \to m$ is a polynomial-time algorithm that takes as input a quantum ciphertext $\ct$ and outputs a message $m$.
    \item $\RTRE.\Del(\ct) \to \cert$ is a quantum algorithm that takes as input a quantum ciphertext $\ct$ and outputs a (potentially quantum) deletion certificate $\cert$.
    \item $\RTRE.\Ver(\vk,\cert) \to \{\top,\bot\}$ is a (potentially quantum) algorithm that takes as input a (potentially quantum) verification key $\vk$ and a (potentially quantum) deletion certificate $\cert$ and outputs either $\top$ or $\bot$.
\end{itemize}

We say that $\RTRE$ satisfies revocable hiding if the following holds.

\begin{definition}[Revocably hiding time-released encryption]\label{def:TRE-revocable-hiding}
A timed-released encryption scheme $\RTRE = (\RTRE.\Enc,\RTRE.\Dec,\RTRE.\Del,\RTRE.\Ver)$ is $T(\secp)$-revocably hiding if for any non-uniform quantum polynomial-time adversary $\cA = \{\cA_\secp,\ket{\psi_\secp}\}_{\secp \in \bbN}$ with at most $T(\secp)$ parallel time, it holds that \[\TD\left(\EVEXP_\secp^{\cA}(0),\EVEXP\secp^{\cA}(1)\right) = \negl(\secp),\]
and  \[\bigg|\Pr\left[\CEXP_\secp^{\cA}(0) = 1\right] - \Pr\left[\CEXP_\secp^{\cA}(1) = 1\right]\bigg| = \negl(\secp),\]
where the experiment $\EVEXP\secp^{\cA}(b)$ is defined as follows.
\begin{itemize}
    \item Sample $(\ct,\vk) \gets \RTRE.\Enc(1^\secp,b)$.
    \item Initialize $\cA_\secp(\ket{\psi_\secp})$ with $\ct$.
    \item Parse $\cA_\secp$'s output as a deletion certificate $\cert$ and a residual state on register $\sA'$.
    \item If $\RTRE.\Ver(\vk,\cert) = \top$ then output $\sA'$, and otherwise output $\bot$.
\end{itemize}
and the experiment $\CEXP\secp^{\cA}(b)$ is defined as follows.
\begin{itemize}
    \item Sample $(\ct,\vk) \gets \RTRE.\Enc(1^\secp,b)$.
    \item Initialize $\cA_\secp(\ket{\psi_\secp})$ with $\ct$.
    \item Parse $\cA_\secp$'s output as a deletion certificate $\cert$ and a residual state on register $\sA'$.
    \item Output $\cA_\secp(\sA', \Ver(\vk,\cert))$.
\end{itemize}
\end{definition}

We say that $\RTRE$ is a revocable time-released encryption scheme against $T(\secp)$-parallel time adversaries if it satisfies (i) hiding (\cref{def:TRE-hiding}) against $T(\secp)$-parallel time adversaries, (ii) correctness of deletion (\cref{def:CD-correctness}), and (iii) recovable hiding (\cref{def:TRE-revocable-hiding}) against $T(\secp)$-parallel time adversaries.

Then, we have the following corollary of \cref{thm:main}.

\begin{corollary}
Given any post-quantum secure time-released encryption $\TRE = (\TRE.\Enc,\TRE.\Dec)$ against $T(\secp)$-parallel time adversaries, the scheme $\RTRE = (\Enc',\Dec',\Del,\Ver)$ defined as follows is a secure \emph{revocable} time-released encryption scheme against $T(\secp)$-parallel time adversaries.
\begin{itemize}
     \item $\Enc'(\pk,m):$ sample $x,\theta \gets \{0,1\}^\secp$ and output \[\ct \coloneqq \left(\ket{x}_\theta,\TRE.\Enc\left(\theta,b \oplus \bigoplus_{i: \theta_i = 0} x_i\right)\right), \ \ \ \vk \coloneqq (x,\theta).\]
    \item $\Dec'(\sk,\ct):$ parse $\ct \coloneqq \left(\ket{x}_\theta,\ct'\right)$, compute $(\theta,b') \gets \TRE.\Dec(\sk,\ct')$, measure $\ket{x}_\theta$ in the $\theta$-basis to obtain $x$, and output $b = b' \oplus \bigoplus_{i:\theta_i = 0}x_i$.
    \item $\Del(\ct):$ parse $\ct \coloneqq \left(\ket{x}_\theta,\ct'\right)$ and measure $\ket{x}_\theta$ in the Hadamard basis to obtain a string $x'$, and output $\cert \coloneqq x'$.
    \item $\Ver(\vk,\cert):$ parse $\vk$ as $(x,\theta)$ and $\cert$ as $x'$ and output $\top$ if and only if $x_i = x_i'$ for all $i$ such that $\theta_i = 1$.
\end{itemize}
\end{corollary}

\begin{proof}
Hiding follows immediately from the hiding of $\TRE$. Correctness of deletion follows immediately from the description of the scheme. Revocable hiding follows because
\begin{itemize}
\item First, \[\TD\left(\EVEXP_\secp^{\cA}(0),\EVEXP\secp^{\cA}(1)\right) = \negl(\secp),\]
follows from \cref{thm:main} and the hiding of $\TRE$, by setting the distribution $\cZ(\theta,b',\sA)$ to sample $(\ct,\vk) \gets \TRE.\Enc(1^\secp,(\theta,b'))$ and output $(\sA,\ct)$, and setting the class of adversaries $\mathscr{A}$ to be all non-uniform QPT adversaries $\{\cA_\secp,\ket{\psi_\secp}\}_{\secp \in \bbN}$ with at most $T(\secp)$ parallel time.
\item Second, \[\bigg|\Pr\left[\CEXP_\secp^{\cA}(0) = 1\right] - \Pr\left[\CEXP_\secp^{\cA}(1) = 1\right]\bigg| = \negl(\secp),\]
follows from the fact that the timed-release encryption remains (computationally) semantically secure even when the adversary is given the verification key corresponding to the challenge ciphertext.
\end{itemize}
This completes our proof.
\end{proof}

\begin{remark}
Similarly to the setting of public-key encryption, single-bit certified everlasting hiding for timed-release encryption implies certified everlasting hiding for multi-bit messages.
\end{remark}
\fi
\ifsubmission\else\section{Cryptography with Everlasting Security Transfer}
\label{sec:2pc-est}

In this section, we construct bit commitment and secure computation schemes that satisfy our notion of Everlasting Security Transfer (EST). In \cref{subsec:2PC-defs}, we formalize the notion of deletion and EST in the context of simulation security. We also derive a quantum sequential composition theorem for \emph{reactive functionalities}, extending the framework of \cite{C:HalSmiSon11}. Extending to reactive functionalities is crucial for us, since the bit commitment functionality we compose has multiple phases. Finally, we formalize composition of protocols with EST, defining the notion of a ``deletion-composable'' protocol. Next, we show how to construct ideal commitments with EST in \cref{subsec:one-sided} and \cref{subsec:ideal}, via a two-step process outlined in \cref{subsec:tech-overview}. Finally, in \cref{subsec:secure-computation}, we make use of our ideal commitment with EST, our composition theorems, and known compilers, to obtain the notions of two-party and multi-party computation with EST.

\subsection{Definitions}\label{subsec:2PC-defs}

\paragraph{Ideal functionalities.}

An ideal functionality $\cF$ is a classical interactive machine specifying some (potentially reactive) distributed classical computation. Reactive means that the distributed computation is broken into multiple ``phases'' with distinct inputs and outputs, and the outputs of previous phases may be used as inputs in later phases. For now, we will specifically consider \emph{two-party} functionalities. Each invocation of an ideal functionality is associated with some session id $\mathsf{sid}$. We will be interested in designing protocols that \emph{securely realize} ideal functionalities (defined later), but first we specify the main ideal functionality that we consider in this work: bit commitment $\cF_\com$. 

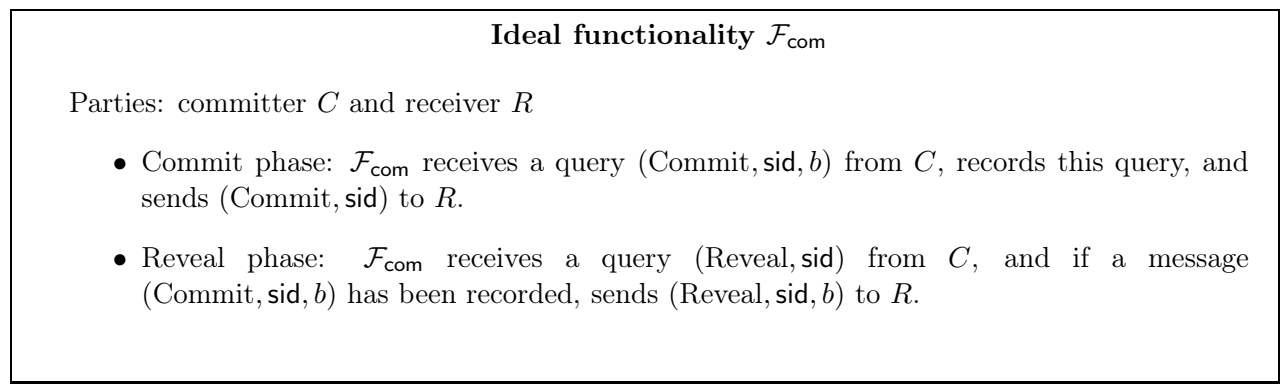
\protocol{Ideal functionality $\cF_\com$}{Specification of the bit commitment ideal functionality.}{fig:comIF}{
Parties: committer $C$ and receiver $R$
\begin{itemize}
    \item Commit phase: $\cF_\com$ receives a query $(\text{Commit},\mathsf{sid},b)$ from $C$, records this query, and sends $(\text{Commit},\mathsf{sid})$ to $R$.
    \item Reveal phase: $\cF_\com$ receives a query $(\text{Reveal},\mathsf{sid})$ from $C$, and if a message $(\text{Commit},\mathsf{sid},b)$ has been recorded, sends $(\text{Reveal},\mathsf{sid},b)$ to $R$.
\end{itemize}
\noindent}


In this work, we will consider augmenting ideal functionalities with a ``deletion phase'', which can be used by parties to transfer everlasting security. When parties are labeled $A$ and $B$, we maintain the precedent that deletion is from $B$ to $A$, that is, $A$ can request that $B$ deletes $A$'s information.

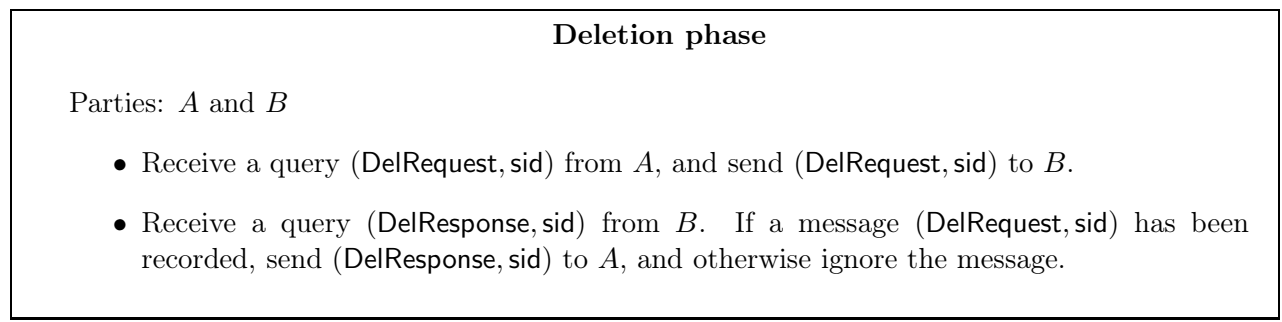
\protocol{Deletion phase}{Specification of a generic deletion phase that can be added to any ideal functionality $\cF$.}{fig:delPhase}{
Parties: $A$ and $B$
\begin{itemize}
    \item Receive a query $(\mathsf{DelRequest},\mathsf{sid})$ from $A$, and send $(\mathsf{DelRequest},\mathsf{sid})$ to $B$.
    \item Receive a query $(\mathsf{DelResponse},\mathsf{sid})$ from $B$. If a message $(\mathsf{DelRequest},\mathsf{sid})$ has been recorded, send $(\mathsf{DelResponse},\mathsf{sid})$ to $A$, and otherwise ignore the message. 
\end{itemize}
}

Importantly, if a deletion phase is added to a \emph{reactive} functionality $\cF$, we allow party $A$ to request the Deletion phase (send $\mathsf{DelRequest}$) \emph{between any two phases of $\cF$}, or \emph{at the end of $\cF$}. But, once the Deletion phase has been executed, this marks the end of the reactive functionality, so no other phases will be executed.

\paragraph{Security with abort.} In what follows, we will by default consider the notion of \emph{security with abort}, where the ideal functionality $\cF$ is always modified to (1) know the identities of corrupted parties and (2) be slightly reactive: after all parties have provided input, the functionality computes outputs and sends these outputs to the corrupt parties only. Then the functionality awaits either a ``deliver'' or ``abort'' command from the corrupted parties. Upon receiving ``deliver'', the functionality delivers any honest party outputs. Upon receiving ``abort'', the functionality instead delivers $\abort$ to all the honest parties.

\paragraph{The real-ideal paradigm.} A two-party protocol $\Pi_\cF$ for computing the (potentially reactive) functionality $\cF$ consists of two families of quantum interactive machines $A$ and $B$. An adversary intending to attack the protocol by corrupting a party $M \in \{A,B\}$ can be described by a family of sequences of quantum interactive machines $\{\cA_\secp \coloneqq (\cA_{\secp,1},\dots,\cA_{\secp,\ell})\}_{\secp \in \bbN}$, where $\ell$ is the number of phases of $\cF$. This adversarial interaction happens in the presence of an \emph{environment}, which is a family of sequences of quantum operations $\{\cZ_\secp \coloneqq (\cZ_{\secp,1},\dots,\cZ_{\secp,\ell})\}_{\secp \in \bbN}$, and a family of initial advice states $\{\ket{\psi_\secp}\}_{\secp \in \bbN}$. It proceeds as follows.

\begin{itemize}
    \item $\cZ_{\secp,1}$ receives as input $\ket{\psi_\secp}$. It outputs what (if any) inputs the honest party $H \in \{A,B\}$ is initialized with for the first phase of $\Pi_\cF$. It also outputs a quantum state on registers $(\sA,\sZ)$, where $\sA$ holds the state of the adversary and $\sZ$ holds the state of the environment,
    \item $\cA_{\secp,1}$ receives as input a state on register $\sA$, and interacts with the honest party in the first phase of $\Pi_\cF$. It outputs a state on register $\sA$.
    \item $\cZ_{\secp,2}$ receives as input registers $(\sA,\sZ)$ along with the honest party outputs from the first phase. It computes honest party inputs for the second phase, and updates registers $(\sA,\sZ)$.
    \item $\cA_{\secp,2},\cZ_{\secp,3},\dots,\cA_{\secp,\ell}$ are defined analogously.
\end{itemize}

Given an adversary, environment, and advice, we define the random variable $\Pi_\cF[\cA_\secp,\cZ_\secp,\ket{\psi_\secp}]$ as the output of the above procedure, which includes registers $(\sA,\sZ)$ and the final honest party outputs.

An \emph{ideal-world} protocol $\widetilde{\Pi}_\cF$ for functionality $\cF$ consists of ``dummy'' parties $\widetilde{A}$ and $\widetilde{B}$ that have access to an additional ``trusted'' party that implements $\cF$. That is, $\widetilde{A}$ and $\widetilde{B}$ only interact directly with $\cF$, providing inputs and receiving outputs, and do not interact with each other. We consider the execution of ideal-world protocols in the presence of a simulator, described by a family of sequences of quantum interactive machines $\{\cS_\secp \coloneqq (\cS_{\secp,1},\dots,\cS_{\secp,\ell})\}_{\secp \in \bbN}$, analogous to the definition of an adversary above. This interaction also happens in the presence of an environment $\{\cZ_\secp \coloneqq (\cZ_{\secp,1},\dots,\cZ_{\secp,\ell})\}_{\secp \in \bbN}$, and a family of initial advice states $\{\ket{\psi_\secp}\}_{\secp \in \bbN}$, as described above, and we define the analogous random variable $\widetilde{\Pi}_\cF[\cS_\secp,\cZ_\secp,\ket{\psi_\secp}]$.

\paragraph{Secure realization and composition.} Now, we formally define what it means for a protocol $\Pi_\cF$ to securely realize a (potentially reactive) functionality $\cF$. We give definitions for both computational and statistical security.

\begin{definition}[Computational secure realization]\label{def:securerealization}
A protocol $\Pi_\cF$ \emph{computationally securely realizes} the $\ell$-phase functionality $\cF$ if for every QPT adversary $\{\cA_\secp \coloneqq (\cA_{\secp,1},\dots,\cA_{\secp,\ell})\}_{\secp \in \bbN}$ corrupting either party $A$ or $B$, there exists a QPT simulator $\{\cS_\secp \coloneqq (\cS_{\secp,1},\dots,\cS_{\secp,\ell})\}_{\secp \in \bbN}$ such that for any QPT environment $\{\cZ_{\secp} \coloneqq (\cZ_{\secp,1},\dots,\cZ_{\secp,\ell})\}_{\secp \in \bbN}$, polynomial-size family of advice $\{\ket{\psi_\secp}\}_{\secp \in \bbN}$, and QPT distinguisher $\{\cD_\secp\}_{\secp \in \bbN}$, it holds that \[\bigg|\Pr\left[\cD_\secp\left(\Pi_\cF[\cA_\secp,\cZ_\secp,\ket{\psi_\secp}]\right) = 1\right] - \Pr\left[\cD_\secp\left(\widetilde{\Pi}_\cF[\cS_\secp,\cZ_\secp,\ket{\psi_\secp}]\right) = 1\right]\bigg| = \negl(\secp).\]
\end{definition}

For the notion of statistical secure realization, we allow the adversary and environment to be unbounded, but we require that the simulator is at most polynomially larger than the adversary. 

\begin{definition}[Statistical secure realization]\label{def:statsecurerealization}
A protocol $\Pi_\cF$ \emph{statistically securely realizes} the $\ell$-phase functionality $\cF$ if there exists a polynomial $p(\cdot)$ such that for every (potentially unbounded) adversary $\{\cA_\secp \coloneqq (\cA_{\secp,1},\dots,\cA_{\secp,\ell})\}_{\secp \in \bbN}$ corrupting either party $A$ or $B$, there exists a simulator $\{\cS_\secp \coloneqq (\cS_{\secp,1},\dots,\cS_{\secp,\ell})\}_{\secp \in \bbN}$ with size at most $p(\secp)$ times the size of $\{\cA_\secp\}_{\secp \in \bbN}$, such that for any (potentially unbounded) environment $\{\cZ_{\secp} \coloneqq (\cZ_{\secp,1},\dots,\cZ_{\secp,\ell})\}_{\secp \in \bbN}$ and polynomial-size family of advice $\{\ket{\psi_\secp}\}_{\secp \in \bbN}$, it holds that \[\TD\left(\Pi_\cF[\cA_\secp,\cZ_\secp,\ket{\psi_\secp}],\widetilde{\Pi}_\cF[\cS_\secp,\cZ_\secp,\ket{\psi_\secp}]\right) = \negl(\secp).\]
\end{definition}

\begin{remark}
Recall that trace distance between two distributions is an upper bound on the advantage that any unbounded machine has in distinguishing the distributions, so the above definition is equivalent to saying that no unbounded distinguisher has better than negligible advantage in distinguishing the real and ideal world outputs.
\end{remark}

Next, we consider the \emph{hybrid} model, where parties can make calls to an ideal-world protocol implementing some ideal functionality $\cG$. We call such a protocol a $\cG$-hybrid protocol, and denote it $\Pi^{\cG}$. Supposing that we also have a real-world protocol $\Gamma$ implementing $\cG$, we can consider the \emph{composed} protocol $\Pi^{\cG/\Gamma}$, where each invocation of $\cG$ is replaced with an invocation of the protocol $\Gamma$ for computing $\cG$. In this work, while we allow $\Pi$ to utilize many invocations of $\Gamma$, \textbf{we require that each phase of each invocation of $\Gamma$ is \emph{atomic}, meaning that no other protocol messages are interleaved during each phase of $\Gamma$}. That is, if $\cG$ is a reactive functionality, we allow different phases of different invocations to be interleaved, but we require that at any point in time, only a single phase is being executed, and no other protocol messages are interleaved during the computation of this phase. In this case, we can show the following sequential composition theorem, which is a straightforward extension of the composition theorem given in \cite{C:HalSmiSon11} to handle reactive functionalities.

\begin{theorem}[Extension of \cite{C:HalSmiSon11}]\label{thm:composition}
Let $\cF$ and $\cG$ be (potentially reactive) functionalities, let $\Pi^\cG$ be a $\cG$-hybrid protocol that computationally (resp. statistically) securely realizes $\cF$, and let $\Gamma$ be a protocol that computationally (resp. statistically) securely realizes $\cG$. Then, $\Pi^{\cG/\Gamma}$ computationally (resp. statistically) securely realizes $\cF$.
\end{theorem}

\begin{proof}
Let $\cG$ be a reactive $\ell$-phase functionality. Throughout this proof, we drop the dependence on $\secp$ for convenience. Let $(\cA,\cZ)$ be any adversary and environment attacking the protocol $\Pi^{\cG/\Gamma}$. Consider the first time in $\Pi^{\cG/\Gamma}$ that $\Gamma$ is invoked, which means the first time that the first phase of some subroutine $\Gamma$ is invoked (we note that other $\Gamma$ subroutines could occur between the phases of this first invocation). Write $(\cZ_1,\cA_1,\dots,\cZ_\ell,\cA_\ell)$ as an adversary and environment attacking the protocol $\Gamma$, according to the following. 

\begin{itemize}
    \item $\cZ_1$ runs $\cZ$ and then runs the interaction between $\cA$ and the honest party until right before the first time $\Gamma$ is invoked in $\Pi^{\cG/\Gamma}$. It outputs the adversary's state on register $\sA$, the honest party's input to $\Gamma$, and any other state kept by $\cZ$ along with the honest party's state on register $\sZ$.
    \item $\cA_1$ consists of the part of $\cA$ that interacts in the first phase of $\Gamma$. It takes as input a state on $\sA$ and outputs a state on $\sA$.
    \item $\cZ_2$ takes as input registers $(\sA,\sZ)$ and the honest party's output from $\Gamma$. It runs the interaction between $\cA$ and the honest party in $\Pi^{\cG/\Gamma}$ until right before the second phase of $\Gamma$ is invoked.
    \item $\cA_2,\cZ_3,\dots,\cA_\ell$ are defined analogously.
\end{itemize}

Now, since $\Gamma$ computationally (resp. statistically) securely realizes $\cG$, there exists a simulator $(\cS_1,\dots,\cS_\ell)$ defined based on $(\cA_1,\dots,\cA_\ell)$ such that, if we replace each $\cA_i$ interacting with the honest party with $\cS_i$ interacting with the ideal functionality $\cG$, then the output remains computationally (resp. statistically) indistinguishable. Note that the resulting interaction, defined by $(\cZ_1,\cS_1,\dots,\cZ_\ell,\cS_\ell)$, can be described by an adversary and environment $(\cA',\cZ)$ attacking the protocol $\Pi^{\cG/\Gamma}$ where the first invocation of $\Gamma$ is replaced with the parties querying the ideal functionality $\cG$. This follows because only the parts of $\cA$ that interacted in the first invocation of $\Gamma$ were changed, since each phase of $\Gamma$ was atomic. Now, continuing this argument for each invocation of $\Gamma$, we eventually arrive at an adversary and environment $(\cA'',\cZ)$ attacking the $\cG$-hybrid protocol $\Pi^\cG$. Note that $\cA''$ was defined based on $\cA$, and $\cZ$ remained unchanged. Thus, the fact that $\Pi^\cG$ computationally (resp. statistically) securely realizes $\cF$ completes the proof of the theorem, since we can define a simulator $\cS''$ based on $\cA''$, where indistinguishability will hold for any environment $\cZ$. 
%
\end{proof}

\paragraph{Secure realization with everlasting security transfer.}

Next, we define the notion of secure realization with everlasting security transfer (EST). Here, parties are interested in securely computing an ideal functionality $\cF$ with a \emph{deletion phase} added to the end, which we denote by $\cF^{\Del}$. 

The deletion phase adds one bits to each honest party output, which we denote by $\mathsf{DelReq}$ (which is party $B$'s output, and is set to 1 if party $A$ initiates the Delete phase by issuing a request, and 0 otherwise) and $\mathsf{DelRes}$ (which is party $A$'s output, and is set to 1 if party $B$ sends a Delete response and 0 otherwise). Then, we have the following definition.

\begin{definition}[Secure realization with Everlasting Security Transfer]\label{def:security-with-EST}
A protocol $\Pi_\cF$ securely realizes the $\ell$-phase functionality $\cF$ between parties $A$ and $B$ with EST if $\Pi_\cF$ computationally securely realizes $\cF^\Del$ (\cref{def:securerealization}) and the following additional properties hold.
\begin{itemize}
\item
{\bf Statistical security against $A$ when no security transfer occurs.} There exists a polynomial $p(\cdot)$ such that for every (potentially unbounded) adversary $\{\cA_\secp \coloneqq (\cA_{\secp,1},\dots,\cA_{\secp,\ell})\}_{\secp \in \bbN}$ corrupting party $A$, there exists a simulator $\{\cS_\secp \coloneqq (\cS_{\secp,1},\dots,\cS_{\secp,\ell})\}_{\secp \in \bbN}$ with size at most $p(\secp)$ times the of size of $\{\cA_\secp\}_{\secp \in \bbN}$, such that for any (potentially unbounded) environment $\{\cZ_{\secp} \coloneqq (\cZ_{\secp,1},\dots,\cZ_{\secp,\ell})\}_{\secp \in \bbN}$ and polynomial-size family of advice $\{\ket{\psi_\secp}\}_{\secp \in \bbN}$, \[\TD\left(\Pi_{\cF^\Del}^{\mathsf{DelReq} = 0}[\cA_\secp,\cZ_\secp,\ket{\psi_\secp}], \widetilde{\Pi}_{\cF^\Del}^{\mathsf{DelReq} = 0}[\cS_\secp,\cZ_\secp,\ket{\psi_\secp}]\right) = \negl(\secp),\] where $\Pi^{\mathsf{DelReq} = 0}_{\cF^\Del}[\cA_\secp,\cZ_\secp,\ket{\psi_\secp}]$ is defined to be equal to $\Pi_{\cF^\Del}[\cA_\secp,\cZ_\secp,\ket{\psi_\secp}]$ if party $B$'s output $\mathsf{DelReq}$ is set to 0, and defined to be $\bot$ otherwise, and likewise for $ \widetilde{\Pi}_{\cF^\Del}^{\mathsf{DelReq} = 0}[\cS_\secp,\cZ_\secp,\ket{\psi_\secp}]$.

\item {\bf Certified everlasting security against $B$.} For every QPT adversary $\{\cA_\secp \coloneqq (\cA_{\secp,1},\dots,\cA_{\secp,\ell})\}_{\secp \in \bbN}$ corrupting party $B$, there exists a QPT simulator $\{\cS_\secp \coloneqq (\cS_{\secp,1},\dots,\cS_{\secp,\ell})\}_{\secp \in \bbN}$ such that for any QPT environment $\{\cZ_{\secp} \coloneqq (\cZ_{\secp,1},\dots,\cZ_{\secp,\ell})\}_{\secp \in \bbN}$, and polynomial-size family of advice $\{\ket{\psi_\secp}\}_{\secp \in \bbN}$, \[\TD\left(\Pi_{\cF^\Del}^{\mathsf{DelRes} = 1}[\cA_\secp,\cZ_\secp,\ket{\psi_\secp}], \widetilde{\Pi}_{\cF^\Del}^{\mathsf{DelRes} = 1}[\cS_\secp,\cZ_\secp,\ket{\psi_\secp}]\right) = \negl(\secp),\] where $\Pi^{\mathsf{DelRes} = 1}_{\cF^\Del}[\cA_\secp,\cZ_\secp,\ket{\psi_\secp}]$ is defined to be equal to $\Pi_{\cF^\Del}[\cA_\secp,\cZ_\secp,\ket{\psi_\secp}]$ if party $A$'s output $\mathsf{DelRes}$ is set to 1, and defined to be $\bot$ otherwise, and likewise for $ \widetilde{\Pi}_{\cF^\Del}^{\mathsf{DelRes} = 1}[\cS_\secp,\cZ_\secp,\ket{\psi_\secp}]$.
\end{itemize}
\end{definition}

\paragraph{Deletion-composable protocols.} Finally, we consider the composition of protocols that securely realize functionalities with EST. Suppose we have a $\cG^\Del$-hybrid protocol $\Pi^{\cG^\Del}$ for implementing a functionality $\cF^\Del$. We say that $\Pi^{\cG^\Del}$ is \emph{deletion-composable} if the following two properties hold. 
\begin{enumerate}
    \item If the deletion phase of $\cF^\Del$ is never requested, then none of the deletion phases of $\cG^\Del$ are requested.
    \item If the deletion phase of $\cF^\Del$ is accepted by party $A$, meaning that $\mathsf{DelRes} = 1$, then it must be the case that the deletion phases of all the $\cG^\Del$ sub-routines are requested and accepted by $A$.
\end{enumerate}

Then, we can show the following composition theorem, which essentially follows from \cref{thm:composition}.

\begin{theorem}
\label{thm:compose-EST}
Let $\Pi^{\cG^\Del}$ be a deletion-composable protocol that statistically securely realizes\footnote{One could strengthen this theorem to only requiring that $\Pi^{\cG^\Del}$ securely realizes $\cF^\Del$ with EST, but we state the theorem with statistical security for simplicity.} a functionality $\cF^\Del$, and let $\Gamma$ be a protocol that securely implements $\cG^\Del$ with EST. Then $\Pi^{\cG^\Del / \Gamma}$ securely implements $\cF^\Del$ with EST.
\end{theorem}

\begin{proof}
First, the fact that $\Pi^{\cG^\Del/\Gamma}$ computationally securely realizes $\cF^\Del$ follows from directly from \cref{thm:composition} and the fact that both statistical secure realization and secure realization with EST imply computational secure realization. 

Next, statistical security against $A$ in $\Pi^{\cG^\Del/\Gamma}$ when $\mathsf{DelReq} = 0$ also follows directly from \cref{thm:composition} (applied to statistical secure realization), since by the first property of deletion-composability, all of the underlying $\Gamma$ protocols are statistically secure against $A$.

Finally, we argue certified everlasting security against $B$ in $\Pi^{\cG^\Del/\Gamma}$ when $\mathsf{DelRes} = 1$. This does not follow generically from the statement of \cref{thm:composition}. However, it can be shown via essentially the same proof as the proof of \cref{thm:composition}. Starting with $\Pi^{\cG^\Del/\Gamma}$, we replace each invocation of $\Gamma$ with an invocation of $\cG^\Del$ one by one. Conditioned on the deletion phase of $\cG^\Del$ passing, we know that this switch is statistically indistinguishable by the environment, due to the fact that $\Gamma$ securely implements $\cG^\Del$ with EST. Thus, conditioned on the deletion phase of each $\cG^\Del$ being accepted, we know that protocols $\Pi^{\cG^\Del/\Gamma}$ and $\Pi^{\cG^\Del}$ are statistically indistinguishable by the environment. We also know that $\Pi^{\cG^\Del}$ and $\cF^\Del$ are statistically indistinguishable by the environment, by assumption. Thus, by the second property of deletion composability, it follows that $\Pi^{\cG^\Del / \Gamma}$ and $\cF^\Del$ are statistically indistinguishable by the environment conditioned on $\mathsf{DelRes} = 1$, completing the proof.
\end{proof}

\subsection{One-sided ideal commitments}\label{subsec:one-sided}

In this section, we construct what we call a \emph{one-sided ideal commitment with EST}. In the following subsection, we define this primitive as well as some underlying building blocks.

\subsubsection{Definitions and building blocks}
A one-sided ideal commitment with EST satisfies full-fledged security with EST against a malicious committer, but not against a malicious receiver. This commitment satisfies the weaker property of certified everlasting hiding against a malicious receiver. These properties are formalized below, where we denote by $\cF_{\mathsf{Com}}^{\mathsf{Del}}$ the commitment ideal functionality from \proref{fig:comIF} augmented with the delete phase from \proref{fig:delPhase}.

\begin{definition}[One-sided ideal commitment with EST]
\label{def:one-sided}

A three-phase (Commit, Reveal, Delete) commitment scheme is a \emph{one-sided ideal commitment with EST} if
\begin{enumerate}
    \item It computationally securely realizes $\cF_{\mathsf{Com}}^{\mathsf{Del}}$ (\cref{def:securerealization}) against a corrupt committer $C$.
    \item It satisfies statistical security against a corrupt committer $C$ that does not initiate deletion (first part of \cref{def:security-with-EST}). 
    
    \item It satisfies correctness of deletion (\cref{def:correctness-deletion}) and it satisfies certified everlasting hiding (\cref{def:CEH}) against adversaries that corrupt the receiver $R$.
\end{enumerate}
\end{definition}

To construct this object, our building block will be a \emph{computationally-hiding statistically-efficiently-extractable} (CHSEE) commitment, which is a two-phase (Commit, Reveal) commitment that satisfies correctness (\cref{def:correctness-decommitment}), standard computational hiding (\cref{def:comp-hiding}), and the following notion of binding. Note that this is similar to \cref{def:SB}, except that the extractor must be \emph{efficient}.

\begin{definition}[Statistical efficient extractability]\label{def:SEE}
A commitment scheme satisfies \emph{statistical efficient extractability} if for any QPT adversary $\{\cC^*_{\Com,\secp}\}_{\secp \in \bbN}$ in the Commit phase, there exists a QPT extractor $\cE = \{\cE_\secp\}_{\secp \in \bbN}$, such that for any initial advice $\{\ket{\psi_\secp}^{\sAux,\sC^*}\}_{\secp \in \bbN}$ and any QPT adversary $\{\cC^*_{\Rev,\secp}\}_{\secp \in \bbN}$ in the Reveal phase,
\[\TD\left(\Real^{\cC^*}_\secp,\Ideal^{\cC^*,\cE}_\secp\right) = \negl(\secp),\] where $\Real^{\cC^*}_\secp$ and $\Ideal^{\cC^*,\cE}_\secp$ are defined as follows.

\begin{itemize}
    \item $\Real_\secp^{\cC^*}$: Execute the Commit phase $(\sC^*,\sR) \gets \Com\langle \cC^*_{\Com,\secp}(\sC^*),\cR_{\Com,\secp}\rangle$, where $\cC^*_{\Com,\secp}$ has as input the $\sC^*$ register of $\ket{\psi_\secp}^{\sAux,\sC^*}$. Execute the Reveal phase to obtain a trit $\mu \gets \Rev\langle \cC^*_{\Rev,\secp}(\sC^*),\cR_{\Rev,\secp}(\sR)\rangle$ along with the committer's final state on register $\sC^*$. Output $(\mu,\sC^*,\sAux)$, which includes the $\sAux$ register of the original advice state.
    \item $\Ideal_\secp^{\cC^*,\cE}$: Run the extractor $(b^*,\sC^*,\sR) \gets \cE_\secp(\sC^*)$, where the extractor takes as input the $\sC^*$ register of $\ket{\psi_\secp}^{\sAux,\sC^*}$, and outputs a bit $b^*$ and a state on registers $\sC^*,\sR$. Next, execute the Reveal phase to obtain a trit $\mu \gets \Rev\langle \cC^*_{\Rev,\secp}(\sC^*),\cR_{\Rev,\secp}(\sR)\rangle$ along with the committer's final state on register $\sC^*$. If $\mu \in \{\bot,b^*\}$ output $(\mu,\sC^*,\sAux)$, and otherwise output a special symbol $\Fail$.
\end{itemize}

\end{definition}

\begin{importedtheorem}[\cite{BCKM2021}]\label{thm:CHSEE}
There exists a construction of CHSEE commitments that makes black-box use of any computationally-hiding statistically-binding commitment (\cref{def:comp-hiding} and \cref{def:SB}).\footnote{In \cite{BCKM2021}, a different notion of statistical binding for the underlying commitment was used, but it was noted by \cite{cryptoeprint:2021:1663} that the extractor-based definition of statistical binding suffices.}
\end{importedtheorem}
These are implied by OT with statistical security against one party, which was constructed in~\cite{BCKM2021}, based on the black-box use of computationally-hiding statistically-binding commitments.
Alternatively, CHSEE commitments can be obtained more directly by plugging in the statistically-equivocal computationally-extractable commitments from \cite{BCKM2021} into the extractability compiler \cite[Section 5]{BCKM2021}. Furthermore, this implies that CHSEE commitments can be based on the black-box use of one-way functions~\cite{BCKM2021} or pseudo-random quantum states~\cite{cryptoeprint:2021:1691,cryptoeprint:2021:1663}. 

\subsubsection{Construction}
We construct one-sided ideal commitments with EST from CHSEE commitments in \proref{fig:one-sided}. We note that constructing one-sided ideal commitments with EST does not just follow immediately from applying our certified deletion compiler, as in our construction of commitments with certified everlasting hiding from statistically-binding commitments in \cref{subsec:com-ZK}. The reason is that we need indistinguishability between the real and ideal worlds to hold against a malicious committer \emph{even if the delete phase is run}. To satisfy this property, we actually use additional invocations of the CHSEE commitment going in the ``opposite'' direction during the Delete phase of \proref{fig:one-sided}.

Next, we prove the following theorem.
\begin{theorem}\label{thm:one-sided}
\proref{fig:one-sided} is a one-sided ideal commitment with EST (according to \cref{def:one-sided}).
\end{theorem}
The theorem follows by combining Lemmas \ref{lem:os-id-one}, \ref{lem:os-id-two}, \ref{lem:os-idthree} and \ref{lem:os-idfour} proved below.

\protocol{\proref{fig:one-sided}: One-sided ideal commitment with EST}{Construction of one-sided ideal commitment with EST, from a CHSEE commitment.}{fig:one-sided}{
Ingredients: a CHSEE commitment $(\Com,\Rev)$ \\
Parties: committer $C$ with input $b \in \{0,1\}$ and receiver $R$.\\

\underline{\textbf{Commit phase}} 
\begin{itemize}
    \item $C$ samples $x,\theta \gets \{0,1\}^\secp$.
    \item $C$ and $R$ execute $\sC_\theta,\sR_\theta \gets \Com\langle C(\theta,b \oplus \bigoplus_{i:\theta_i = 0} x_i),R\rangle$.
    \item $C$ sends $\ket{x}_\theta$ to $R$ on register $\sX$.
\end{itemize}

\underline{\textbf{Reveal phase}} 
\begin{itemize}
    \item $C$ and $R$ execute $(\theta,b') \gets \Rev\langle C(\sC_\theta),R(\sR_\theta)\rangle$.
    \item  $R$ measures the qubits $i$ of register $\sX$ such that $\theta_i = 0$ to obtain $x_i$, and then outputs $b' \oplus \bigoplus_{i:\theta_i = 0} x_i$.
\end{itemize}

\underline{\textbf{Delete phase}\footnote{Note that both the Reveal phase and the Delete phase require $C$ and $R$ to run $\Rev\langle C(\sC_\theta),R(\sR_\theta)\rangle$.
So if the Reveal phase has already been run, we can instruct $R$ to abort if a deletion is requested, since we don't require any correctness of deletion or everlasting security after Reveal.}} 
\begin{itemize}
    \item $R$ measures all qubits of register $\sX$ in the Hadamard basis to obtain a string $x' \in \{0,1\}^\secp$.
    \item $R$ and $C$ execute $\secp$ Commit phases of $\Com$, with $R$ as the committer, committing bit-by-bit to $x'$: $\sR_{x',i},\sC_{x',i} \gets \Com \langle R(x'_i),C\rangle$.
    \item $C$ and $R$ execute $(\theta,b') \gets \Rev\langle C(\sC_\theta),R(\sR_\theta)\rangle$.
    \item $R$ and $C$ execute the Reveal phase of $\Com$ for each $i$ such that $\theta_i = 1$: $x'_i \gets \Rev \langle R(\sR_{x',i}),C(\sC_{x',i})\rangle$.
    \item $C$ accepts (outputs 1) if $x'_i = x_i$ for all $i$ such that $\theta_i = 1$.
\end{itemize}

}

\subsubsection{Security against a corrupt committer}
\begin{lemma}
\label{lem:os-id-one}
\proref{fig:one-sided} computationally securely realizes $\cF_{\mathsf{Com}}^{\mathsf{Del}}$ (\cref{def:securerealization}) against a corrupt committer $C$. That is, for every QPT committer $\{\cC^*_\secp \coloneqq (\cC^*_{\secp,\Com},\cC^*_{\secp,\Rev},\cC^*_{\secp,\Del})\}_{\secp \in \bbN}$, there exists a QPT simulator $\{\cS_\secp \coloneqq (\cS_{\secp,\Com},\cS_{\secp,\Rev},\cS_{\secp,\Del})\}_{\secp \in \bbN}$, such that for any QPT environment $\{\cZ_{\secp} \coloneqq (\cZ_{\secp,1},\cZ_{\secp,2},\cZ_{\secp,3})\}_{\secp \in \bbN}$, polynomial-size family of advice $\{\ket{\psi_\secp}\}_{\secp \in \bbN}$, and QPT distinguisher $\{\cD_\secp\}_{\secp \in \bbN}$, it holds that \[\bigg|\Pr\left[\cD_\secp\left(\Pi_{\cF_\Com^\Del}[\cC^*_\secp,\cZ_\secp,\ket{\psi_\secp}]\right) = 1\right] - \Pr\left[\cD_\secp\left(\widetilde{\Pi}_{\cF_\Com^\Del}[\cS_\secp,\cZ_\secp,\ket{\psi_\secp}]\right) = 1\right]\bigg| = \negl(\secp).\] 
\end{lemma}

\begin{lemma}
\label{lem:os-id-two}
\proref{fig:one-sided} satisfies statistical security against a corrupt committer $C$ that does not initiate deletion. That is, there exists a polynomial $p(\cdot)$ such that for every (potentially unbounded) committer $\{\cC^*_\secp \coloneqq (\cC^*_{\Com,\secp},\cC^*_{\Rev,\secp},\cC^*_{\Del,\secp})\}_{\secp \in \bbN}$, there exists a simulator $\{\cS_\secp \coloneqq (\cS_{\secp,\Com},\cS_{\secp,\Rev},\cS_{\secp,\Del})\}_{\secp \in \bbN}$ with size at most $p(\secp)$ times the of size of $\cC^*$, such that for any (potentially unbounded) environment $\{\cZ_{\secp} \coloneqq (\cZ_{\secp,1},\cZ_{\secp,2},\cZ_{\secp,3})\}_{\secp \in \bbN}$ and polynomial-size family of advice $\{\ket{\psi_\secp}\}_{\secp \in \bbN}$, \[\TD\left(\Pi_{\cF_{\mathsf{Com}}^\Del}^{\mathsf{DelReq} = 0}[\cC^*_\secp,\cZ_\secp,\ket{\psi_\secp}], \widetilde{\Pi}_{\cF_{\mathsf{Com}}^\Del}^{\mathsf{DelReq} = 0}[\cS_\secp,\cZ_\secp,\ket{\psi_\secp}]\right) = \negl(\secp),\] where $\Pi^{\mathsf{DelReq} = 0}_{\cF_{\mathsf{Com}}^\Del}[\cC^*_\secp,\cZ_\secp,\ket{\psi_\secp}]$ is defined to equal $\Pi_{\cF_{\mathsf{Com}}^\Del}[\cC^*_\secp,\cZ_\secp,\ket{\psi_\secp}]$ if the receiver's output $\mathsf{DelReq}$ is set to 0, and defined to be $\bot$ otherwise, and likewise for $ \widetilde{\Pi}_{\cF_{\mathsf{Com}}^\Del}^{\mathsf{DelReq} = 0}[\cS_\secp,\cZ_\secp,\ket{\psi_\secp}]$.
\end{lemma}

\begin{proof} (of Lemmas \ref{lem:os-id-one} and \ref{lem:os-id-two})
We define a simulator $\cS = (\cS_\Com,\cS_\Rev,\cS_\Del)$ based on any adversary $\cC^* = (\cC^*_{\Com}, \cC^*_{\Rev}, \cC^*_{\Del})$ that will suffice to prove both lemmas. We have dropped the dependence on $\secp$ for notational convenience.
\begin{enumerate}
    \item {\bf Commit Phase.} $\cS_\Com$ does the following.
    \begin{itemize}
        \item Run the CHSEE extractor $((\theta^*,d^*),\sC^*,\sR_\theta) \gets \cE_\secp[\cC^*_{\Com,\theta}](\sC^*)$, where the extractor is defined based on the part of $\cC^*_\Com$ that interacts in the commitment to $(\theta,b \oplus \bigoplus_{i: \theta_i = 0} x_i)$. It takes as input $\cC^*_\Com$'s private state register $\sC^*$, and outputs a sequence of committed bits $(\theta^*,d^*)$ and a state on $\sC^*,\sR_\theta$.
        \item Keep running $\cC^*_\Com$ until it outputs a state on register $\sX$.
        \item Measure the qubits $i$ of register $\mathsf{X}$ such that $\theta^*_i = 0$ to obtain $x^*_i$, and then send 
        $(\text{Commit}, \mathsf{sid}, b^*)$ where
        $b^* = d^* \oplus \bigoplus_{i:\theta^*_i = 0} x^*_i$ to the ideal functionality.
    \end{itemize}
    \item {\bf Reveal Phase.}
    $\cS_\Rev$ does the following.
    \begin{itemize}
        \item Execute the Reveal phase of CHSEE to obtain $(\theta',d') \leftarrow \Rev\langle \cC^*_{\Rev,\theta}(\sC^*),R(\sR_\theta)\rangle$(and update the register $\sC^*$).
        \item If $\theta^* = \theta'$ and $d^* = d'$ then send $(\text{Reveal}, \mathsf{sid})$ to the ideal functionality.
    \end{itemize}
    \item {\bf Delete Phase.}
    If $\cC^*_\Del$ initializes the Delete phase, $\cS_\Del$ sends $(\mathsf{DelRequest},\mathsf{sid})$ to the ideal functionality, and upon obtaining $(\mathsf{DelResponse},\mathsf{sid})$, it does the following.
    \begin{itemize}
        \item For every $i$ such that $\theta^*_i = 1$, measure the $i^{th}$ qubit of register $\sX$ in the Hadamard basis to obtain $x'_i$.
        For all $i$ such that $\theta^*_i \in \{0,\bot\}$, set $x'_i = 0$.
        \item Execute $\lambda$ commit phases of $\Com$, with the simulator as the committer, committing bit-by-bit to $x'$.
        \item Execute the Reveal phase of $\cC^*$'s commitments to obtain $(\theta',b') \leftarrow \Rev\langle \cC^*_{\Del,\theta}(\sC^*),R(\sR_\theta)\rangle$ (and update the register $\sC^*$), where $\cC^*_{\Del,\theta}$ is the part of $\cC^*_\Del$ that interacts in the reveal phase of its commitments.
        \item If $\theta' \neq \theta^*$, abort. Otherwise, execute with $\cC^*$ the Reveal phase of commitments to $x'$ restricted to indices $i \in [\secp]$ such that $\theta'_i = 
        1$.
    \end{itemize}
\end{enumerate}
It is straightforward to see that the simulator runs in quantum polynomial time as long as $\cC^*$ runs in quantum polynomial time.

Statistical indistinguishability between the real and ideal distributions at the end of the Commit Phase or the Reveal Phase follows directly from \cref{def:SEE}, thereby proving Lemma \ref{lem:os-id-two}.

Furthermore, in the Delete phase, simulator and receiver strategies are identical on indices where $\theta^*_i = 1$. The only difference between these strategies is that the simulator commits to $0$ when $\theta^*_i = 0$ whereas the receiver commits to outcomes of measurements of the $i^{th}$ qubit of register $\sX$ in the Hadamard basis. Now, the Reveal phase for these commitments are only run when $\theta' = \theta^*$, and only restricted to indices $i \in [\lambda]$ such that $\theta_i' = 1$. Thus, computational indistinguishability during the Delete phase follows from a reduction to the computational hiding of $\Com$. This proves \cref{lem:os-id-one}.
\end{proof}

\subsubsection{Security against a corrupt receiver}
\begin{lemma}
\label{lem:os-idthree}
\proref{fig:one-sided} satisfies correctness of deletion (\cref{def:correctness-deletion}). 
\end{lemma}
\begin{proof}
This follows immediately from the description of the scheme.
\end{proof}

\begin{lemma}
\label{lem:os-idfour}
\proref{fig:one-sided} satisfies certified everlasting hiding (\cref{def:CEH}) against adversaries that corrupt the receiver $R$.
\end{lemma}
\begin{proof}
The first property of certified everlasting hiding follows immediately from the computational hiding of CHSEE and the fact that the delete phase is completely independent of the committed bit $b$.

The second property follows from the computational hiding of CHSEE and \cref{thm:main} by setting $\cZ_\secp(\theta,b',\sA)$ and $\cA_\secp$ as follows, based on any non-uniform corrupt receiver $\cR^* = \{\cR^*_{\secp,\Com},\cR^*_{\secp,\Del},\ket{\psi_\secp}\}_{\secp \in \bbN}$.

\begin{itemize}
    \item $\cZ_\secp(\theta,b',\sA)$ initializes registers $(\sR^*,\sD^*)$ with $\ket{\psi_\secp}$, runs $(\sC_\theta,\sR^*) \gets \Com\langle C(\theta,b'),\cR^*_{\secp,\Com}(\sR^*)\rangle$ with the first part of $\cR^*_{\Com,\secp}$, and outputs the resulting state on registers $(\sR^*,\sA)$.
    \item $\cA_\secp$ receives registers $(\sR^*,\sA)$ and runs $\cR^*_{\secp,\Del}(\sR^*,\sA)$ until the beginning of the part where $\cR^*_{\secp,\Del}$ is supposed to commit to $x'$. At this point, it runs the extractor $\cE_\secp(\sR^*)$ for $\Com$, which outputs a certificate $x'$ and a left-over quantum state on register $\sR^*$.
    
\end{itemize}




Note that the delete phase succeeds in the experiment $\EVEXP_\secp^{\cR^*}(b)$ iff for every $i$ where $\theta_i = 1$, $x'_i = x_i$, where the $x'_i$ are opened by $\cR^*_{\secp,\Del}$. Also, by statistical efficient extractability of Com, the $\{x'_i\}_{i: \theta_i = 1}$ output by $\cE_\secp$ are equal to the $\{x'_i\}_{i: \theta_i = 1}$ opened by $\cR^*_{\secp,\Del}$, except with negligible probability. Thus, \cref{thm:main} implies that

\[\TD\left(\EVEXP_\secp^{\cR^*}(0),\EVEXP_\secp^{\cR^*}(1)\right) = \negl(\secp).\]
This completes the proof of the lemma.
\end{proof}

\subsection{Ideal commitments}\label{subsec:ideal}

In this section, we show how to generically upgrade a one-sided ideal commitment with EST to a full-fledged ideal commitment with EST. Our construction, which is given in \proref{fig:ideal}, is essentially the ``equivocality compiler'' from \cite{BCKM2021} with an added Delete phase.

\begin{theorem}\label{thm:ideal}
\proref{fig:ideal} securely realizes the commitment ideal functionality with EST (according to \cref{def:security-with-EST}).
\end{theorem}
The theorem follows by combining Lemmas \ref{lem:ts-id-one}, \ref{lem:ts-id-two}, \ref{lem:ts-idthree} and \ref{lem:ts-idfour} proved below.

\protocol
{\proref{fig:ideal}: Ideal commitment with EST}
{Ideal commitment with EST, from a one-sided ideal commitment with EST.}
{fig:ideal}
{
Ingredients: a one-sided ideal commitment with EST $(\Com,\Rev,\Del)$.\\
Parties: committer $C$ with input $b \in \{0,1\}$ and receiver $R$.\\

\vspace{-4mm}
\underline{\textbf{Commit phase}}
\begin{enumerate}
\item $C$ samples uniformly random bits $a_{i,j}$ for $i \in [\lambda]$ and $j \in \{0,1\}$.
\item For every $i \in [\lambda]$, $C$ and $R$ sequentially perform the following steps.
\begin{enumerate}
\item $C$ and $R$ execute four Commit phases sequentially, namely:
\begin{itemize}
    \item $\sC_{i,0,0}, \sR_{i,0,0} \leftarrow \Com \langle C(a_{i,0}), R \rangle$,
    \item $\sC_{i,0,1}, \sR_{i,0,1} \leftarrow \Com \langle C(a_{i,0}), R \rangle$, 
    \item $\sC_{i,1,0}, \sR_{i,1,0} \leftarrow \Com \langle C(a_{i,1}), R \rangle$,
    \item $\sC_{i,1,1}, \sR_{i,1,1} \leftarrow \Com \langle C(a_{i,1}), R \rangle$.
\end{itemize}
\item $R$ sends a choice bit $c_i \leftarrow \{0,1\}$.
\item $C$ and $R$ execute two Reveal phases, obtaining the opened bits:
\begin{itemize}
\item $u \leftarrow \Rev \langle C(\sC_{i,c_i,0}), R(\sR_{i,c_i,0}) \rangle$,
\item $v \leftarrow \Rev \langle C(\sC_{i,c_i,1}), R(\sR_{i,c_i,1}) \rangle$. 
\end{itemize}
If $u \neq v$, $R$ aborts. Otherwise, $C$ and $R$ continue.
\end{enumerate}
\vspace{-4mm}
\item For $i \in [\lambda]$, 
$C$ sets $b_i = b \oplus a_{i, 1-c_i}$ and
sends $\{b_i\}_{i \in [\lambda]}$ to $R$.
\end{enumerate}

\vspace{-2mm}
\underline{\textbf{Reveal phase}}
\vspace{-2mm}
\begin{enumerate}
    \item $C$ sends $b$ to $R$. In addition,
    \vspace{-2mm}
    \begin{enumerate}
    \item For $i \in [\lambda]$, 
    $C$ picks $\alpha_i \gets \{0,1\}$ and sends it to $\cR$. 
    \item $C$ and $R$ execute  $a_i' \leftarrow \Rev \langle C(\sC_{i,1-c_i,\alpha_i}), R(\sR_{i,1-c_i,\alpha_i}) \rangle$.
    \end{enumerate}
    \vspace{-4mm}
    \item $R$ accepts and outputs $b$ if for every $i \in [\lambda]$, $a_i' = b \oplus b_i$.
\end{enumerate}

\vspace{-2mm}
\underline{\textbf{Delete phase}}
\vspace{-2mm}
\begin{enumerate}
    \item For every $i \in [\secp]$, $C$ and $R$ sequentially perform the following steps.
    \vspace{-2mm}
    \begin{enumerate}
        \item If Reveal was performed, execute $D_i \gets \Del\langle C(\sC_{i,1-c_i,1-\alpha_i},R(\sR_{i,1-c_i,1-\alpha_i})\rangle$.
        \item Otherwise, set $D_i = D_{i,0} \wedge D_{i,1}$ where
        $D_{i,0} \gets \Del\langle C(\sC_{i,1-c_i,0},R(\sR_{i,1-c_i,0})\rangle$ and $D_{i,1} \gets \Del\langle C(\sC_{i,1-c_i,1},R(\sR_{i,1-c_i,1})\rangle$.
    \end{enumerate}
    \item If $D_i = 1$ for all $i \in [\secp]$, then $C$ outputs 1.
\end{enumerate}
}

\subsubsection{Security against a corrupt committer}

\begin{lemma}
\label{lem:ts-id-one}
\proref{fig:ideal} computationally securely realizes $\cF_{\mathsf{Com}}^{\mathsf{Del}}$ (\cref{def:securerealization}) against a corrupt committer $C$.
\end{lemma}

\begin{lemma}
\label{lem:ts-id-two}
\proref{fig:ideal} satisfies statistical security against a corrupt committer $C$ that does not initiate deletion. That is, there exists a polynomial $p(\cdot)$ such that for every (potentially unbounded) adversary $\{\cC^*_\secp \coloneqq (\cC^*_{\secp,\Com},\cC^*_{\secp,\Rev},\cC^*_{\secp,\Del})\}_{\secp \in \bbN}$ corrupting $C$, there exists a simulator $\{\cS_\secp \coloneqq (\cS_{\secp,\Com},\cS_{\secp,\Rev},\cS_{\secp,\Del})\}_{\secp \in \bbN}$ with size at most $p(\secp)$ times the of size of $\{\cC^*_\secp\}_{\secp \in \bbN}$, such that for any (potentially unbounded) environment $\{\cZ_{\secp} \coloneqq (\cZ_{\secp,1},\cZ_{\secp,2},\cZ_{\secp,3})\}_{\secp \in \bbN}$ and polynomial-size family of advice $\{\ket{\psi_\secp}\}_{\secp \in \bbN}$, \[\TD\left(\Pi_{\cF_{\mathsf{Com}}^\Del}^{\mathsf{DelReq} = 0}[\cC^*_\secp,\cZ_\secp,\ket{\psi_\secp}], \widetilde{\Pi}_{\cF_{\mathsf{Com}}^\Del}^{\mathsf{DelReq} = 0}[\cS_\secp,\cZ_\secp,\ket{\psi_\secp}]\right) = \negl(\secp),\] where $\Pi^{\mathsf{DelReq} = 0}_{\cF_{\mathsf{Com}}^\Del}[\cC^*_\secp,\cZ_\secp,\ket{\psi_\secp}]$ is defined to equal $\Pi_{\cF_{\mathsf{Com}}^\Del}[\cC^*_\secp,\cZ_\secp,\ket{\psi_\secp}]$ if the receiver's output $\mathsf{DelReq}$ is set to 0, and defined to be $\bot$ otherwise, and likewise for $ \widetilde{\Pi}_{\cF_{\mathsf{Com}}^\Del}^{\mathsf{DelReq} = 0}[\cS_\secp,\cZ_\secp,\ket{\psi_\secp}]$.
\end{lemma}

\begin{proof} (of Lemmas \ref{lem:ts-id-one} and \ref{lem:ts-id-two})\\

\noindent {\bf \underline{The Simulator.}} The simulator $(\cS_{\Com},\cS_{\Rev}, \cS_{\Del})$ is defined as follows. 
\begin{enumerate}
    \item {\bf Commit Phase.} $\cS_\Com$ does the following.
    \begin{itemize}
        \item For all $i \in [\secp]$,
        \begin{itemize}
            \item Execute four sequential simulated Commit phases where the simulator for the commitment $\Com$ is run on the part of the committer $\cC^*_\Com$ participating in each of the four sequential sessions. Denote the bit output by the simulator in each session by $(d_{i,0,0},d_{i,0,1},d_{i,1,0},d_{i,1,1})$.
            \item Sample and send choice bit $c_i \leftarrow \{0,1\}$ to $\cC^*_\Com$.
            \item Execute two simulated Reveal phases where the simulator is run on the part of the committer $\cC^*_\Com$ corresponding to sessions $(i, c_i, 0)$ and $(i, c_i, 1)$. If the simulator outputs $(\text{Reveal},\mathsf{sid})$ for both sessions and $d_{i,c_i,0} = d_{i,c_i,1}$, continue, and otherwise abort.
        \end{itemize}
        \item Obtain $\{b_i\}_{i \in [\secp]}$ from $\cC^*_\Com$.
        Fix $b^*$ to be the most frequently occuring bit in $\{b_i \oplus d_{i,1-c_i,0}\}_{i \in [\secp]}$. Send $(\text{Commit},\mathsf{sid},b^*)$ to the commitment ideal functionality.
    \end{itemize}
    \item {\bf Reveal Phase.}
    $\cS_\Rev$ does the following.
    \begin{enumerate}
        \item Obtain $b$ from $\cC^*_\Rev$. Additionally, for $i \in [\secp]$,
        \begin{itemize}
            \item Obtain $\alpha_i$ from $\cC^*_\Rev$.
            \item Execute the simulated Reveal phase where simulator is run on the part of the committer $\cC^*_\Rev$ corresponding to session $(i, 1-c_i, \alpha_i)$. If $\cS_\Rev$ outputs $(\text{Reveal},\mathsf{sid})$ and  $d_{i,1-c_i,\alpha_i} = b \oplus b_i$, continue. Otherwise, abort.
        \end{itemize}
        \item Send $(\text{Reveal},\mathsf{sid})$ to the ideal functionality.
    \end{enumerate}
    \item {\bf Delete Phase.}
   If $\cC^*$ makes a delete request, $\cS_\Del$ sends $(\mathsf{DelRequest},\mathsf{sid})$ to the ideal functionality, and upon obtaining $(\mathsf{DelResponse},\mathsf{sid})$, it does the following.
   \begin{itemize}
       \item If the Reveal phase was executed, then for every $i \in [\secp]$, run the simulator on the part of $\cC^*_\Del$ that interacts in the delete phase of session $(i, 1-c_i, 1-\alpha_i)$.
       \item If the Reveal phase was not executed, then for every $i \in [\secp]$, run the simulator on the part of $\cC^*_\Del$ that interacts (sequentially) in the delete phases of sessions $(i, 1-c_i, 0)$ and $(i, 1-c_i, 1)$.
   \end{itemize}
\end{enumerate}

\noindent{\bf \underline{Analysis.}\\}

Note that there are a total of $4\lambda$ commitment sessions. Denote the real experiment by $\mathsf{Hybrid}_{0,1,1}$.
For each $i \in [\lambda], j \in [0,1], k \in [0,1]$, define $\mathsf{Hybrid}_{i,j,k}$ to be the distribution obtained as follows.

\paragraph{Commit Phase.}
Set $\counternew = 1$, $\mathsf{DelReq}= 0$ and do the following:
\begin{enumerate}
    \item If $\counternew = \lambda + 1$, obtain $\{b_i\}_i$ from $\cC^*_\Com$ and end.
    \item If $\counternew < i$,
    \begin{enumerate}
    \item Execute four sequential simulated Commit phases where the simulator for the commitment $\Com$ is run on the part of the committer $\cC^*_\Com$ participating in each of the four sequential sub-sessions.
    Denote the bit output by the the simulator in each sub-session respectively by $(d_{\counternew,0,0},d_{\counternew,0,1},d_{\counternew,1,0},d_{\counternew,1,1})$.
    \item Sample and send choice bit $c_\counternew \leftarrow \{0,1\}$ to $\cC^*_\Com$.
    \item Execute two simulated Reveal phases where the simulator is run on the part of the committer $\cC^*_\Com$ corresponding to sub-sessions $(\counternew, c_\counternew, 0)$ and $(\counternew, c_\counternew, 1)$. If the simulator outputs $(\text{Reveal},\mathsf{sid})$ for both sub-sessions and $d_{\counternew,c_\counternew,0} = d_{\counternew,c_\counternew,1}$, continue, and otherwise abort.
    \end{enumerate}
    \item If $\counternew = i$, 
    \begin{enumerate}
        \item 
        Do the same as above (i.e., for the case $\counternew <i$) except execute simulated Commit (and if needed, Reveal) phases where the simulator for the commitment $\mathsf{Com}$ is run on the part of the committer $\cC^*_\Com$ participating in  sequential sub-sessions
        $(\counternew,j',k')$ whenever $(j',k') \leq (j,k)$ where we have $(0,0) \leq (0,1) \leq (1,0) \leq (1,1)$ for transitive relation $\leq$. But for $(j',k') \not \leq (j,k)$, follow honest receiver strategy in sub-session $(i,j',k')$.
    \end{enumerate}
    \item If $\counternew > i$,
    \begin{enumerate}
        \item Execute honest receiver strategy for all Commit (and Reveal) phases for all sessions $(i, j', k')$ for every $j', k' \in \{0,1\}^2$.
    \end{enumerate}
    \item Set $\gamma = \gamma + 1$.
\end{enumerate}
\paragraph{Reveal Phase.}
    Do the following.
    \begin{itemize}
        \item Obtain $b$ from $\cC^*_\Rev$. Additionally, for $\counternew \in [\secp]$,
        \begin{itemize}
            \item Obtain $\alpha_\gamma$ from $\cC^*_\Rev$.
            \item If $\counternew < i$, execute the simulated Reveal phase where the simulator is run on the part of the committer $\cC^*_\Rev$ corresponding to session $(\gamma, 1-c_\gamma, \alpha_\gamma)$.
            If the simulator outputs $(\text{Reveal},\mathsf{sid})$ and if $d_{\gamma,1-c_\gamma,\alpha_\gamma} = b \oplus b_i$, continue. Otherwise, abort.
            \item If $\gamma = i$, do the same as above when $(1-c_\gamma,\alpha_\gamma) \leq (j,k)$ otherwise follow honest receiver strategy.
            If $\gamma > i$, follow honest receiver strategy.
        \end{itemize}
        \item Set $b^* = b$. 
    \end{itemize}

\noindent {\bf Delete Phase.}
   If $\cC^*$ makes a delete request, send $(\mathsf{DelRequest},\mathsf{sid})$ to the ideal functionality, and upon obtaining $(\mathsf{DelResponse},\mathsf{sid})$, do the following.
   \begin{itemize}
        \item If the Reveal phase was executed, then
        \begin{itemize}
           \item For every $\gamma \in [1,i-1]$, run the simulator on the part of $\cC^*_\Del$ that interacts in the delete phase of session $(\gamma, 1-c_\gamma, 1-\alpha_\gamma)$.
           \item For $\gamma = i$, if $(1-c_\gamma, 1-\alpha_\gamma) \leq (j,k)$, run the simulator on the part of $\cC^*_\Del$ that interacts in the delete phase of session $(\gamma, 1-c_\gamma, 1-\alpha_\gamma)$. Otherwise run honest receiver strategy on session $(\gamma, 1-c_\gamma, 1-\alpha_\gamma)$.
           \item For $\gamma \in [i+1,\lambda]$, follow honest receiver strategy.
        \end{itemize}
        \item If the Reveal phase was not executed, then 
        \begin{itemize}
            \item For every $\gamma \in [1,i-1]$, run the simulator on the part of $\cC^*_\Del$ that interacts in the delete phases (sequentially) of $(\gamma, 1-c_\gamma, 0)$ and $(\gamma, 1-c_\gamma, 1)$.
            \item For $\gamma = i$, run the simulator on the part of $\cC^*_\Del$ that interacts in the delete phases (sequentially) of $(\gamma, 1-c_\gamma, b)$ for all $b \in \{0,1\}$ for which $(1-c_\gamma, b) \leq (j,k)$, and use honest receiver strategy on other sessions.
            \item For $\gamma \in [i+1,\lambda]$, follow honest receiver strategy.
        \end{itemize}
        \item Set $\mathsf{DelReq} = 1$.
   \end{itemize}
\underline{The output} of $\mathsf{Hybrid}_{i,j,k}$ is the final state of $\cC^*$ together with the bit $b^*$ (which is set to $\bot$ if the game aborted before $b^*$ was set), and the bit $\mathsf{DelReq}$.

We consider the interaction of $\cC^*$ with an honest receiver, and denote the state output by $\cC^*$ jointly with the bit output by the honest receiver in this interaction by $\mathsf{Hybrid}_{0,1,1}$. 
We now prove the following claim about consecutive hybrids.



%
\begin{claim}
There exists a negligible function $\mu(\cdot)$ such that for every $i \in [\lambda]$,  every $(\iota, j, k, \iota', j', k') \in \{(i - 1, 1, 1, i, 0, 0), (i, 0, 0, i, 0, 1), (i, 0, 1, i, 1, 0), (i, 1, 0, i, 1, 1)\}$,
\begin{itemize}
\item for every QPT distinguisher $\cD$,
$$|\Pr[\cD(\mathsf{Hybrid}_{\iota, j, k}) = 1] -  \Pr[\cD(\mathsf{Hybrid}_{\iota', j',k'}) = 1]| = \mu(\lambda)$$
\item and furthermore, for every unbounded distinguisher $\cD$,
$$|\Pr[\cD(\mathsf{Hybrid}_{\iota, j, k}^{\mathsf{DelReq}=0}) = 1] -  \Pr[\cD(\mathsf{Hybrid}_{\iota', j',k'}^{\mathsf{DelReq}=0}) = 1]| = \mu(\lambda)$$
where $\mathsf{Hybrid}_{\iota, j, k}^{\mathsf{DelReq}=0}$ is defined to be equal to $\mathsf{Hybrid}_{\iota,j,k}$ when $\mathsf{DelReq}$ is set to $0$, and defined to be $\bot$ otherwise, and likewise for $\mathsf{Hybrid}_{\iota', j',k'}^{\mathsf{DelReq}=0}$.
\end{itemize}
\end{claim}
\begin{proof} Suppose this is not the case. Then there exists an adversarial QPT committer $\cC^*$,  a polynomial $p(\cdot)$, and an initial committer state $\ket{\psi}$ that corresponds to a state just before the beginning of commitment $(\iota', j', k')$
where for some QPT distinguisher $\cD$,
\begin{equation}
\label{eq:hybdone}
\Pr[\cD(\mathsf{Hybrid}_{\iota, j, k}) = 1] -  \Pr[\cD(\mathsf{Hybrid}_{\iota', j',k'}) = 1]| \geq \frac{1}{p(\lambda)}.
\end{equation}
or for unbounded $\cC^*$ and some unbounded distinguisher $\cD'$,
\begin{equation}
\label{eq:hybdtwo}
|\Pr[\cD'(\mathsf{Hybrid}_{\iota, j, k}^{\mathsf{DelReq}=0}) = 1] -  \Pr[\cD'(\mathsf{Hybrid}_{\iota', j',k'}^{\mathsf{DelReq}=0}) = 1]| \geq \frac{1}{p(\lambda)}
\end{equation}
Consider a reduction/adversarial committer $\widetilde{\cC}$ that obtains initial state $\ket{\psi}$, then internally runs $\cC^*$, forwarding all messages between an external receiver and $\cC^*$ for the $(\iota', j', k')^{th}$ commitment session, while running all other sessions according to the strategy in $\mathsf{Hybrid}_{\iota,j,k}$. 
The commit phase then ends, 
and $\widetilde{C}$ initiates the opening phase with the external receiver.
Internally, it continues to run the remaining commit sessions with $\cC^*$ -- generating for it the messages on behalf of the receiver according to the strategy in $\mathsf{Hybrid}_{\iota, j, k}$. 
The only modification is that it forwards $\cC^*$'s opening of the $(\iota', j', k')^{th}$ commitment (if and when it is executed) to the external challenger.
Finally, $\widetilde{\cC}$ behaves similarly if there is a delete phase, i.e., it forwards $\cC^*$'s deletion request and any messages generated in the delete phase of the $(\iota', j', k')^{th}$ commitment between $\cC^*$ and the external challenger.

Then, equation~(\ref{eq:hybdone}) and equation~(\ref{eq:hybdtwo}) respectively contradict the security of one-sided ideal commitments with EST against the committer $\cC^*$ (\cref{def:one-sided}). More specifically, equation~(\ref{eq:hybdone}) contradicts the computationally secure realization of $\cF_{\Com}^{\Del}$ whereas equation~(\ref{eq:hybdtwo}) contradicts the statistical security of $\Com$ against a corrupt committer that does not initiate deletion. This completes the proof of the claim.
%
%
\end{proof}
To complete the proof of the two lemmas, we observe that 
the only difference between $\mathsf{Hybrid}_{\lambda,1,1}$ and $\mathsf{Ideal}$ is the way the bit $b^*$ (output by the honest receiver) is computed. In more detail, in  $\mathsf{Hybrid}_{\lambda,1,1}$, the bit $b^*$ is computed as the majority of $\{b_i \oplus d_{i,1-c_i,0}\}_{i \in [\secp]}$.
Now for every commitment strategy and every $i \in [\lambda]$, by correctness of extraction (which follows from the indistinguishability between real and ideal distributions for every commitment),
the probability that $d_{i,1-c_i,0} \neq d_{i,1-c_i,1}$ and yet the receiver does not abort in Step 2(c) in the $i^{th}$ sequential repetition, is $\leq \frac{1}{2} + \negl(\secp)$. 
Thus, this implies that the probability that $\mathsf{Hybrid}_{\lambda,1,1}$ and $\mathsf{Ideal}$ output different bits $b^*$ is at most $2^{-\secp/2} + \negl(\secp) = \negl(\secp)$, which implies that the two are statistically close. 

This, combined with the claim above, completes the proof.
\end{proof}

\subsubsection{Security against a corrupt receiver}
\begin{lemma}
\label{lem:ts-idthree}
\proref{fig:ideal} computationally securely realizes $\cF_{\mathsf{Com}}^{\mathsf{Del}}$ (\cref{def:securerealization}) against a corrupt $R$.
\end{lemma}

\begin{lemma}
\label{lem:ts-idfour}
\proref{fig:ideal} satisfies certified everlasting security against $R$.
That is, for every QPT adversary $\{\cR_\secp^* \coloneqq (\cR^*_{\secp,\Com}, \cR^*_{\secp,\Rev}, \cR^*_{\secp,\Del})\}_{\secp \in \bbN}$ corrupting party $R$, there exists a QPT simulator $\{\cS_\secp \coloneqq (\cS_{\secp,\Com},\cS_{\secp,\Rev},\cS_{\secp,\Del})\}_{\secp \in \bbN}$ such that for any QPT environment $\{\cZ_{\secp} \coloneqq (\cZ_{\secp,1},\cZ_{\secp,2},\cZ_{\secp,3})\}_{\secp \in \bbN}$, and polynomial-size family of advice $\{\ket{\psi_\secp}\}_{\secp \in \bbN}$, \[\TD\left(\Pi_{\cF_{\mathsf{Com}}^\Del}^{\mathsf{DelRes} = 1}[\cR^*_\secp,\cZ_\secp,\ket{\psi_\secp}], \widetilde{\Pi}_{\cF_{\mathsf{Com}}^\Del}^{\mathsf{DelRes} = 1}[\cS_\secp,\cZ_\secp,\ket{\psi_\secp}]\right) = \negl(\secp),\] where $\Pi^{\mathsf{DelRes} = 1}_{\cF_{\mathsf{Com}}^\Del}[\cR^*_\secp,\cZ_\secp,\ket{\psi_\secp}]$ is defined to be equal to $\Pi_{\cF_{\mathsf{Com}}^\Del}[\cR^*_\secp,\cZ_\secp,\ket{\psi_\secp}]$ if the committer's output $\mathsf{DelRes}$ is set to 1, and defined to be $\bot$ otherwise, and likewise for $ \widetilde{\Pi}_{\cF_\Com^\Del}^{\mathsf{DelRes} = 1}[\cS_\secp,\cZ_\secp,\ket{\psi_\secp}]$.
\end{lemma}
\begin{proof} (of Lemmas \ref{lem:ts-idthree} and \ref{lem:ts-idfour})\\

\noindent{\textbf{\underline{The simulator.}}}
The first stage of the simulator $\cS_{\Com}$, defined based on $\cR^*_\Com$, will be obtained via the use of the Watrous rewinding lemma (Lemma \ref{lem:qrl})~\cite{STOC:Watrous06}.
For the purposes of defining the simulation strategy, it will be sufficient (w.l.o.g.) to consider a restricted receiver $\cR^*_\Com$ that operates as follows in the $i^{th}$ sequential step of the commitment phase of the protocol.
In the simulation, the state of this step of $\cR^*_\Com$ will be initialized to the final state at the end of simulating the $(i-1)^{th}$ step.

\begin{enumerate}
\item $\cR^*_\Com$ takes a quantum register $\mathsf{W}$, representing its auxiliary quantum input. $\cR^*_\Com$ will use two additional quantum registers that function as work space: $\mathsf{V}$, which is an arbitrary (polynomial-size) register, and $\mathsf{A}$, which is a single qubit register. The registers $\mathsf{V}$ and $\mathsf{A}$ are initialized to the all-zero state before the protocol begins.

\item Let
$\textsf{M}$ denote the polynomial-size register used by the committer $C$ to send messages to $\cR^*_\Com$. After carrying out step 2(a) by running on registers $(\mathsf{W},\mathsf{V},\mathsf{A},\mathsf{M})$, $\cR^*_\Com$ measures the register $\mathsf{A}$ to obtain a bit $c_i$ for Step 2(b), which it sends back to $C$.

\item Next, $\cR^*_\Com$ computes the reveal phases (with messages from $C$ placed in register $\mathsf{M}$) according to Step 2(c). $\cR^*_\Com$ outputs registers $(\mathsf{W},\mathsf{V},\mathsf{A}, \mathsf{M})$.
\end{enumerate}

Any QPT receiver can be modeled as a receiver of this restricted form followed by some polynomial-time post-processing of the restricted receiver’s output. The same post-processing can be applied to the output of the simulator that will be constructed for the given restricted receiver.

Following~\cite{STOC:Watrous06}, we define a simulator that uses two additional registers, $\mathsf{C}$ and $\mathsf{Z}$, which are both initialized to the all-zero state. $\mathsf{C}$ is a one qubit register, while $\mathsf{Z}$ is an auxiliary register used to implement the computation that will be described next. Consider a quantum procedure $\cS_{\mathsf{partial}}$ that implements the strategy described in  \proref{fig:eqsimulator} using these registers.

\protocol
{\proref{fig:eqsimulator}}
{Partial Equivocal Simulator.}
{fig:eqsimulator}
{

\underline{\textbf{Circuit $\cS_{\mathsf{partial}}$.}}
\begin{enumerate}
\item Sample a uniformly random classical bit $\widehat{c}$, and store it in register $\mathsf{C}$.
\item Sample uniformly random bits $(z,d)$.
\item 
If $\widehat{c} = 0$, initialize committer input as follows, corresponding to four sequential sessions:
\begin{itemize}
    \item For the first two sessions, set committer input to $z$.
    \item For the third and fourth sessions, set committer input to $d$ and $1-d$ respectively.
\end{itemize}
\item 
If $\widehat{c} = 1$, initialize committer input as follows, corresponding to four sequential sessions:
\begin{itemize}
    \item For the first and second sessions, set committer input to $d$ and $1-d$ respectively.
    \item For the last two sessions, set committer input to $z$.
\end{itemize}
\item Run the commitment phase interaction between the honest committer and $\cR^*_\Com$'s sequence of unitaries on registers $(\sW,\sV,\sA,\sM)$ initialized as above.
\item Measure the qubit register $A$ to obtain a bit $c$. If $c=\widehat{c}$, output 0, otherwise output $1$.
\end{enumerate}
}

Next, we apply the Watrous rewinding lemma to the $\cS_{\mathsf{partial}}$ circuit to obtain a circuit $\widehat{\cS}_{\mathsf{partial}}$. To satisfy the premise of \cref{lem:qrl}, we argue that the probability $p(\ket{\psi})$ that $\cS_{\mathsf{partial}}$ outputs $0$ is such that $|p(\ket{\psi})-\frac{1}{2}| = \mathsf{negl}(\lambda)$, regardless of the auxiliary input $\ket{\psi}$ to the $i$'th sequential stage of $\cR^*_\Com$. This follows from the fact that the commitments are computationally hiding.
In more detail, by definition, Step 5 produces a distribution on $\cR^*_\Com$'s side that is identical to the distribution generated by $\cR^*_\Com$ in its interaction with the committer, who either has input $(z,z,d,1-d)$ (if $\widehat{c} = 0$) or input $(d,1-d,z,z)$ (if $\widehat{c} = 1$). If $|p(\ket{\psi})-\frac{1}{2}|$ were non-negligible, then the sequence of unitaries applied by $\cR^*_\Com$ could be used to distinguish commitments generated according to the case $\widehat{c} = 0$ from commitments generated according to the case $\widehat{c} = 1$, which would contradict the hiding of the commitment.

Now consider the residual state on registers $(\sW,\sV,\sA,\sM,\sC,\sZ)$ of $\cS_{\mathsf{partial}}$ conditioned on a measurement of its output register $\sA$ being $0$. The output state of $\widehat{\cS}_{\mathsf{partial}}$ will have negligible trace distance from the state on these registers. Now, the simulator $\cS_{\Com}$ must further process this state as follows.
\begin{itemize}
\item Measure the register $\mathsf{C}$, obtaining challenge $c$. Place the classical bits $(c,d)$ in the register $\sZ$, which also contains the current state of the honest committer algorithm.
\item Use information in register $\sZ$ to execute Step 2(c) of \proref{fig:ideal}.
\item Discard register $\sC$, re-define register $\sZ_i \coloneqq \sZ$ to be used later in the Reveal / Delete phases, and output registers $(\mathsf{W}, \mathsf{V}, \mathsf{A}, \mathsf{M})$ to be used in the next sequential step of the Commit phase. 
\end{itemize}


The simulator $\cS_\Com$ for the commit phase executes all $\lambda$ sequential interactions in this manner, and then samples $b_1, \ldots, b_{\lambda} \leftarrow \{0,1\}^{\lambda}$, as the committer messages for Step 3 of \proref{fig:ideal}. It then outputs the final state of $\cR^*_\Com$ on registers $(\mathsf{W}, \mathsf{V}, \mathsf{A}, \mathsf{M})$, and additionally outputs a private state on registers $(\sZ_1,\dots,\sZ_\secp)$, which consist of the honest committer's state after each of the $i$ sequential steps, as well as bits $(b_1,c_1,d_1,\dots,b_\secp,c_\secp,d_\secp)$.

The reveal stage of the simulator $\cS_\Rev$ takes as input a bit $b$, and a state on registers $(\sZ_1,\dots,\sZ_\secp,\allowbreak\sW,\allowbreak\sV,\sA,\sM)$, and does the following for each $i \in [\secp]$. 
\begin{itemize}
\item Let $\widehat{d}_i = b \oplus b_i$.
\item If $c_i = 0$, it executes the decommitment phase for the $((\widehat{d}_i \oplus d_i)+2)^{th}$ session with $\cR^*_\Rev$.
\item If $c_i = 1$, it executes the decommitment phase for the $(\widehat{d}_i \oplus d_i)^{th}$ session with $\cR^*_\Rev$.
\item Output $\cR^*_\Rev$'s resulting state. Note that each decommitment will be to the bit $\widehat{d}_i = b \oplus b_i$.
\end{itemize}
Finally, the simulator $\cS_\Del$ for the delete phase executes the honest committer's algorithm on the commitments that were not revealed above.\\

\noindent{\textbf{\underline{Analysis.}}}
%
Lemma \ref{lem:ts-idthree} follows from the computational hiding of the underlying commitment scheme $\Com$, via an identical proof to~\cite{BCKM2021}. We have already argued above that the distribution produced by $\cS_\Com$ is statistically close to the distribution that would result from conditioning on the output of $\cS_{\mathsf{partial}}$ being 0 in each sequential step. Thus, it remains to argue that this is computationally indistinguishable from the real distribution. If not, then there exists a session $i \in [\lambda]$ such that the distribution in the real experiments up to the ${i-1}^{th}$ session is indistinguishable, but up to the $i^{th}$ session is distinguishable. However, this directly contradicts the computational hiding of the underlying commitment scheme.


In what follows, we prove Lemma \ref{lem:ts-idfour}. This only considers executions where $\mathsf{DelRes} = 1$, i.e., executions where $\cR^*$ successfully completes the delete phase. We again consider a sequence of $\lambda$ intermediate hybrids between the real and ideal executions.
We will let $\mathsf{Hybrid}_0^{\mathsf{DelRes} = 1}$ denote the final state of $\cR^*$ in the real experiment when the honest party output $\mathsf{DelRes} = 1$ and $\bot$ otherwise. Let $\mathsf{Hybrid}_i$ denote the final state of $\cR^*$ when the first $i$ (out of $\lambda$) sequential commit sessions are simulated using the $\widehat{\cS}_{\mathsf{partial}}$ circuit, defined based on $\cS_{\mathsf{partial}}$ from \proref{fig:eqsimulator}. Let $\mathsf{Hybrid}_i^{\mathsf{DelRes} = 1}$ denote the output of $\mathsf{Hybrid}_i$ when the honest party output $\mathsf{DelRes} = 1$ and $\bot$ otherwise. 

For every $i \in [\lambda]$, statistical indistinguishability between $\mathsf{Hybrid}_{i-1}^{\mathsf{DelRes} = 1}$ and $\mathsf{Hybrid}_{i}^{\mathsf{DelRes} = 1}$ follows by a reduction to the certified everlasting security of $\Com$ (according to \cref{def:one-sided}), as follows. The reduction $\mathsf{Red}$ is different depending on whether or not the Reveal phase is executed.
\begin{itemize}
    \item \underline{Case 1: The Reveal Phase is not executed.}
    $\mathsf{Red}$ acts as receiver in one session of $\Com$, interacting with an external challenger. 
    $\mathsf{Red}$ samples a uniformly random bit $d$ and sends it to the challenger. The challenger samples a uniformly random bit $b'$. If $b' = 0$, the challenger participates as a committer in a commit session to $d$ and otherwise to $(1-d)$.
    
    $\mathsf{Red}$ internally follows the strategy in $\mathsf{Hybrid}_{i-1}$ in the Commit phase for sessions $1, \ldots, i-1$ and $i+1, \ldots, \lambda$, based on the adversary $\cR^*_\Com$. During the $i^{th}$ session, $\mathsf{Red}$ interacts with the challenger and the adversary. In particular, it runs the strategy $\cS_{\mathsf{partial}}$ from \proref{fig:eqsimulator}, with the following exception. For $\widehat{c}$ sampled uniformly at random, if $\widehat{c} = 0$, it 
    forwards messages between $\cR^*_\Com$ and the challenger for either the third or fourth commitment (sampled randomly) and commits to $d$ in the other session and otherwise forwards messages between $\cR^*_\Com$ and the challenger for either the first or second commitment (sampled randomly) and commits to $d$ in the other session. If $\cR^*_\Com$'s challenge $c_i = \widehat{c}$, $\mathsf{Red}$ continues the experiment, otherwise it aborts. 
    $\mathsf{Red}$ continues to follow the strategy in $\mathsf{Hybrid}_{i-1}$, except setting $b_i = b \oplus d$.
    Note that the 
    challenge commitment is never opened.
    
    In the Delete phase, $\mathsf{Red}$ again follows the strategy in $\mathsf{Hybrid}_{i-1}$ except that it executes the Delete phase for the (two) unopened commitments in the $i^{th}$ session, one that it generated on its own, and the other by forwarding messages between $\cR^*_\Del$ and the external challenger. 
    
    By computational hiding of the challenger's commitment, the probability that the reduction aborts is at most $\frac{1}{2} + \negl(\secpar)$. Furthermore, conditioned on not aborting, the distribution output by $\mathsf{Red}$ is identical to $\mathsf{Hybrid}_{i-1}^{\mathsf{DelRes} = 1}$ 
    when $b' = 0$ and is statistically close to $\mathsf{Hybrid}_{i}^{\mathsf{DelRes} = 1}$ when $b' = 1$ (the latter follows because the output of $\widehat{\cS}_{\mathsf{partial}}$ and $\cS_{\mathsf{partial}}$ conditioned on $c_i = \widehat{c}$ are statistically close, due to Watrous rewinding).
    Thus if $\mathsf{Hybrid}_{i-1}^{\mathsf{DelRes} = 1}$ and $\mathsf{Hybrid}_{i}^{\mathsf{DelRes} = 1}$ are not negligibly close in trace distance, $\mathsf{Red}$ breaks certified everlasting hiding of $\mathsf{Com}$, as desired. 
    Finally, we observe that $\mathsf{Hybrid}_{\lambda}$ is identical to the ideal experiment.
    
    \item \underline{Case 2: The Reveal Phase is executed.}
    $\mathsf{Red}$ acts as receiver in one session of $\Com$, interacting with an external challenger. 
    $\mathsf{Red}$ samples a uniformly random bit $d$ and sends it to the challenger.
    The challenger samples a uniformly random bit $b'$. If $b' = 0$, the challenger generates a commitment to $d$, and otherwise to $1-d$.
    
    $\mathsf{Red}$ internally follows the strategy in $\mathsf{Hybrid}_{i-1}$ in the Commit phase for sessions $1, \ldots, i-1$ and $i+1, \ldots, \lambda$. For the $i^{th}$ session $\mathsf{Red}$ runs the strategy $\cS_{\mathsf{partial}}$ from \proref{fig:eqsimulator}, with the following exception. For bits $\widehat{c}, \widehat{b}$ sampled uniformly at random, it sets the commitment in sub-session $(2\widehat{c} + 1 + \widehat{b})$ as the external commitment, and generates the commitment in sub-session $(2\widehat{c} + 1 + (1-\widehat{b}))$ as a commitment to $d$. It sets commitments in the remaining two sessions according to the strategy in $\mathsf{Hybrid}_{i - 1}$.
    If $\cR^*_\Com$'s challenge $c_i = \widehat{c}$, $\mathsf{Red}$ continues the experiment, otherwise it aborts.
    $\mathsf{Red}$ continues to follow the strategy in $\mathsf{Hybrid}_{i-1}$, except setting $b_i = b \oplus d$.
    Note that the challenge commitment is not opened in the Commit phase.

    In the Reveal phase, $\mathsf{Red}$ behaves identically to $\mathsf{Hybrid}_{i-1}$ in sessions $(1, \ldots, i-1, i+1, \ldots, \lambda)$, and for session $i$ it runs the Reveal phase of the commitment in sub-session $(2\widehat{c} + 1 + (1-\widehat{b}))$. 
    
    In the Delete phase, $\mathsf{Red}$ again follows the strategy of $\mathsf{Hybrid}_{i-1}$ except that it executes the Delete phase for the unopened commitments in the $i^{th}$ session by forwarding messages between $\cR^*_\Del$ and the external challenger. Thus, if $\mathsf{Hybrid}_{i-1}^{\mathsf{DelRes} = 1}$ and $\mathsf{Hybrid}_{i}^{\mathsf{DelRes} = 1}$ are not negligibly close in trace distance, $\mathsf{Red}$ breaks certified everlasting hiding of $\mathsf{Com}$, as desired. 
    Finally, we observe that $\mathsf{Hybrid}_{\lambda}$ is identical to the ideal experiment.
\end{itemize}
This completes the proof.
\end{proof}


\subsection{Secure computation}\label{subsec:secure-computation}

In this section, we show that, following compilers in previous work, ideal commitments with EST imply oblivious transfer with EST and thus two-party computation of arbitrary functionalities with EST. Since prior compilers in the commitment hybrid model actually make use of ``commitments with selective opening'', we will first discuss this primitive, then describe a simple (deletion-composable) protocol that securely realizes commitments with selective opening. Next, we will invoke prior results~\cite{EC:GLSV21} that together with our composition theorem imply secure two-party computation with EST.

Finally, we define the notion of multi-party computation with EST, and again show that it follows from ideal commitments with EST.

\subsubsection{Two-party computation}\label{subsec:two-party-computation}

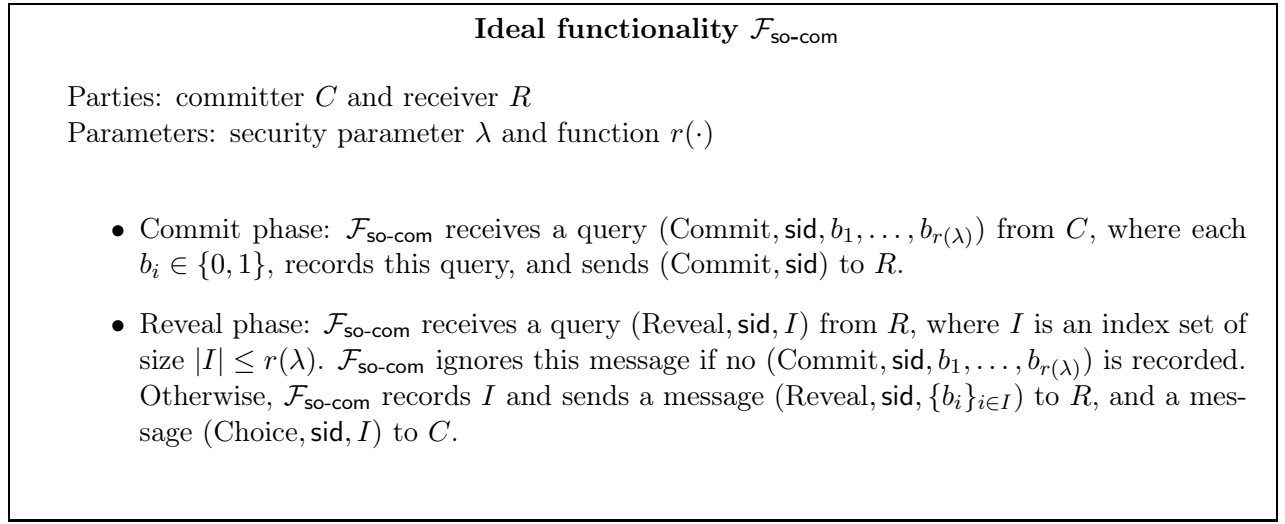
\protocol{Ideal functionality $\cF_{\socom}$}{Specification of the bit commitment with selective opening ideal functionality.}{fig:socom}{
Parties: committer $C$ and receiver $R$\\
Parameters: security parameter $\secp$ and function $r(\cdot)$\\
\begin{itemize}
    \item Commit phase: $\cF_{\socom}$ receives a query $(\text{Commit},\mathsf{sid},b_1,\dots,b_{r(\secp)})$ from $C$, where each $b_i \in \{0,1\}$, records this query, and sends $(\text{Commit},\mathsf{sid})$ to $R$.
    \item Reveal phase: $\cF_{\socom}$ receives a query $(\text{Reveal},\mathsf{sid},I)$ from $R$, where $I$ is an index set of size $|I| \leq r(\secp)$. $\cF_{\socom}$ ignores this message if no $(\text{Commit},\mathsf{sid},b_1,\dots,b_{r(\secp)})$ is recorded. Otherwise, $\cF_{\socom}$ records $I$ and sends a message $(\text{Reveal},\mathsf{sid},\{b_i\}_{i \in I})$ to $R$, and a message $(\text{Choice},\mathsf{sid},I)$ to $C$.
\end{itemize}
\noindent}

We define the ``commitment with selective opening'' ideal functionality $\cF_{\socom}$ in \proref{fig:socom}, and we describe a simple (deletion-composable) protocol $\Pi^{\cF_\com^\Del}$ that statistically securely realizes $\cF_{\socom}^\Del$.

\begin{itemize}
    \item The committer, with input $(b_1,\dots,b_{r(\secp)})$, sequentially sends $(\mathsf{Commit},i,b_i)$ to $\cF_\Com^\Del$ for $i \in [r(\secp)]$.
    \item The receiver, with input $I$, sends $I$ to the committer.
    \item The committer sequentially sends $(\mathsf{Reveal},i)$ for $i \in I$ to $\cF_\com^\Del$.
    \item The receiver obtains output $\{(\mathsf{Reveal},i,b_i)\}_{i \in I}$ from $\cF_\com^\Del$.
    \item The parties perform the delete phase as follows.
    \begin{itemize}
    \item If the committer is instructed to request a deletion, it sends $\{(\mathsf{DelRequest},i)\}_{i \in [r(\secp)]}$ to $\cF_\com^\Del$, which are forwarded to the receiver. 
    \item If the receiver obtains \emph{any} $(\mathsf{DelRequest},i)$ for $i \in [r(\secp)]$, it sets its output $\mathsf{DelReq} = 1$.
    \item For each $(\mathsf{DelRequest},i)$ obtained by the receiver, it sends $(\mathsf{DelResponse},i)$ to $\cF_\com^\Del$, which are forwarded to the committer. 
    \item If the committer obtains \emph{all} $\{(\mathsf{DelResponse},i)\}_{i \in [r(\secp)]}$, it sets its output $\DelRes = 1$.
    \end{itemize}
\end{itemize}

It is clear by definition that the above protocol statistically securely realizes $\cF_{\socom}^\Del$. Thus, by combining \cref{thm:CHSEE}, \cref{thm:one-sided}, and \cref{thm:ideal}, which together show that there exists a protocol that realizes $\cF_\com^\Del$ with EST assuming computationally-hiding statistically-binding commitments, and \cref{thm:composition}, which is our composition theorem, we obtain the following theorem.

\begin{theorem}\label{thm:socom}
There exists a protocol that securely realizes $\cF_{\socom}^\Del$ with EST that makes black-box use of a computationally-hiding statistically-binding commitment (\cref{def:comp-hiding} and \cref{def:SB}).
\end{theorem}

Note that it was necessary that our composition theorem handled reactive functionalities in order to establish this claim. Moreover, it is crucial in the definition of a reactive functionality with a deletion phase that we allow the deletion phase to be run after any phase of the reactive functionality. Indeed, in the above construction, some underlying commitments are not revealed but they still must be deleted. 

Finally, it was shown in \cite{EC:GLSV21} (building on the work of \cite{FOCS:CreKil88,C:DFLSS09}, among others) that using quantum communication, it is possible to statistically realize the primitive of \emph{oblivious transfer} in the $\cF_{\socom}$-hybrid model. Moreover, the work of \cite{STOC:Kilian88} showed how to statistically realize \emph{arbitrary two-party computation} in the oblivious transfer hybrid model. Since it is straightfoward to make these protocols deletion-composable with a delete phase at the end, we have the following corollary.

\begin{corollary}
Secure two-party computation of any polynomial-time functionality with Everlasting Security Transfer (\cref{def:security-with-EST}) exists, assuming only black-box use of a computationally-hiding statistically-binding commitment (\cref{def:comp-hiding} and \cref{def:SB}).
\end{corollary}

\subsubsection{Multi-party computation}

In order to define and construct multi-party computation with EST, we first have to specify a multi-party version of the Delete phase, which is described in \proref{fig:multidelPhase}.

\protocol{Muti-party Deletion phase}{A specification of a generic multi-party deletion phase that can be added to any multi-party ideal functionality $\cF$.}{fig:multidelPhase}{
Parties: $\{P_i\}_{i \in [n]}$
\begin{itemize}
    \item Receive a sequence of queries $(\mathsf{DelRequest},\mathsf{sid},i,j)$ which indicate that party $i$ is requesting party $j$ to delete their data. For each such query received, record it and send $(\mathsf{DelRequest},\mathsf{sid},i,j)$ to party $j$.
    \item Receive a sequence of queries $(\mathsf{DelResponse},\mathsf{sid},i,j)$ which indicate that party $i$ has deleted party $j$'s data. For each such query received, if there does not exist a recorded $(\mathsf{DelRequest},\mathsf{sid},j,i)$, then ignore the query. Otherwise, send $(\mathsf{DelResponse},\mathsf{sid},i,j)$ to party $j$.
\end{itemize}
}

To define security, we first note that it is straightforward to extend the discussion on the ``real-ideal paradigm'' from \cref{subsec:2PC-defs} to handle multi-party protocols where the adversary may corrupt any subset $M \subset [n]$ of $n$ parties. One can similarly generalize the definitions of computational and statistical secure realization (\cref{def:securerealization} and \cref{def:statsecurerealization}) to apply to multi-party protocols. Finally, we note that the multi-party Deletion phase introduced above adds $2(n-1)$ bits to each honest party $i$'s output, which we denote by $\{\DelReq_{j \to i},\DelRes_{j \to i}\}_{j \in [n] \setminus \{i\}}$, where each $\DelReq_{j \to i}$ indicates whether party $j$ requested that party $i$ delete its data, and each $\DelRes_{j \to i}$ indicates whether party $j$ deleted party $i$'s data. Now, we can generalize the notion of secure realization with EST (\cref{def:security-with-EST}) to the multi-party setting.

\begin{definition}[Secure realization with Everlasting Security Transfer: Multi-party protocols]\label{def:EST-multiparty}
A protocol $\Pi_\cF$ securely realizes the $\ell$-phase $n$-party functionality $\cF$ with EST if $\Pi_\cF$ computationally securely realizes $\cF^\Del$ (\cref{def:securerealization}) and the following holds. For every QPT adversary $\{\cA_\secp \coloneqq (\cA_{\secp,1},\dots,\cA_{\secp,\ell})\}_{\secp \in \bbN}$ corrupting a subset of parties $M \subset [n]$, there exists a QPT simulator $\{\cS_\secp \coloneqq (\cS_{\secp,1},\dots,\cS_{\secp,\ell})\}_{\secp \in \bbN}$ such that for any QPT environment $\{\cZ_{\secp} \coloneqq (\cZ_{\secp,1},\dots,\cZ_{\secp,\ell})\}_{\secp \in \bbN}$, and polynomial-size family of advice $\{\ket{\psi_\secp}\}_{\secp \in \bbN}$, \[\TD\left(\Pi_{\cF^\Del}^{\mathsf{Del}_M = 1}[\cA_\secp,\cZ_\secp,\ket{\psi_\secp}], \widetilde{\Pi}_{\cF^\Del}^{\mathsf{Del}_M = 1}[\cS_\secp,\cZ_\secp,\ket{\psi_\secp}]\right) = \negl(\secp),\]

where $\Pi^{\mathsf{Del}_M = 1}_{\cF^\Del}[\cA_\secp,\cZ_\secp,\ket{\psi_\secp}]$ is defined to be equal to $\Pi_{\cF^\Del}[\cA_\secp,\cZ_\secp,\ket{\psi_\secp}]$ if the bit $\mathsf{Del}_M = 1$ and defined to be $\bot$ otherwise, and likewise for $ \widetilde{\Pi}_{\cF^\Del}^{\mathsf{Del}_M = 1}[\cS_\secp,\cZ_\secp,\ket{\psi_\secp}]$. The bit $\mathsf{Del}_M$ is computed based on the honest party outputs, and is set to 1 if and only if for all $i \in [n] \setminus M$ and $j \in M$, $\DelReq_{j \to i} = 0$ and $\DelRes_{j \to i} = 1$.

\end{definition}

Note that we here we did not include a designated party (or parties) against whom statistical security should hold by default (as in \cref{def:security-with-EST}), but in principle one could define security in this manner. The definition as written in the multi-party case captures a type of \emph{dynamic} statistical security property, where after the completion of the protocol, any arbitrary subset of parties can comply with a deletion request and certifiably remove information about the other party inputs from their view.

Finally, we can prove the following corollary of \cref{thm:socom}.

\begin{corollary}
Secure multi-party computation of any polynomial-time functionality with Everlasting Security Transfer (\cref{def:EST-multiparty}) exists, assuming only black-box use of a computationally-hiding statistically-binding commitment (\cref{def:comp-hiding} and \cref{def:SB}).
\end{corollary}

\begin{proof}
It was shown by \cite{C:CreVanTap95} that multi-party computation of any polynomial-time functionality can be statistically realized in the oblivious transfer hybrid model, where each pair of parties has access to an ideal oblivious transfer functionality. The stand-alone composition theorem (\cref{thm:composition}) shows that this is also true in the quantum setting, so it remains to argue that the resulting multi-party protocol can be made to satisfy security with EST, assuming that the underlying oblivious transfers do. This only requires extending the notion of deletion-composability to this setting, which can be achieved with the following deletion phase.
\begin{itemize}
    \item If party $i$ is instructed to issue a deletion request to party $j$, they issue deletion requests for \emph{all} oblivious transfers that occurred between party $i$ and $j$ that were not already statistically secure against $j$. 
    \item If party $i$ obtains a deletion request from \emph{any} one of the oblivious transfers between party $i$ and party $j$, they output $\DelReq_{j \to i} = 1$.
    \item For each deletion request obtained by party $i$ from party $j$, party $i$ is instructed to send a deletion response to party $j$. 
    \item If party $i$ obtains a deletion response from party $j$ for \emph{all} oblivious transfers between party $i$ and party $j$ that were not already statistically secure against $j$, they output $\DelRes_{j \to i} = 1$.
\end{itemize}
This completes the proof.
\end{proof}

\fi

\ifsubmission
\else
\section*{Acknowledgments}
We thank Bhaskar Roberts and Alex Poremba for comments on an earlier draft, and for noting that quantum fully-homomorphic encryption is not necessary for our FHE with certified deletion scheme, classical fully-homomorphic encryption suffices.

D.K. was supported in part by DARPA and NSF QIS award 2112890.
This material is based on work supported by DARPA under Contract No. 
HR001120C0024. Any opinions, findings and conclusions or recommendations expressed in
this material are those of the author(s) and do not necessarily reflect the views of the United States
Government or DARPA.
\fi

\ifsubmission
\else
\bibliographystyle{alpha}
\fi

\addcontentsline{toc}{section}{References}
\bibliography{abbrev3,custom,crypto}

\appendix
\ifsubmission\else\fi
\ifsubmission\else\fi
\section{Relation with~\cite{10.1007/978-3-030-92062-3_21}'s definitions}
\label{sec:implication}
In this section, we prove that our definitions of certified deletion security for PKE and ABE imply prior definitions~\cite{10.1007/978-3-030-92062-3_21}. First, we reproduce the definitions in~\cite{10.1007/978-3-030-92062-3_21}, albeit following our notational conventions, for the settings of public-key encryption and attribute-based encryption below.

\begin{definition}[Certified deletion security for PKE in~\cite{10.1007/978-3-030-92062-3_21}]
\label{def:CD-security-prior}
$\CD$-$\PKE = (\Gen,\Enc,\Dec,\Del,\Ver)$ satisfies \emph{certified deletion security} if for any non-uniform QPT adversary $\cA = \{\cA_\secp,\ket{\psi}_\secp\}_{\secp \in \bbN}$, it holds that 
\[\bigg|\Pr\left[\mathsf{C}'\text{-}\EXP_\secp^{\cA}(0) = 1\right] - \Pr\left[\mathsf{C}'\text{-}\EXP_\secp^{\cA}(1) = 1\right]\bigg| = \negl(\secp),\]
where the experiment $\mathsf{C}'\text{-}\EXP_\secp^{\cA}(b)$ is defined as follows.
\begin{itemize}
    \item Sample $(\pk,\sk) \gets \Gen(1^\secp)$ and $(\ct,\vk) \gets \Enc(\pk,b)$.
    \item Initialize $\cA_\secp(\ket{\psi_\secp})$ with $\pk$ and $\ct$.
    \item Parse $\cA_\secp$'s output as a deletion certificate $\cert$ and a residual state on register $\sA'$.
    \item If $\Ver(\vk,\cert) = \top$, set $\mathsf{ret} = \sk$, otherwise set $\mathsf{ret} = \bot$.
    \item Output $\cA_\secp \left(\sA', \mathsf{ret} \right)$.
\end{itemize}
\end{definition}

\begin{definition}[Certified deletion security for ABE in~\cite{10.1007/978-3-030-92062-3_21}]
\label{def:CD-sec-ABE-prior}
$\CD$-$\ABE = (\Gen,\KG,\Enc,\allowbreak\Dec,\Del,\Ver)$ satisfies \emph{certified deletion security} if for any non-uniform QPT adversary $\cA = \{\cA_\secp,\ket{\psi}_\secp\}_{\secp \in \bbN}$, it holds that 
\[\bigg|\Pr\left[\mathsf{C}'\text{-}\EXP_\secp^{\cA}(0) = 1\right] - \Pr\left[\mathsf{C}'\text{-}\EXP_\secp^{\cA}(1) = 1\right]\bigg| = \negl(\secp),\] 
where the experiment $\mathsf{C}'\text{-}\EXP_\secp^{\cA}(b)$ is defined as follows.
\begin{itemize}
    \item Sample $(\pk,\msk) \gets \Gen(1^\secp)$
    and 
    initialize
    $\cA_\secp(\ket{\psi_\secp})$ with $\pk$.
    \item Set $i = 1$.
    \item If $\cA_\secp$ outputs a key query $P_i$, return $\sk_{P_i} \leftarrow \KG(\msk, P_i)$ to $\cA_\secp$ and set $i = i + 1$. This process can be repeated polynomially many times.
    \item If $\cA_\secp$ outputs an attribute $X^*$ and a pair of messages $(m_0, m_1)$ where $P_i(X^*) = 0$ for all predicates $P_i$ queried so far, then compute $(\vk,\ct) = \Enc(\pk, X^*, m_b)$ and return $\ct$ to $\cA_\secp$, else exit and output $\bot$.
    \item If $\cA_\secp$ outputs a key query $P_i$ such that $P_i(X^*) = 0$, return $\sk_{P_i} \leftarrow \KG(\msk, P_i)$ to $\cA_\secp$ (otherwise return $\bot$) and set $i = i + 1$. This process can be repeated polynomially many times.
    \item Parse $\cA_\secp$'s output as a deletion certificate $\cert$ and a residual state on register $\sA'$.
    \item If $\Ver(\vk,\cert) = \top$ set $\mathsf{ret} = \msk$, and otherwise set $\mathsf{ret} = \bot$. Send $\mathsf{ret}$ to $\cA_\secp$.
    \item Again, upto polynomially many times, $\cA_\secp$ sends key queries $P_i$. For each $i$, if $P_i(X^*) = 0$, return $\sk_{P_i} \leftarrow \KG(\msk, P_i)$ to $\cA_\secp$ (otherwise return $\bot$) and set $i = i + 1$.
    Finally, $\cA_\secp$ generates an output bit, which is set to be the output of $\mathsf{C}'\text{-}\EXP_\secp^{\cA}(b)$.
\end{itemize}
\end{definition}

\begin{claim}
Any PKE scheme satisfying Definition \ref{def:CD-security} also satisfies Definition \ref{def:CD-security-prior}.
\end{claim}
\begin{proof}
Suppose the claim is not true. Then there exists an adversary $\cA$ and polynomial $p(\cdot)$ such that with respect to the notation in Definition \ref{def:CD-security-prior},
\[\bigg|\Pr\left[\mathsf{C}'\text{-}\EXP_\secp^{\cA}(0) = 1\right] - \Pr\left[\mathsf{C}'\text{-}\EXP_\secp^{\cA}(1) = 1\right]\bigg| = \frac{1}{p(\secp)},\] 
and yet for every adversary $\cB$, with respect to the notation in Definition \ref{def:CD-security},
\begin{equation}
\label{eq:everlasting}
\TD\left(\EVEXP_\secp^{\cB}(0),\EVEXP_\secp^{\cB}(1)\right) = \negl(\secp),
\end{equation}
and
\[\bigg|\Pr\left[\CEXP_\secp^{\cB}(0) = 1\right] - \Pr\left[\CEXP_\secp^{\cB}(1) = 1\right]\bigg| = \negl(\secp),\] 
Now we consider the following (efficient) reduction $\cR = \{\cR_\secp\}_{\secp \in \bbN}$ that acts as an adversary in the experiment $\CEXP$.
$\cR_\secp$ passes $\pk$ and $\ct$ to $\cA_\secp$, then passes the deletion certificate output by $\cA_\secp$ to its challenger and saves its residual state on register $\sA'$. $\cR_\secp$ then obtains a verification outcome in $\{ \bot, \top \}$ from the challenger. 
If the outcome is $\top$, $\cR$ aborts, and otherwise $\cR_\secp$ outputs $\cA_\secp(\sA',\bot)$. 

By equation~(\ref{eq:everlasting}), when the outcome is $\top$, the resulting state on register $\sA'$ in $\mathsf{C}'\text{-}\EXP_\secp^{\cA}(b)$ is statistically independent of $b$ (and in particular, since the distribution over $\sk$ is fixed by $\pk$ and is otherwise independent of $\ct$, the state is also statistically independent given $\sk$). Thus, except for a negligible loss, any advantage of $\cA_\secp$ can only manifest in the case when the output is $\bot$, which implies that
\[\bigg|\Pr\left[\CEXP_\secp^{\cR}(0) = 1\right] - \Pr\left[\CEXP_\secp^{\cR}(1) = 1\right]\bigg| = \frac{1}{p(\secp)} - \negl(\secp) > \frac{1}{2p(\secp)},\]
a contradiction.
\end{proof}

\begin{claim}
Any ABE scheme satisfying Definition \ref{def:CD-sec-ABE} also satisfies Definition \ref{def:CD-sec-ABE-prior}.
\end{claim}
\begin{proof}
Suppose the claim is not true. Then there exists an adversary $\cA$ and polynomial $p(\cdot)$ such that with respect to the notation in Definition \ref{def:CD-sec-ABE-prior},
\[\bigg|\Pr\left[\mathsf{C}'\text{-}\EXP_\secp^{\cA}(0) = 1\right] - \Pr\left[\mathsf{C}'\text{-}\EXP_\secp^{\cA}(1) = 1\right]\bigg| = \frac{1}{p(\secp)},\] 
and yet for every adversary $\cB$,with respect to the notation in Definition \ref{def:CD-sec-ABE},
\begin{equation}
\label{eq:everlasting2}
\TD\left(\EVEXP_\secp^{\cB}(0),\EVEXP_\secp^{\cB}(1)\right) = \negl(\secp),
\end{equation}
and
\[\bigg|\Pr\left[\CEXP_\secp^{\cB}(0) = 1\right] - \Pr\left[\CEXP_\secp^{\cB}(1) = 1\right]\bigg| = \negl(\secp),\] 
Now we consider the following (efficient) reduction $\cR = \{\cR_\secp\}_{\secp \in \bbN}$ that acts as an adversary in the experiment $\CEXP$.
$\cR_\secp$ passes $\pk$ to $\cA_\secp$, then forwards all key queries of $\cA_\secp$ to its challenger, and forwards challenger responses back to $\cA_\secp$. Furthermore, it forwards any attribute $X^*$ and pair of messages $(m_0, m_1)$ output by $\cA_\secp$ to its challenger.
Finally, it passes the deletion certificate output by $\cA_\secp$ to its challenger and saves its residual state on register $\sA'$. $\cR_\secp$ then obtains a verification outcome in $\{ \bot, \top \}$ from the challenger. 
If the outcome is $\top$, $\cR$ aborts, and otherwise $\cR_\secp$ outputs $\cA_\secp(\sA',\bot)$. 

By equation~(\ref{eq:everlasting2}), when the outcome is $\top$, the resulting state on register $\sA'$ in $\mathsf{C}'\text{-}\EXP_\secp^{\cA}(b)$ is statistically independent of $b$ (and in particular, since the distribution over $\msk$ is fixed by $\pk$ and is otherwise independent of $\ct$, the state is also statistically independent given $\msk$). Thus, except for a negligible loss, any advantage of $\cA_\secp$ can only manifest in the case when the output is $\bot$, which implies that
\[\bigg|\Pr\left[\CEXP_\secp^{\cR}(0) = 1\right] - \Pr\left[\CEXP_\secp^{\cR}(1) = 1\right]\bigg| = \frac{1}{p(\secp)} - \negl(\secp) > \frac{1}{2p(\secp)},\]
a contradiction.
\end{proof}

\end{document}